\documentclass{article}

\usepackage[utf8]{inputenc}
\usepackage[margin=1in]{geometry}
\usepackage{wrapfig}

\RequirePackage{amsthm,amsmath,amsfonts,amssymb}
\RequirePackage[authoryear]{natbib}
\RequirePackage[colorlinks,citecolor=blue,urlcolor=blue]{hyperref}
\RequirePackage{graphicx}

\RequirePackage{algorithm}
\RequirePackage{bm}
\RequirePackage{slashed}
\RequirePackage{subcaption}
\RequirePackage{tikz}
\RequirePackage{pgfplots}
\RequirePackage{tikz-3dplot}
\usetikzlibrary{arrows.meta, positioning}
\RequirePackage{color}
\RequirePackage{colortbl}
\RequirePackage{mathtools}

\theoremstyle{plain}

\newtheorem{theorem}{Theorem}[section]
\newtheorem{lemma}[theorem]{Lemma}
\newtheorem{remark}{Remark}
\newtheorem{proposition}{Proposition}
\newtheorem{corollary}{Corollary}
\newtheorem{assumption}{Assumption}
\newtheorem{notation}{Notation}
\theoremstyle{definition}
\newtheorem{definition}[theorem]{Definition}

\theoremstyle{remark}

\DeclareMathOperator{\hol}{hol}
\DeclareMathOperator{\nodes}{nodes}

\DeclareMathOperator{\Tr}{Tr}

\DeclareMathOperator{\rmi}{\mathrm{i}}

\DeclareMathOperator{\res}{\mathrm{res}}
\DeclareMathOperator{\bad}{\mathrm{bad}}

\DeclareMathOperator{\mult}{\mathrm{mult}}
\DeclareMathOperator{\randomsplit}{\mathrm{randomsplit}}
\DeclareMathOperator{\forgetbase}{\mathrm{forgetbase}}
\DeclareMathOperator{\forgetorder}{\mathrm{forgetorder}}
\DeclareMathOperator{\scp}{\mathrm{scp}}
\DeclareMathOperator{\spt}{\mathrm{split}}

\renewcommand{\tau}{N}

\def\bfm{\mathbf{m}}
\newcommand{\calD}{\mathcal{D}}
\newcommand{\calN}{\mathcal{N}}

\newcommand{\calA}{\mathcal{A}}

\newcommand{\calX}{\mathcal{X}}
\newcommand{\calF}{\mathcal{F}}
\newcommand{\calJ}{\mathcal{J}}

\newcommand{\var}{\mathrm{var}}
\newcommand{\rmE}{\mathrm{E}}
\newcommand{\calH}{\mathcal{H}}

\newcommand{\bv}{\mathbf{v}}

\newcommand{\be}{\mathbf{e}}
\newcommand{\bE}{\mathbf{E}}

\newcommand{\bu}{\mathbf{u}}
\newcommand{\bU}{\mathbf{U}}

\newcommand{\bb}{\mathbf{b}}
\newcommand{\bff}{\mathbf{f}}

\def\cyclepopping{\textsc{CyclePopping}}
\def\verbosecyclepopping{\textsc{VerboseCyclePopping}}

\newcommand{\E}{\mathbb{E}}
\newcommand{\I}{\mathbb{I}}

\newcommand{\calC}{\mathcal{C}}

\newcommand{\calS}{\mathcal{S}}

\newcommand{\calV}{\mathcal{V}}
\newcommand{\calE}{\mathcal{E}}

\newcommand{\calT}{\mathcal{T}}
\newcommand{\calU}{\mathcal{U}}

\newcommand{\BLoops}{\mathrm{BL}}
\newcommand{\Loops}{\mathrm{L}}
\newcommand{\cycles}{\mathrm{cycles}}

\newcommand{\Uone}{{\rm U}(1)}
\renewcommand{\Re}{\operatorname{Re}}

\title{Cycling in the forest\\ with Wilson's algorithm}
\author{M. Fanuel and R. Bardenet\\
Université de Lille, CNRS, Centrale Lille\\
UMR 9189 - CRIStAL, F-59000 Lille, France\\
\texttt{\{michael.fanuel, remi.bardenet\}@univ-lille.fr} 
\\}
\date{\today}

\begin{document}
\maketitle
\begin{abstract}
    We consider a probability measure on cycle-rooted spanning forests (CRSFs) introduced by Kenyon.
    CRSFs are spanning subgraphs, each connected component of which has a unique cycle; they generalize spanning trees.
    A generalization of Wilson's celebrated \cyclepopping{} algorithm for uniform spanning trees has been proposed for CRSFs, and several concise proofs have been given that the algorithm samples from Kenyon's distribution.
    In this survey, we flesh out all the details of such a proof of correctness, progressively generalizing a proof by Marchal for spanning trees. 
    This detailed proof has several interests. 
    First, it serves as a modern tutorial on Wilson's algorithm, suitable for applied probability and computer science audiences. 
    Compared to uniform spanning trees, the more sophisticated motivating application to CRSFs brings forth connections to recent research topics such as loop measures, partial rejection sampling, and heaps of cycles. 
    Second, the detailed proof \emph{à la} Marchal yields the law of the time complexity of the sampling algorithm, shedding light on practical situations where the algorithm is expected to run fast.
    
    
\end{abstract}
\tableofcontents

\section{Introduction \label{sec:intro}}

Random trees and forests in a connected graph have multiple applications. 
In statistical physics, spanning trees (STs) as in Figure~\ref{fig:toy_ST}
have an intimate connection with the large-scale properties of systems exihibiting critical phenomena, such as sandpile models \citep{cori2003sand,ruelle2021sandpile}.
In data science, uniform STs provide preconditioners for Laplacian systems \citep{Fung11,KyngSong,KKS22}, whereas determinantal spanning forests help estimate the regularised inverse graph Laplacian \citep{PABT2020}.
While applications historically started with uniform spanning trees, more sophisticated probability distributions have been investigated by probabilists, which are being transferred to applications. 
For instance, a probability distribution over cycle-rooted spanning forests parametrized by a $\Uone$-connection has been introduced by \cite{kenyon2011}, generalizing uniform STs.
\begin{figure}[b]
    \centering
    \hfill
      \begin{subfigure}[b]{0.25\linewidth}
          \begin{tikzpicture}
              \node (R11) at (0,0) {\textbullet};
  
              \node (R12) at (1,0) {\textbullet};
              \node (R13) at (2,0) {\textbullet};
              \node (R14) at (3,0) {\textbullet};
  
              \node (R21) at (0,-1) {\textbullet};
              \node (R22) at (1,-1) {\textbullet};
              \node (R23) at (2,-1) {\textbullet};
              \node (R24) at (3,-1) {\textbullet};
  
              \node (R31) at (0,-2) {\textbullet};
              \node (R32) at (1,-2) {\textbullet};
              \node (R33) at (2,-2) {\textbullet};
              \node (R34) at (3,-2) {\textbullet};
  
              \draw [thick, draw=gray, opacity=0.2] (R12.center)--(R13.center);
              \draw [thick, draw=gray, opacity=0.2] (R22.center)--(R23.center);
              \draw [thick, draw=gray, opacity=0.2] (R22.center)--(R32.center);
              \draw [thick, draw=gray, opacity=0.2] (R23.center)--(R33.center);
              \draw [thick, draw=gray, opacity=0.2] (R33.center)--(R34.center);
  
            \path [-, ultra thick, draw=gray] (R12.center) edge  (R13.center);
            \path [-, ultra thick, draw=gray] (R32.center) edge  (R33.center);

              \draw [thick, draw=gray, opacity=0.2] (R11.center)--(R12.center);
              \path [-, ultra thick, draw=gray] (R12.center) edge  (R22.center);
              \draw [thick, draw=gray, opacity=0.2] (R12.center)--(R22.center);
              \path [-, ultra thick, draw=gray] (R22.center) edge  (R21.center);
              \draw [thick, draw=gray, opacity=0.2] (R22.center)--(R21.center);
              \path [-, ultra thick, draw=gray] (R21.center) edge  (R11.center);
              \draw [thick, draw=gray, opacity=0.2] (R21.center)--(R11.center);
  
              \path [-,ultra thick, draw=gray] (R31.center) edge  (R21.center);
              \draw [thick, draw=gray, opacity=0.2] (R31.center)--(R21.center);
              \path [-,ultra thick, draw=gray] (R32.center) edge  (R31.center);
              \draw [thick, draw=gray, opacity=0.2] (R32.center)--(R31.center);
  
              \path [-,ultra thick, draw=gray] (R13.center) edge  (R14.center);            
              \draw [thick, draw=gray, opacity=0.2] (R13.center)--(R14.center);
              \draw [thick, draw=gray, opacity=0.2] (R14.center)--(R24.center);
              \path [-,ultra thick, draw=gray] (R24.center) edge  (R23.center);
              \draw [thick, draw=gray, opacity=0.2] (R24.center)--(R23.center);
              \path [-,ultra thick, draw=gray] (R23.center) edge  (R13.center);
              \draw [thick, draw=gray, opacity=0.2] (R23.center)--(R13.center);
  
              \path [-, ultra thick, draw=gray] (R34.center) edge  (R24.center);
              \draw [thick, draw=gray, opacity=0.2] (R34.center)--(R24.center);
  
              \draw [thick, draw=gray, opacity=0.2] (R33.center)--(R32.center);

              \node at (0,0) {\textbullet};
  
              \node at (1,0) {\textbullet};
              \node at (2,0) {\textbullet};
              \node at (3,0) {\textbullet};
  
              \node at (0,-1) {\textbullet};
              \node at (1,-1) {\textbullet};
              \node at (2,-1) {\textbullet};
              \node at (3,-1) {\textbullet};
  
              \node at (0,-2) {\textbullet};
              \node at (1,-2) {\textbullet};
              \node at (2,-2) {\textbullet};
              \node at (3,-2) {\textbullet};
          \end{tikzpicture}
          \caption{\label{fig:toy_ST}}
      \end{subfigure}             
      \hfill
      \begin{subfigure}[b]{0.25\linewidth}
          \begin{tikzpicture}
              \node (R11) at (0,0) {\textbullet};
              \node (R11) at (0,0) {\textbullet};
  
              \node (R12) at (1,0) {\textbullet};
              \node (R13) at (2,0) {\textbullet};
              \node (R14) at (3,0) {\textbullet};
  
              \node (R21) at (0,-1) {\textbullet};
              \node (R22) at (1,-1) {\textbullet};
              \node (R23) at (2,-1) {\textbullet};
              \node (R24) at (3,-1) {\textbullet};
  
              \node (R31) at (0,-2) {\textbullet};
              \node (R32) at (1,-2) {\textbullet};
              \node (R33) at (2,-2) {\textbullet};
              \node (R34) at (3,-2) {\textbullet};
  
              \draw [thick, draw=gray, opacity=0.2] (R12.center)--(R13.center);
              \draw [thick, draw=gray, opacity=0.2] (R22.center)--(R23.center);
              \draw [thick, draw=gray, opacity=0.2] (R22.center)--(R32.center);
              \draw [thick, draw=gray, opacity=0.2] (R23.center)--(R33.center);
              \draw [thick, draw=gray, opacity=0.2] (R33.center)--(R34.center);
  
              \path [-, ultra thick, draw=gray] (R11.center) edge  (R12.center);
              \draw [thick, draw=gray, opacity=0.2] (R11.center)--(R12.center);
              \path [-, ultra thick, draw=gray] (R12.center) edge  (R22.center);
              \draw [thick, draw=gray, opacity=0.2] (R12.center)--(R22.center);
              \path [-, ultra thick, draw=gray] (R22.center) edge  (R21.center);
              \draw [thick, draw=gray, opacity=0.2] (R22.center)--(R21.center);
              \path [-, ultra thick, draw=gray] (R21.center) edge  (R11.center);
              \draw [thick, draw=gray, opacity=0.2] (R21.center)--(R11.center);

              \path [-, ultra thick, draw=gray] (R33.center) edge  (R32.center);
              \draw [thick, draw=gray, opacity=0.2] (R33.center)--(R32.center);
              \path [-, ultra thick, draw=gray] (R31.center) edge  (R21.center);
              \draw [thick, draw=gray, opacity=0.2] (R31.center)--(R21.center);
              \path [-, ultra thick, draw=gray] (R32.center) edge  (R31.center);
              \draw [thick, draw=gray, opacity=0.2] (R32.center)--(R31.center);

              \path [-, ultra thick, draw=gray] (R13.center) edge  (R14.center);            
              \draw [thick, draw=gray, opacity=0.2] (R13.center)--(R14.center);
              \path [-, ultra thick, draw=gray] (R14.center) edge  (R24.center);
              \draw [thick, draw=gray, opacity=0.2] (R14.center)--(R24.center);
              \path [-, ultra thick, draw=gray] (R24.center) edge  (R23.center);
              \draw [thick, draw=gray, opacity=0.2] (R24.center)--(R23.center);
              \path [-, ultra thick, draw=gray] (R23.center) edge  (R13.center);
              \draw [thick, draw=gray, opacity=0.2] (R23.center)--(R13.center);
  
              \path [-, ultra thick, draw=gray] (R34.center) edge  (R24.center);
              \draw [thick, draw=gray, opacity=0.2] (R34.center)--(R24.center);
  
              \path [-, ultra thick, draw=gray] (R33.center) edge  (R32.center);
              \draw [thick, draw=gray, opacity=0.2] (R33.center)--(R32.center);

              \node (R11) at (0,0) {\textbullet};
  
              \node (R12) at (1,0) {\textbullet};
              \node (R13) at (2,0) {\textbullet};
              \node (R14) at (3,0) {\textbullet};
  
              \node (R21) at (0,-1) {\textbullet};
              \node (R22) at (1,-1) {\textbullet};
              \node (R23) at (2,-1) {\textbullet};
              \node (R24) at (3,-1) {\textbullet};
  
              \node (R31) at (0,-2) {\textbullet};
              \node (R32) at (1,-2) {\textbullet};
              \node (R33) at (2,-2) {\textbullet};
              \node (R34) at (3,-2) {\textbullet};
          \end{tikzpicture}
          \caption{\label{fig:toy_CRSF}}
      \end{subfigure}
      \hspace{1.5cm}
  
      \caption{A spanning tree (a) and a cycle-rooted spanning forest (b). \label{fig:toy_examples}}
  \end{figure}
A cycle-rooted spanning forest (CRSF) is a spanning subgraph, in which each connected component has exactly one cycle; see Figure~\ref{fig:toy_CRSF}.
A $\Uone$-connection over a graph is a map that associates to each oriented edge $xy$ a unit-modulus complex number $\phi_{xy}$ such that $\phi_{yx} = \phi_{xy}^*$.
In statistical physics, the connection models a magnetic field, while in data science, it models a relative rotation; cycle-rooted spanning forests drawn from the distribution introduced by \cite{kenyon2011} lead again to preconditioners for a class of linear systems called \emph{magnetic Laplacian}, with applications to semi-supervised machine learning or angular synchronization \citep{FanBar22}.

The popularity of random trees and forests in applications rests on the availability of efficient sampling algorithms for a few key distributions.
The seminal algorithm is \cyclepopping{} by \cite{Wilson96}, a loop-erased random walk from which one can obtain a uniform spanning tree.
Many proofs of the correctness of Wilson's algorithm have been given, from the concise proof in the original paper \citep{Wilson96} to the detailed analysis of Wilson's loop-erased random walk by \cite{Marchal99}, as well as through the so-called \emph{Diaconis-Fulton} or \emph{stacks-of-cards} representation, see e.g. \citep[Section 4.1]{lyons_peres_2017}, or more recently as an instance of partial rejection sampling \citep*{guo2019uniform,jerrum2021fundamentals}.
Among these proofs, the analysis of \cite{Marchal99} stands out as the easiest to follow step by step, and yields the law of the running time of the algorithm as an immediate corollary.
In comparison, other proofs are typically more elegant, but invariably include a high-level statement that is not so easily checked by a non-expert mechanical reader. 

In this survey, we propose to switch the classical stress put on uniform spanning trees to more general cycle-rooted spanning forests, and we contribute a detailed proof \emph{à la} \cite{Marchal99} that Wilson's algorithm samples the law on cycle-rooted spanning forests introduced by \cite{kenyon2011}. 
This law is defined as follows.
Associate to each oriented cycle $c$ of a connected graph $G$ a weight
$
    \alpha(c) \in [0,1].
$
Each edge $ij$ further comes with a positive weight $w_{ij}$, and we let $p_e = w_{ij}/\deg(i)$ for $e=ij$, with $\deg(i)$ the sum of the edge weights incident to $i$.
Consider now a CRSF $\calU$\footnote{We use the letter $\calU$ to emphasize that the CRSF is \emph{unoriented}, whereas \emph{oriented} CRSFs will be considered later on in the paper.}.
For definiteness, assign a (conventional) orientation of its cycles.
The measure of a  CRSF $\calU$ is then defined by the product
\begin{equation}
    \label{eq:proba_CRSF_non_det}
    \mu_{\mathrm{CRSF},\alpha}(\calU)
    = \frac{1}{Z_\alpha} \prod_{e \text{ edge of } \calU} p_e \times
    \prod_{c \substack{\text{ cycle of } \calU }}
    \left(\alpha(c) + \alpha(\overline{c})\right),
\end{equation}
with $Z_\alpha >0$ a normalization factor, and where $\overline{c}$ is the cycle $c$ with flipped orientation.
Note that \eqref{eq:proba_CRSF_non_det} is actually independent of the chosen orientations for edges and cycles.
When a $\Uone$-connection is available, the cycle weight in \eqref{eq:proba_CRSF_non_det} is typically defined as 
\begin{equation}
    \label{eq:det_cycle_weight}
    \alpha(c) = 1-\cos \theta(c),
\end{equation}
where $\theta(c)$ is obtained by letting $\exp \rmi \theta(c) = \prod_{e\in c}\phi_e$, using a fixed orientation of $c$.
When $\alpha(c)$ has the specific form \eqref{eq:det_cycle_weight}, \cite{kenyon2011} showed that the edges in a $\calU$ sampled from \eqref{eq:proba_CRSF_non_det} form a determinantal point process; see \cite{HKPV06}.
\cite{KK2017} have introduced a variant of Wilson's loop-erased random walk that samples \eqref{eq:proba_CRSF_non_det}, where cycles encountered by the walk now have a probability \emph{not} to be popped. 
By analogy with the cases of spanning trees and forests, the algorithm is still called \cyclepopping{}.

This survey can serve as an advanced introduction to \cyclepopping{} for a probability or computer science audience. 
The focus on a non-trivial target distribution shows the degrees of freedom one can have to generalize uniform spanning trees, while preserving Wilson's sampling algorithm. 
On a pedagogical side, CSRFs further connect \cyclepopping{} to many modern notions in probability. 
For instance, a key role in the proof is played by a last-exist decomposition, a notion also used by \cite{LaLi10}. 
Links also naturally appear along the exposition to Poissonian loop ensembles \citep{lawler2004brownian,LeJan11,Sznitman2012,kassel2021covariant}, or combinatorial structures like heaps of cycles \citep{viennot2006heaps}.
These \emph{heaps} are crucial to understand in depth the partial rejection sampling approach to \cyclepopping{} \citep*{guo2019uniform,jerrum2021fundamentals}.
Besides its pedagogical merits, and similarly to the case of uniform spanning trees \citep{Marchal99}, the proof that we present provides for free the law of the running time of the algorithm, an important quantity for computational applications.

This paper is mostly a survey, in the sense that the correctness of \cyclepopping{} for sampling from the distribution \eqref{eq:proba_CRSF_non_det} over cycle-rooted spanning forests has been known since the \emph{stacks-of-cards} proof of \cite{KK2017}.
A sketch of direct proof \emph{à la Marchal} was actually proposed in \cite{Kassel15}, for the particular case \eqref{eq:determinantal_cycle_weights} of determinantal weights.
Our merit is to flesh out all the details of the latter sketch of proof, and generalize it to the target distribution initially proposed by \cite{kenyon2011}, without restriction to determinantal measures.
As a side product, we obtain for the first time the distribution of the running time
\begin{equation}
    T \stackrel{\text{(law)}}{=} n + \sum_{\substack{[\gamma]\in \calX}}|\gamma|,
    \label{eq:law_of_T_Poisson}
\end{equation}
where $n$ is the number of vertices of $G$, $|\gamma|$ is the length of the oriented loop $[\gamma]$, $\calX$ is the Poisson point process on loops of $G$ with intensity $m_\alpha$, and $m_\alpha([\gamma])$ is proportional both to the probability of the simple random walk following $[\gamma]$ and to the weight of any cycle in $[\gamma]$. 
Finally, an additional contribution is to provide an alternative expression for $T$ using a Poisson process of \emph{pyramids} of (popped) oriented cycles.
This approach describes the loop formed by popped oriented cycles as a \emph{heap of cycles with a unique maximal element}, a \emph{pyramid} in the sense of \citet{viennot2006heaps}.

\subsection{Organization}

The paper is organized as follows. 
To give context, we review in Section~\ref{s:trees_and_forests} known results for the time complexity of \cyclepopping{} when cycles are always popped, i.e., when sampling spanning trees and spanning forests.
A formalization of \cyclepopping{} for sampling CRSFs -- including the necessary variables to describe the probability of accepting cycles -- is given in Section~\ref{s:formal_def_cycle_popping}.
Armed with these tools, Section~\ref{sec:SamplingCRSFs} gives the proof of correctness and the law of the sampling time $T$ when the CRSF measure is determinantal, i.e. when the cycle weights in \eqref{eq:proba_CRSF_non_det} are of the form \eqref{eq:det_cycle_weight} and $\theta(c)\in [-\pi/2,\pi/2]$.
Thanks to the introduction of loop measures, Section~\ref{s:Poisson} gives an alternative derivation of the law of $T$ for generic cycle weights in \eqref{eq:proba_CRSF_non_det} satisfying $\alpha(c)\in [0,1]$, so that the measure is not necessarily determinantal.
We also further discuss how the law of $T$ can be described in terms of heaps of cycles popped by the algorithm.
In Section~\ref{sec:prs}, we reinterpret \cyclepopping{} as a Markov chain that progressively builds such a heap, culminating in an intuitive formula for the law of the number of cycles popped by the algorithm.
This reinterpretation is a restatement of the partial rejection sampling algorithm of \cite{guo2019uniform} in the case of CRSFs, highlighting the role of an implicit Markov chain.
Finally, empirical evaluations of the expectation and variance of $T$ are given in Section~\ref{sec:numerics} for a few random $\Uone$-connection graphs, in order to illustrate the results.

\subsection{Notations \label{sec:notations}}

For ease of reference, we gather here our main notations.

\paragraph*{Laplacian.} Consider a connected undirected graph $G$ with vertex set $\calV$ and edge set $\calE$, with $n = |\calV| > 1$ and $m = |\calE|$. 
Often, we implicitly identify $\calV$ with $\{1, \dots, n\}$.
When two vertices $x,y\in \calV$ are connected by an edge, we write $x\sim y$.
In what follows, we assume that $G$ has no multiple edge and no self-loop.
A $\Uone$-\emph{connection} endows each oriented edge $xy$  with a complex phase $\phi_{xy} = \exp(-\rmi \vartheta(xy))$ such that $\vartheta(yx) = -\vartheta(xy)\in [0,2\pi)$.
Let $\mathsf{\Phi}$ be the $n\times n$ matrix such that $(\mathsf{\Phi})_{xy} = \phi_{xy}$ for each $x \sim y$ and $(\mathsf{\Phi})_{xy} = 0$ otherwise.
Similarly, let $\mathsf{W}$ be the symmetric $n\times n$ matrix such that $(\mathsf{W})_{xy} = w_{xy} > 0$ for each $x \sim y$ and $(\mathsf{W})_{xy} = 0$ otherwise.
For $x \in \calV$, let $\deg(x) = \sum_{y: y\sim x} w_{xy}$.
The \emph{magnetic Laplacian} is the $n\times n$ matrix defined by 
\[
    \mathsf{\Delta} = \mathsf{D} - \mathsf{W} \odot \mathsf{\Phi}
\]
where $\mathsf{D}$ is the diagonal matrix such that $\mathsf{D}_{xx} = \deg(x) >0$ for all $x\in \calV$, and where $\mathsf{W}\odot \mathsf{\Phi}$ denotes the entrywise product.
In particular, writing $\mathsf{\Lambda} = \mathsf{D} - \mathsf{W}$ for the usual (combinatorial) Laplacian, we see that $\mathsf{\Delta}$ and $\mathsf{\Lambda}$ coincide when the phase is trivial, i.e. when $\vartheta\equiv 0$.
We refer to e.g.\ \citep{LiLo93,zbMATH01219775,AFST_2011_6_20_3_599_0,Berkolaiko,Fanuel2018MagneticEigenmapsVisualization} for references on the magnetic Laplacian and a motivation for the adjective \emph{magnetic}.

\paragraph*{Transition matrix.} Furthermore, to each oriented edge $xy$, we associate a transition probability $p_{xy}\triangleq \frac{w_{xy}}{\deg(x)}$, which will be used to define Markov chains on the graph.
As we did for the weight matrix $\mathsf{W}$, we also define an $n\times n$ transition matrix $\mathsf{P} = \mathsf{D}^{-1} \mathsf{W}$, so that $(\mathsf{P})_{xy} = p_{xy}$ for $x\sim y$ and $(\mathsf{P})_{xy} = 0$ otherwise.
Finally, for $\calS \subseteq \{1,\dots,\ell\}$ and $\mathsf{M}$ an $\ell\times \ell$ matrix, $\mathsf{M}_{\overline{\calS}}$ is the matrix obtained by \emph{removing} the rows and columns of $\mathsf{M}$ which are indexed by $\calS$.

\paragraph*{Loops.} In what follows, we denote an oriented path in $G$ by $(x_0, \dots, x_k)$, where consecutive nodes are linked by an edge. 
A \emph{based loop} $\gamma$ is an oriented path $\gamma = (x_0, \dots, x_k)$ in the graph $G$, for some integer $k\geq 2$, such that $x_k = x_0$. 
We then call $x_0$ the \emph{base node} of $\gamma$.
A based loop of the form $(x_0, x_1, x_0)$ is called a \emph{backtrack}.
Let $\bar{\gamma}$ be the based loop $\gamma$ with opposite orientation.
For convenience, we also consider based loops of zero length, consisting only of one (base) point.
The number of edges in a based loop $\gamma$ is denoted by $|\gamma|$. 
Two based loops $\gamma = (x_0,x_1, \dots, x_k)$ and $\gamma^\prime = (x_0, x^\prime_1, \dots, x^\prime_{k^\prime})$ based at $x_0$ can be concatenated to yield $\gamma \circ \gamma^\prime = (x_0,x_1, \dots, x_{k-1}, x_0, x^\prime_1, \dots, x^\prime_{k^\prime})$. 
Naturally, $\gamma^m$ denotes the $m$-th power of $\gamma$ w.r.t.\ concatenation. 
Contrary to \citep{LeJan11}, all based loops considered here are discrete.

Next, we define an equivalence relation between based loops by identifying $\gamma = (x_0, \dots, x_k)$ with any based loop {obtained by a $j$-shift of the form
\begin{equation}
    (x_0, \dots, x_{k-1}, x_0) \mapsto (x_j, x_{j+1}, \dots, x_{k-1},x_0, \dots, x_{j-1}, x_j)\label{eq:shift}
\end{equation}
for $0\leq j \leq k-1$; see \citet[Section 9.1]{LaLi10}.}
The corresponding equivalence class $[\gamma]$ is simply a based loop, of which we forget the base node: we call the equivalence class an \emph{unbased loop}, or sometimes simply a \emph{loop}.
By definition,  any representative $\gamma$ of $[\gamma]$ has the same number of edges $|\gamma|$, so that we write $\vert \gamma\vert$ for $\vert[\gamma]\vert$ without ambiguity.
In particular, we also consider the trivial loop of only one point as an unbased loop with zero edge.
Finally, the number of representatives in the class $[\gamma]$ is denoted by $\mathcal{N}_{[\gamma]}$; note that it is not necessarily equal to $|\gamma|$.

\paragraph*{Cycles.} A based loop $c$ is called an oriented \emph{cycle} if it is minimal, i.e., if it does not include another loop. 
In particular, an oriented cycle of length $2$ -- called a backtrack is also considered as an oriented cycle.
Denote by $n_c$ the number of oriented cycles of the graph.
A non-oriented cycle is simply an equivalence class of oriented cycles of length larger than $2$ under orientation flip.
\section{\cyclepopping{} for spanning trees and forests \label{s:trees_and_forests}}
In this section, we outline the results of \cite{Marchal99} on the law of the number of steps of \cyclepopping{} for spanning trees.
The case of spanning forests follows immediately.

A spanning tree (ST) of $G$ is a connected spanning subgraph of $G$ without cycle.
Identifying an ST with its edges, STs can be endowed with a determinantal probability measure \citep{Pemantle91}
\begin{equation}
    \mu_{\mathrm{ST}}(\calF) = \frac{1}{\det(\mathsf{\Lambda}_{\bar{r}})} \prod_{e\in \calF} w_e,
    \label{eq:ST_measure} 
\end{equation}
where $\mathsf{\Lambda} = \mathsf{D} - \mathsf{W}$ is the combinatorial Laplacian of the graph and $r\in \calV$ is any of its nodes.\footnote{
    It turns out that the normalization factor $\det\mathsf{\Lambda}_{\bar{r}}$ does not depend on $r$. This is a consequence of the fact that the rows of $\mathsf{\Lambda}$ sum to zero.}

A spanning forest (SFs) is a spanning subgraph without cycle.
In other words, an SF is a spanning graph, of which each connected component is a tree.
Let $q>0$. 
A notable determinantal measure \citep{AvGaud2018,pilavci2020} on SFs is
\begin{equation}
    \mu_{\mathrm{SF}}(\calF) = \frac{q^{\rho(\calF)}}{\det(\mathsf{\Lambda} + q \I)}  \prod_{e\in \calF} w_e,   
    \label{eq:SF_measure} 
\end{equation}
where $\rho(\calF) = n - | \calF |$ is the number of connected components  of $\calF$.
\cyclepopping{} refers to an algorithm to sample from \eqref{eq:ST_measure} and \eqref{eq:SF_measure}.

\subsection{\cyclepopping{} for spanning trees\label{s:CyclePopping_for_STs}}
A classical approach to sampling from \eqref{eq:ST_measure} is Wilson's original algorithm \citep{Wilson96}, which sequentially grows a rooted tree.
In Wilson's algorithm, an arbitrary node $r\in \calV$, called the \emph{root}, is selected beforehand, as well as an ordering of the remaining nodes.
We initialize the tree as containing only the root node.
Starting from the lowest-order node not already in the tree, a random walker performs a loop-erased random walk (LERW) with edge weights $w_e$, until reaching the tree.
Upon reaching the tree, the walker's trajectory is added to the tree, forming a new branch.
Next, the walker restarts the LERW from the lowest-order node in the complement of the tree, and the process continues until the tree is spanning.
Forgetting the root and the order of the nodes, we obtain a tree drawn from \eqref{eq:ST_measure}.
Since the classical procedure to draw a trajectory from a LERW involves erasing (``popping'') loops as they appear in a simple random walk, the resulting algorithm -- Wilson's algorithm followed by outputting only an unrooted tree -- is called \cyclepopping{}.
The complexity of the algorithm is the total number of steps made by the involved simple random walks, i.e.\ the number of calls to the Markov kernel of the simple random walk.

We start by quoting a fundamental result on the law of the total number of steps for sampling uniform spanning trees, which we shall progressively generalize.
\begin{proposition}[Proposition 1 in \citep{Marchal99}]
    \label{prop:law_of_T}
    Let $G=(V,E)$ be a connected graph without self-loops, and $r\in V$.
    Let $T_r$ be the random number of steps to complete Wilson's algorithm for sampling a spanning tree rooted at $r$.
    Let $\mathsf{P} = \mathsf{D}^{-1} \mathsf{W}$, and $\mathsf{P}_{\bar{r}}$ be the matrix obtained by deleting the $r$-th row and the $r$-th column of $\mathsf{P}$.
    \begin{itemize}
        \item[(i)] For all $t \in (0,1)$, we have
        $
            \E[t^{T_r}] = t^{n-1}\frac{\det(\I - \mathsf{P}_{\bar{r}})}{\det(\I - t\mathsf{P}_{\bar{r}})}.
        $
        \item[(ii)] It holds that
        $
            \E[T_r] = \Tr\left((\I - \mathsf{P}_{\bar{r}})^{-1}\right) = n - 1 + \Tr\left( \mathsf{P}_{\bar{r}}(\I - \mathsf{P}_{\bar{r}})^{-1}\right) .
        $
    \end{itemize}
\end{proposition}
As noted in \citet[Eq (8)]{Marchal99}, the expected number of steps $\E[T_r]$ to complete Wilson's algorithm is equal to the expectation of the commute time between a $\pi$-random $i$ and the root $r$, where  $\pi(i)= \deg(i)/\sum_{j}\deg(j)$; see also \citet[Thm 2]{Wilson96}. 
In particular, even if \cyclepopping{} forgets the root, the expected complexity does depend on $r$.
Actually, the diagonal element $\left((\I - \mathsf{P}_{\bar{r}})^{-1}\right)_{ii}$ is the inverse of the probability that a random walker starting from $i$ hits the root $r$ before coming back to $i$; see e.g.\ \citep[Eq (9.12) and (9.17)]{LePe2017}.
This suggests that, in order to minimize $\E[T_r]=\Tr\left((\I - \mathsf{P}_{\bar{r}})^{-1}\right)$, the root has to be chosen at a central position in the graph.

Since Proposition~\ref{prop:law_of_T} yields the cumulant generating function for $T_r$, we also note the following consequence.
\begin{corollary}\label{corol:variance_of_T}
    With the notations of Proposition~\ref{prop:law_of_T},
    $$
        \var[T_r] = \Tr\left(\mathsf{P}_{\bar{r}}(\I - \mathsf{P}_{\bar{r}})^{-1}\right) + \Tr\left(\mathsf{P}_{\bar{r}}(\I - \mathsf{P}_{\bar{r}})^{-1}\right)^2.
    $$ 
\end{corollary}
The results in Proposition~\ref{prop:law_of_T} and Corollary~\ref{corol:variance_of_T} indicate that the random number of steps $T_r$ only depends on the eigenvalues of the transition matrix $\mathsf{P}_{\bar{r}}$ of the walk absorbed at $r$; see Section~\ref{sec:intro}.
Intuitively, $T_r$ is likely to be small if the largest eigenvalue of $\mathsf{P}_{\bar{r}}$ is much smaller than one.

Note that elegant expressions for the expected number of popped cycles in popping algorithms are derived in \citep{GJ2021} by using the partial rejection sampling framework.
This is a different viewpoint related to combinatorial objects called heaps of cycles.
However, the law of the number of popped cycles is not explicitly determined in \citep{GJ2021}.

\subsection{\cyclepopping{} for spanning forests \label{sec:SF_case}
}
Marchal's result in Proposition~\ref{prop:law_of_T} can be directly extended to the distribution \eqref{eq:SF_measure} on spanning forests.
\cyclepopping{} for \eqref{eq:SF_measure} also corresponds to running Wilson's algorithm, but on an auxiliary graph $G_r$, defined as $G$ where each node is additionally connected to a new node $r$ with a constant edge weight $q>0$.
Then, a random ST $\calF$ rooted at $r$ is sampled using Wilson's algorithm on $G_r$.
Forgetting $r$ in $\calF$ and all edges that connects to $r$, the $ST$ in $G_r$ becomes an SF in $G$, with distribution \eqref{eq:SF_measure} \citep{AvGaud2018}.
\begin{proposition}\label{prop:mean_var_SF}
    Let $q>0$ and let $T$ denote the number of steps to sample an SF with Wilson's algorithm according to \eqref{eq:SF_measure}.
    We have
        $$
            \E[T] = \Tr\left(\mathsf{D}(\mathsf{\Lambda} + q \I)^{-1}\right) + q \Tr\left(\mathsf{\Lambda} + q \I\right)^{-1}
        $$
        and
        $$ 
            \var[T] =  \Tr\left(\mathsf{M}_q (\mathbb{I}- \mathsf{M}_q)^{-1}\right) + \Tr\left(\mathsf{M}_q (\mathbb{I}- \mathsf{M}_q)^{-1}\right)^2,
        $$
        with $\mathsf{M}_q = (\mathsf{D}+q\mathbb{I})^{-1}\mathsf{W}$.
\end{proposition}
\begin{proof}
     The combinatorial Laplacian $\mathsf{\Lambda}'$ of $G_r$ can be written in block form
    \[  \mathsf{\Lambda}' = 
        \begin{pmatrix}
            \mathsf{\Lambda} + q \I & -q \mathsf{1}\\
            -q \mathsf{1}^\top & n q
        \end{pmatrix} = \mathsf{D}' - \mathsf{W}',
    \]
    where
    \[  
        \mathsf{D}' =     
        \begin{pmatrix}
            \mathsf{D} + q \I & \mathsf{0}\\
            \mathsf{0}^\top & n q
        \end{pmatrix}, 
        \quad
        \mathsf{W}' = 
        \begin{pmatrix}
            \mathsf{W} & -q \mathsf{1}\\
        -q \mathsf{1}^\top & 0
        \end{pmatrix},
    \]
    and $\bm{1}$ is the all-ones vector.
    Upon noting that $\mathsf{D}'_{\bar{r}} = \mathsf{D} + q \I$ and $\mathsf{W}'_{\bar{r}} = \mathsf{W}$, the result follows from Proposition~\ref{prop:law_of_T} and Corollary~\ref{corol:variance_of_T}.
\end{proof}
Notice how the expectation of $T$ in Proposition~\ref{prop:mean_var_SF} includes two terms, one for the number of steps of the walker within $G$, and the other for the number of steps to the auxiliary root of $G_r$. 
As a side note, the latter is also the number of roots in the determinantal point process formed by the roots of this random SF; see \citep{AvGaud2018}. 
\section{\cyclepopping{} for cycle-rooted spanning forests}
\label{s:formal_def_cycle_popping}
To describe the algorithm, we first introduce a Markov chain and a few stopping times.

\subsection{The Markov chain}
\label{s:defining_the_markov_chain}
Let $(X_{n})_{n\geq 0}$ be a simple random walk (SRW) on the graph nodes.
For definiteness, we take the initial distribution $X_0$ to be uniform over the nodes, but we shall most often consider, for $x\in\mathcal{V}$, the law $\mathbb{Q}_x$ of the chain when $X_0\sim\delta_x$; we fix the transition probability to 
\begin{equation}
    \mathbb{Q}_x(X_{1} = y) = \frac{w_{xy}}{\deg(x)} \triangleq  p_{xy}, \label{eq:P}
\end{equation}
as announced in Section~\ref{sec:notations}.
In particular, for $k\geq 1$ and for all $x,y\in \calV$,
$
    \mathbb{Q}_x(X_{k} = y) = (\mathsf{P}^k)_{xy}.
$

To each oriented cycle $c$ in the graph, we also associate a sequence of i.i.d.\  Bernoulli random variables $(B_{c,n})_{n\geq 0}$, with constant success probability $\alpha(c)$. 
For a fixed enumeration $c_1, \dots, c_d$ of the oriented cycles of the graph, we consider the Markov chain\footnote{
    When completing this work, we realized that the description of a similar Markov process was given in the PhD thesis of \citet[Chapter 5]{Constantin23}, who studied infinite volume limits of determinantal and non-determinantal CRSF measures.
}
$$
    Z_n = (X_n, B_{c_1,n}, \dots, B_{c_d,n})_{n\geq 0}.
$$
We still write $\mathbb{Q}_x$ the law of $Z_n$ when $X_0\sim \delta_x$.

\subsection{Hitting and cycle times\label{s:hitting_cycle_time}}
We now describe how the SRW can be conditioned not to intersect a subset of nodes. 

Let $\calA\subset \calV$ be a subset of nodes and let $x\in \calV\setminus \calA$. 
It is customary to define the \emph{first hitting time} as the random variable
\[
    \tau_{\to \calA} = \min\{n\geq 0 \text{ such that } X_{n} \in \calA \}.
\]
It is a \emph{stopping time}, since the event $\{\tau_{\to \calA} = m\}$ only depends on $X_{0}$, $X_{1}$, \dots, $X_{m}$, and is thus in the $\sigma$-algebra $\sigma(Z_1, \dots, Z_m)$.
For later convenience, for $m\geq 0$, we further introduce the \emph{first hitting time after $m$ steps}
\[
    \tau^{(\text{after }m)}_{\to \calA} = \min\{n\geq m \text{ such that } X_{n} \in \calA \},
\]
which is also non-anticipating, and satisfies $\tau^{(\text{after }0)}_{\to \calA} = \tau_{\to \calA}$.

Intuitively, Bernoulli variables will correspond to cycles being accepted after they have been formed by the SRW. 
For $x\in\mathcal{V}$ and a fixed cycle $c$, we thus define 
\[
    \tau_{\circlearrowleft c} = \min\{n\geq 0 \text{ such that } c \subseteq (X_{0}, X_{1},  \dots, X_{n}) \text{ and }  B_{c,n} = 1 \},
\]
that is, the first time a cycle is formed and accepted.
Again, $\tau_{\circlearrowleft c}$ is a stopping time, since  $\{\tau_{\circlearrowleft c} = m\}$ only depends on $Z_{0},Z_{1}, \dots, Z_{m}$, through their components.
Naturally, the value of $\tau_{\circlearrowleft c}$ is at least equal to the length of $c$.

In the same way, we define the \emph{cycle time of $c$ after $m$ steps} by
\[
    \tau^{(\text{after }m)}_{\circlearrowleft c} = \min\{n\geq m \text{ such that } c \subseteq (X_{m}, X_{m+1},  \dots, X_{n}) \text{ and }  B_{c,n} = 1 \}.
\]
Lastly, the first time after $m$ steps at which either a cycle is accepted or the SRW hits $\calA$ is the stopping time\footnote{In this paper, we write for simplicity $a\wedge b = \min \{a,b\}$.}
\[
    \tau^{(\text{after }m)}_{\mathrm{stop},\calA} = \tau^{(\text{after }m)}_{\to \calA} \wedge  \min_{ \text{ cycle }c} \tau^{(\text{after }m)}_{\circlearrowleft c}.
\]
For simplicity, we also write $\tau_{\mathrm{stop},\calA}= \tau^{(\text{after }0)}_{\mathrm{stop},\calA}$.
For later use, we state a lemma whose proof is elementary.
\begin{lemma}\label{lem:stopping_time_after_m}
    The event $\{\tau^{(\text{after }m)}_{\mathrm{stop},\calA} = k\}$ is in the $\sigma$-algebra generated by $Z_{m},Z_{m+1}, \dots, Z_{m+k}$.
\end{lemma}

Finally, the algorithm will be simpler to explain if we define as many copies of $(X_n)$ as there are nodes in $G$, denoting by $X_n^x$ a chain with the same Markov kernel as $(X_n)$, but with initial distribution $\delta_x$, $x\in G$. 
Similarly, we define independent streams i.i.d.\ Bernoullis $(B^x_{c,n})_{n\geq 0}$, one stream for each $x\in\calV$ and each cycle $c$ in $G$.
Again, for $x\in\calV$, the stopping times 
$$
    \tau^{x, (\text{after }m)}_{\rightarrow\calA}, \tau^{x, (\text{after }m)}_{\circlearrowleft c}, \text{ and } \tau^{x, (\text{after }m)}_{\mathrm{stop},\calA}
$$    
are defined by replacing $X_n$ by $X_n^x$ and $B_{c,n}$ by $B_{c,n}^x$ in the corresponding definitions.
We denote by $\mathbb P$ the joint law of all chains and all Bernoulli streams.
We are now ready to formalize the sampling algorithm.

\subsection{Formalization of \cyclepopping{} to sample CRSFs \label{sec:formalization}}
\cyclepopping{}$(G,\mathsf{W},\alpha,\iota)$ is a modified version of Wilson's algorithm, which takes the graph $G${, the edge weights $\mathsf{W}$}, the cycle weights $\alpha$ and a bijection $\iota:\{0,\dots,N-1\}\rightarrow \mathcal{V}$ as arguments. 
It defines a growing set of branches and cycle-rooted branches.
Recall that a \emph{branch} (or \emph{path graph}) is a tree with exactly two nodes of  degree $1$, whereas all the other nodes have degree $2$.
Similarly, a \emph{cycle-rooted branch} (or \emph{lasso}) is either a cycle or a cycle to which a branch has been connected at one of its end points.

The algorithm defines an increasing sequence of graphs $\calT_1\subset\calT_2\subset \dots$, with node set $\calV_\ell$ at iteration $\ell$, until we reach the first $\ell$ such that $\calV_\ell=\calV$, where the algorithm terminates.
The pseudocode of \cyclepopping{}$(G,\mathsf{W},\alpha,\iota)$ is given in Algorithm~\ref{a:cyclepopping}.

\begin{algorithm}
    \begin{itemize}
            \item Initialize $\ell=0$, $\calT_0 = (\calV_0,\calE_0) = (\emptyset, \emptyset)$.
            \item While $\calV\setminus \calV_\ell$ is not empty,
            \begin{enumerate}
            \item Set $x_0 = \iota(k_\star)$, with $k_\star = \min\{k  \text{ such that }  \iota(k) \notin \calV_\ell\}$.
            \item Run $(X_{n}^{x_0})_{n\geq 0}$, and erase the loops up to the first time we hit $\calT_\ell$ or a cycle is accepted, that is, up to the stopping time
            \begin{equation*}
                \tau_{\mathrm{stop},\calT_\ell}^{x_0} = \tau_{\to \calT_\ell}^{x_0} \wedge  \min_{ \text{ cycle }c} \tau^{x_0}_{\circlearrowleft c}.\label{e:t_stop}
            \end{equation*}
            Denote by $\calJ$ the branch or lasso corresponding to $(X^{x_0}_{0}, \dots, X^{x_0}_{\tau_{\mathrm{stop},\calT_\ell}})$.\\
            \item Build the new subgraph $\calT_{\ell + 1} = \calT_\ell \cup \calJ$.
            \item If $\calV\setminus \calV_\ell$ is empty, return $\calT = \calT_{\ell+1}$. Otherwise, increment $\ell$ by $1$.
            \end{enumerate}
    \end{itemize}
    \caption{\cyclepopping{}$(G, \mathsf{W}, \alpha, \iota)$}
    \label{a:cyclepopping}
\end{algorithm}
We make an important observation for later use.
\begin{remark}
    \label{rem:shape}
    Any run of \cyclepopping{}$(G,\mathsf{W},\alpha,\iota)$ which outputs a fixed CRSF $\calF$ produces the \emph{same} sequence of subgraphs.
    In other words for a fixed ordering $\iota$ of the nodes, we have a \emph{one-to-one} function
    \begin{equation}
        \mathrm{stages}_\iota(\calF) =
    (\calF_1, \calF_2, \dots, \calF_{\kappa_\iota(\calF)}),\label{eq:map_r_to_branches}
    \end{equation}
    with $\calF_1 \subset \calF_2 \dots \subset \calF_{\kappa_\iota(\calF)} =\calF$ the realizations of $\calT_1 \subset \dots \subset \calT_{\kappa_\iota(\calF)}$ in Algorithm~\ref{a:cyclepopping}.
    This can be seen as follows.
    At the first stage, the oriented CRSF $\calF$ naturally determines an oriented cycle-rooted tree, say $\calF_1$, which connects the first node in the ordering to an oriented cycle of $\calF$.
    This is true since $\calF$ is spanning.
    The second node in the ordering which does not belong to $\calF_1$ is the origin of another oriented subgraph which either connects to $\calF_1$ or is another cycle-rooted tree.
    The union of this new subgraph and $\calF_1$ is $\calF_2$. 
    The rest of the construction proceeds similarly.
    We refer to Figure~\ref{fig:stages} for an illustration.
\end{remark}

\begin{remark}
    Noticeably, Algorithm~\ref{a:cyclepopping} actually generalizes the original version of Wilson's algorithm for sampling an ST with a given root node. 
    Namely, we can attach an auxilliary cycle to the root node and make sure that all cycles have a vanishing acceptance probability except this auxilliary cycle which is accepted with probability one.
    Running Algorithm~\ref{a:cyclepopping} on this auxilliary graph actually samples a (connected) CRSF including only the auxilliary cycle, which yields a rooted ST upon erasing this cycle after completion of the algorithm.
\end{remark}

\begin{figure}[t]
  \centering
    \begin{subfigure}[b]{0.15\linewidth}
        \begin{tikzpicture}
            \node[blue] (R11) at (0,0) {\textbullet};
            \node [blue]  at (0 + 0.5,0 - 0.3) {$\iota(1)$};
            \node (R12) at (1,0) {\textbullet};
            \node (R13) at (2,0) { $\cdot$};
            \node (R14) at (3,0) { $\cdot$};
            \node (R21) at (0,-1) {\textbullet};
            \node (R22) at (1,-1) {\textbullet};
            \node (R23) at (2,-1) { $\cdot$};
            \node (R24) at (3,-1) { $\cdot$};
            \node (R31) at (0,-2) {$\cdot$};
            \node (R32) at (1,-2) {$\cdot$};
            \node (R33) at (2,-2) {$\cdot$};
            \node (R34) at (3,-2) { $\cdot$};
            \draw [thick, draw=gray, opacity=0.2] (R12.center)--(R13.center);
            \draw [thick, draw=gray, opacity=0.2] (R22.center)--(R23.center);
            \draw [thick, draw=gray, opacity=0.2] (R22.center)--(R32.center);
            \draw [thick, draw=gray, opacity=0.2] (R23.center)--(R33.center);
            \draw [thick, draw=gray, opacity=0.2] (R33.center)--(R34.center);
            \path [->, ultra thick] (R11) edge  (R12);
            \draw [thick, draw=gray, opacity=0.2] (R11.center)--(R12.center);
            \path [->, ultra thick] (R12) edge  (R22);
            \draw [thick, draw=gray, opacity=0.2] (R12.center)--(R22.center);
            \path [->, ultra thick] (R22) edge  (R21);
            \draw [thick, draw=gray, opacity=0.2] (R22.center)--(R21.center);
            \path [->, ultra thick] (R21) edge  (R11);
            \draw [thick, draw=gray, opacity=0.2] (R21.center)--(R11.center);
            \path [->, thick, dashed] (R31) edge  (R21);
            \draw [thick, draw=gray, opacity=0.2] (R31.center)--(R21.center);
            \path [->, thick, dashed] (R32) edge  (R31);
            \draw [thick, draw=gray, opacity=0.2] (R32.center)--(R31.center);
            \path [->, thick, dashed] (R13) edge  (R14);            
            \draw [thick, draw=gray, opacity=0.2] (R13.center)--(R14.center);
            \path [->, thick, dashed] (R14) edge  (R24);
            \draw [thick, draw=gray, opacity=0.2] (R14.center)--(R24.center);
            \path [->, thick, dashed] (R24) edge  (R23);
            \draw [thick, draw=gray, opacity=0.2] (R24.center)--(R23.center);
            \path [->, thick, dashed] (R23) edge  (R13);
            \draw [thick, draw=gray, opacity=0.2] (R23.center)--(R13.center);
            \path [->, thick, dashed] (R34) edge  (R24);
            \draw [thick, draw=gray, opacity=0.2] (R34.center)--(R24.center);
            \path [->, thick, dashed] (R33) edge  (R32);
            \draw [thick, draw=gray, opacity=0.2] (R33.center)--(R32.center);
        \end{tikzpicture}
        \caption{$\calF_1$}
    \end{subfigure}
    \hfill
    \begin{subfigure}[b]{0.15\linewidth}
        \begin{tikzpicture}
            \node (R11) at (0,0) {\textbullet};
            \node (R12) at (1,0) {\textbullet};
            \node (R13) at (2,0) { $\cdot$};
            \node (R14) at (3,0) { $\cdot$};

            \node (R21) at (0,-1) {\textbullet};
            \node (R22) at (1,-1) {\textbullet};
            \node (R23) at (2,-1) {$\cdot$};
            \node (R24) at (3,-1) { $\cdot$};

            \node (R31) at (0,-2) { \textbullet};
            \node (R32) at (1,-2) { \textbullet};
            \node [blue] (R33) at (2,-2) { \textbullet};
            \node [blue]  at (2,-2 + 0.3) {$\iota(2)$};
            \node (R34) at (3,-2) { $\cdot$};

            \draw [thick, draw=gray, opacity=0.2] (R12.center)--(R13.center);
            \draw [thick, draw=gray, opacity=0.2] (R22.center)--(R23.center);
            \draw [thick, draw=gray, opacity=0.2] (R22.center)--(R32.center);
            \draw [thick, draw=gray, opacity=0.2] (R23.center)--(R33.center);
            \draw [thick, draw=gray, opacity=0.2] (R33.center)--(R34.center);

            \path [->, ultra thick] (R11) edge  (R12);
            \draw [thick, draw=gray, opacity=0.2] (R11.center)--(R12.center);
            \path [->, ultra thick] (R12) edge  (R22);
            \draw [thick, draw=gray, opacity=0.2] (R12.center)--(R22.center);
            \path [->, ultra thick] (R22) edge  (R21);
            \draw [thick, draw=gray, opacity=0.2] (R22.center)--(R21.center);
            \path [->, ultra thick] (R21) edge  (R11);
            \draw [thick, draw=gray, opacity=0.2] (R21.center)--(R11.center);

            \path [->, ultra thick] (R33) edge  (R32);
            \draw [thick, draw=gray, opacity=0.2] (R33.center)--(R32.center);
            \path [->, ultra thick] (R31) edge  (R21);
            \draw [thick, draw=gray, opacity=0.2] (R31.center)--(R21.center);
            \path [->, ultra thick] (R32) edge  (R31);
            \draw [thick, draw=gray, opacity=0.2] (R32.center)--(R31.center);

            \path [->, thick, dashed] (R31) edge  (R21);
            \draw [thick, draw=gray, opacity=0.2] (R31.center)--(R21.center);
            \path [->, thick, dashed] (R32) edge  (R31);
            \draw [thick, draw=gray, opacity=0.2] (R32.center)--(R31.center);

            \path [->, thick, dashed] (R13) edge  (R14);            
            \draw [thick, draw=gray, opacity=0.2] (R13.center)--(R14.center);
            \path [->, thick, dashed] (R14) edge  (R24);
            \draw [thick, draw=gray, opacity=0.2] (R14.center)--(R24.center);
            \path [->, thick, dashed] (R24) edge  (R23);
            \draw [thick, draw=gray, opacity=0.2] (R24.center)--(R23.center);
            \path [->, thick, dashed] (R23) edge  (R13);
            \draw [thick, draw=gray, opacity=0.2] (R23.center)--(R13.center);

            \path [->, thick, dashed] (R34) edge  (R24);
            \draw [thick, draw=gray, opacity=0.2] (R34.center)--(R24.center);

            \path [->, thick, dashed] (R33) edge  (R32);
            \draw [thick, draw=gray, opacity=0.2] (R33.center)--(R32.center);
        \end{tikzpicture}
        \caption{$\calF_2$}
    \end{subfigure}
    \hfill
    \begin{subfigure}[b]{0.15\linewidth}
        \begin{tikzpicture}
            \node (R11) at (0,0) {\textbullet};
            \node (R11) at (0,0) {\textbullet};

            \node (R12) at (1,0) {\textbullet};
            \node (R13) at (2,0) {\textbullet};
            \node (R14) at (3,0) {\textbullet};

            \node (R21) at (0,-1) {\textbullet};
            \node (R22) at (1,-1) {\textbullet};
            \node (R23) at (2,-1) {\textbullet};
            \node (R24) at (3,-1) {\textbullet};

            \node (R31) at (0,-2) {\textbullet};
            \node (R32) at (1,-2) {\textbullet};
            \node (R33) at (2,-2) {\textbullet};
            \node [blue] (R34) at (3,-2) {\textbullet};
            \node [blue]  at (3-0.4,-2 + 0.2) {$\iota(3)$};

            \draw [thick, draw=gray, opacity=0.2] (R12.center)--(R13.center);
            \draw [thick, draw=gray, opacity=0.2] (R22.center)--(R23.center);
            \draw [thick, draw=gray, opacity=0.2] (R22.center)--(R32.center);
            \draw [thick, draw=gray, opacity=0.2] (R23.center)--(R33.center);
            \draw [thick, draw=gray, opacity=0.2] (R33.center)--(R34.center);

            \path [->, ultra thick] (R11) edge  (R12);
            \draw [thick, draw=gray, opacity=0.2] (R11.center)--(R12.center);
            \path [->, ultra thick] (R12) edge  (R22);
            \draw [thick, draw=gray, opacity=0.2] (R12.center)--(R22.center);
            \path [->, ultra thick] (R22) edge  (R21);
            \draw [thick, draw=gray, opacity=0.2] (R22.center)--(R21.center);
            \path [->, ultra thick] (R21) edge  (R11);
            \draw [thick, draw=gray, opacity=0.2] (R21.center)--(R11.center);

            \path [->, ultra thick] (R33) edge  (R32);
            \draw [thick, draw=gray, opacity=0.2] (R33.center)--(R32.center);
            \path [->, ultra thick] (R31) edge  (R21);
            \draw [thick, draw=gray, opacity=0.2] (R31.center)--(R21.center);
            \path [->, ultra thick] (R32) edge  (R31);
            \draw [thick, draw=gray, opacity=0.2] (R32.center)--(R31.center);

            \path [->, ultra thick] (R13) edge  (R14);            
            \draw [thick, draw=gray, opacity=0.2] (R13.center)--(R14.center);
            \path [->, ultra thick] (R14) edge  (R24);
            \draw [thick, draw=gray, opacity=0.2] (R14.center)--(R24.center);
            \path [->, ultra thick] (R24) edge  (R23);
            \draw [thick, draw=gray, opacity=0.2] (R24.center)--(R23.center);
            \path [->, ultra thick] (R23) edge  (R13);
            \draw [thick, draw=gray, opacity=0.2] (R23.center)--(R13.center);

            \path [->, ultra thick] (R34) edge  (R24);
            \draw [thick, draw=gray, opacity=0.2] (R34.center)--(R24.center);

            \path [->, ultra thick] (R33) edge  (R32);
            \draw [thick, draw=gray, opacity=0.2] (R33.center)--(R32.center);
        \end{tikzpicture}
        \caption{$\calF_3 = \calF$}
    \end{subfigure}
    \hspace{1.5cm}
    \caption{Illustration of $\mathrm{stages}_\iota(\calF) =
     (\calF_1, \calF_2, \calF)$ given in \eqref{eq:map_r_to_branches} on a CRSF $\calF$ of a $3\times4$ grid graph.
    The first three nodes in the order given by $\iota$ are colored in blue; the remainder of the ordering is not displayed. 
    At each stage of the growing sequence, the subgraph appears in bold whereas its complement in $\calF$ is represented with dashed lines.
    Here $\kappa_{\iota}(\calF) = 3$. \label{fig:stages}}
\end{figure}
\section{Proof of correctness in the determinantal case}\label{sec:SamplingCRSFs}
In this section, we prove that \cyclepopping{} outputs a sample from measure \eqref{eq:proba_CRSF_non_det} in the particular case where the cycle weights derive from a $\Uone$-connection $\phi$, as introduced by \cite{kenyon2011}.
We actually flesh out a sketch of proof given by \cite{Kassel15}, itself based on the algebraic proof for spanning trees by \cite{Marchal99}; see Section~\ref{s:CyclePopping_for_STs}.
A byproduct of the proof is that it allows us to investigate the complexity of \cyclepopping{}.
\subsection{A determinantal probability measure over CRSFs}
Given a $\Uone$-connection $\phi$ on the edges of $G$, we first define the holonomy of a based loop $\gamma = (x_0, \dots, x_k)$ as the unit-modulus complex number
\begin{equation}
    \label{e:holonomy}
    \hol(\gamma) = \phi_{x_0 x_1} \dots \phi_{x_{k-1}x_0} \triangleq \exp(-\rmi \theta(\gamma)).
\end{equation}
In words, $\theta(\gamma)$ is seen as the angular inconsistency obtained by composing the edge angles along $\gamma$. 
Holonomies are conjugated under orientation flip, and the holonomy of a backtrack is always equal to $1$.
\begin{assumption}[non-trivial connection]\label{ass:non-trivial}
    There exists at least one cycle $c$ such that $\hol(c) \neq 1$.
\end{assumption}
This assumption is crucial to define the determinantal measure on CRSFs studied in this paper.
If Assumption~\ref{ass:non-trivial} does not hold, the magnetic Laplacian is unitarily equivalent to the combinatorial Laplacian.
An important remark is that Assumption~\ref{ass:non-trivial} implies that $\mathsf{\Delta}$ is nonsingular, as can be seen by considering the quadratic form associated to $\mathsf{\Delta}$; see e.g.\ \citep{AFST_2011_6_20_3_599_0}.

For technical reasons that will become clear below, we further need to require that the \emph{sign-flipped} connection $xy \mapsto e^{\rmi \pi} \phi_{xy}$ is non-trivial, so that its magnetic Laplacian is also non-singular.
\begin{assumption}[non-trivial sign-flipped connection]
    \label{ass:non-trivial-sign-flipped}
    There is at least one cycle $c$ such that $\hol(c) \neq (-1)^{|c|}$.
\end{assumption}

Denote by $\calU$ a CRSF. 
Under Assumption~\ref{ass:non-trivial}, we define a probability measure on  CRSFs as
\begin{equation}
    \label{eq:proba_CRSF}
    \mu_{\mathrm{CRSF}}(\calU) =
     \frac{1}{\det (\mathsf{\Delta})}\prod_{e\in \calU} w_e \times
    \prod_{\substack{\text{non-oriented}\\\text{cycle }c \subseteq \calU }}
    2 \cdot (1 - \cos \theta(c)).
\end{equation}
That \eqref{eq:proba_CRSF} is well-normalized is a consequence of the generalized matrix-tree theorem of \citet{FORMAN199335}.

Define the weighted and \emph{oriented} edge-vertex incidence matrix $\mathsf{B}\in\mathbb{C}^{m\times n}$ given by
\begin{align}
    (\mathsf{B})_{e,v}
    = \begin{cases}
        \sqrt{w_e}  & \text{ if $e= vu$ for some $u$,} \\
        -\phi_{e}^*\sqrt{w_e} & \text{ if $e= uv$ for some $u$,} \\
        0          & \text{ otherwise.}
    \end{cases}\label{eq:mag_incidence}
\end{align}
As shown by \citet{kenyon2011}, the probability measure \eqref{eq:proba_CRSF} is determinantal over the graph edges with correlation kernel
\[
    \mathsf{K} = \mathsf{B} \mathsf{\Delta}^{-1} \mathsf{B}^*,
\]
that is, the measure of a subset $\calS$ of the edges of a CRSF $\calU$ is $\det(\mathsf{K}_\calS)$. 

Two remarks are in order.
First, the measure \eqref{eq:proba_CRSF} favors CRSFs that include edges with large weights, and cycles with large inconsistencies.
Second, a natural question arising from scrutinizing \eqref{eq:proba_CRSF} goes as follows: what is the distribution of the cycles of a CRSF distributed according to the determinantal measure ? 
Proposition~\ref{prop_incidence_prob_cycles} gives the incidence probabilities of sets of non-concurrent cycles in the samples of \eqref{eq:proba_CRSF}.
\begin{proposition}[determinantal formula for cycles in determinantal CRSFs]\label{prop_incidence_prob_cycles}
    Let $\calU$ be a CRSF. 
    Define by $\cycles(\calU)$ the set of  cycles of $\calU$.
    Under Assumption~\ref{ass:non-trivial}, let $\calT$ be a random CRSF distributed according to the determinantal measure \eqref{eq:proba_CRSF}.
    Let $\calC$ be a set of non-concurrent cycles. 
    We have
    \[
        \mathbb{P}(\calC \subseteq \cycles(\calT)) = \nu(\calC) \det (\Delta^{-1})_{\nodes(\calC)},
    \]
    where $\nodes(\calC)$  is the set of nodes in $\calC$ and with the weight 
    $$
     \nu(\calC) = \prod_{e\in \calC} w_e \prod_{c \in \calC}(2-2 \cos \theta(c)).
    $$
\end{proposition}

Incidentally, Proposition~\ref{prop_incidence_prob_cycles} is a magnetic analogue for the expression of the process of roots of the determinantal spanning forest \eqref{eq:SF_measure} as given by \citet{AvGaud2018}.
Its proof is given in Appendix.

Let us now discuss sampling algorithms. 
Since the measure \eqref{eq:proba_CRSF} is determinantal, it is thus licit to use the algebraic algorithm of \citet{HKPV06} to sample from it.
In this paper, we take a different route to sampling \eqref{eq:proba_CRSF} and consider a variant of \cyclepopping{} introduced by \citet{KK2017}.
\cyclepopping{} has several interesting features such as being completely graph-based, decentralized, and robust to numerical errors; see e.g.\ \citep{FanBar22} for more details and an application.
Third, \cyclepopping{} actually allows to consider more general cycle weights than \eqref{eq:determinantal_cycle_weights}, yielding non-determinantal measures; we later treat this general case in Section~\ref{s:Poisson}.

Actually, since it is based on a random walk, \cyclepopping{} samples an \emph{oriented} CRSF, that is, a spanning subgraph, each connected component of which has a unique oriented cycle, towards which all edges are directed.
Accounting for the fact that each cycle may have two orientations, the probability measure {of interest} is given by, for  $\calF$ an oriented CRSF,
\begin{equation}
    \label{eq:proba_CRSF_oriented}
    \mu_{\mathrm{CRSF}}(\calF) =
     \frac{1}{\det (\mathsf{\Delta})}\prod_{e\in \calF} w_e \times
    \prod_{\substack{\text{oriented}\\\text{cycle }c \subseteq \calF }}
    (1 - \cos \theta(c)).
\end{equation}
This is a special case of \eqref{eq:proba_CRSF_non_det}, with the cycle weights taken to be
\begin{equation}
    \label{eq:determinantal_cycle_weights}
    \alpha(c) = 1 - \cos \theta(c).
\end{equation} 
We further make the following assumption on cycle holonomies, which shall later allow us to interpret $\alpha(c)$ in \eqref{eq:determinantal_cycle_weights} as a probability.
\begin{assumption}[weakly inconsistent cycles]\label{ass:weak}
    For all cycles $c$, $\cos\theta(c) \geq 0$.
\end{assumption} 
In physical terms, this condition intuitively assumes a weak flux of the external magnetic field through the cycles.

\begin{remark}\label{rem:ass2}
    If Assumption~\ref{ass:non-trivial} and Assumption~\ref{ass:weak} hold, then Assumption~\ref{ass:non-trivial-sign-flipped} is necessarily satisfied.
\end{remark}
We shall make use of the following matrix, which we think of as a connection-aware transition matrix, 
\begin{equation}
    \mathsf{\Pi} = \mathsf{D}^{-1}(\mathsf{W} \odot \mathsf{\Phi}). 
    \label{eq:Pi}
\end{equation}
Note that $(\mathsf{\Pi})_{xx} =0$ for all $x\in \calV$.
\begin{proposition}\label{prop:Pi}
    If Assumption~\ref{ass:non-trivial} holds, the eigenvalues of $\mathsf{\Pi}$ are in the interval $[-1,1)$. If Assumption~\ref{ass:non-trivial-sign-flipped}  holds, the eigenvalues of $\mathsf{\Pi}$ are in the  interval $(-1,1]$.
\end{proposition}
In particular, under Assumption~\ref{ass:non-trivial}, $\I - \Pi$ is non-singular. 
The first part of the proof of Proposition~\ref{prop:Pi} is inspired by \citet[page 63]{SingWu16}. 
\begin{proof}
    The matrix $\Pi$ is related by a similarity transformation to the Hermitian matrix $\mathsf{\Pi}' = \mathsf{D}^{-1/2} (\mathsf{W} \odot \mathsf{\Phi})\mathsf{D}^{-1/2}$. Thus, the eigenvalues of $\mathsf{\Pi}$ and $\mathsf{\Pi}'$ are identical and real-valued. Next, we observe that the eigenvalues of $\I \pm \mathsf{\Pi}'$ are non-negative. Indeed, for all $\mathsf{v}\in \mathbb{C}^n$, we have
    \[
        \mathsf{v}^*(\I \pm \mathsf{\Pi}')\mathsf{v} = \sum_{x\sim y} w_{xy}\left|\frac{\mathsf{v}_x}{\sqrt{\deg(x)}} \pm \frac{\phi_{yx}\mathsf{v}_y}{\sqrt{\deg(y)}}\right|^2 \geq 0,
    \]
    where $w_{xy}\geq 0$.
    From this observation, we deduce that the eigenvalues of $\I + \mathsf{\Pi}^\prime$ and $\I - \mathsf{\Pi}^\prime$ are non-negative. Thus, the smallest (resp.\ largest) eigenvalue of $\mathsf{\Pi}^\prime$ cannot be smaller (resp.\ larger) than than $-1$ (resp.\ $1$). Since $\mathsf{\Pi}$ and $\mathsf{\Pi}^\prime$ share the same eigenvalues, the spectrum of $\mathsf{\Pi}$ lies in $[-1,1]$. 
    
    Furthermore, since $\mathsf{\Pi} = \I - \mathsf{D}^{-1} \mathsf{\Delta}$ with $\mathsf{\Delta}$ non-singular due to Assumption~\ref{ass:non-trivial}, $\mathsf{\Pi}$ cannot have an eigenvalue equal to $1$. 
    This shows that the spectrum of $\mathsf{\Pi}$ lies within the interval  $[-1,1)$.
    A similar argument applies to show that $\I + \mathsf{\Pi}'$ is nonsingular under Assumption~\ref{ass:non-trivial-sign-flipped}. We simply need to define the magnetic Laplacian for the sign-flipped connection. Hence, if Assumption~\ref{ass:non-trivial-sign-flipped} holds, $\mathsf{\Pi}$ cannot have an eigenvalue equal to $-1$.
    This completes the proof.
\end{proof}
We are ready to state the correctness of \cyclepopping{}.
\begin{proposition}[correctness]\label{prop:correctness}
    Let Assumption~\ref{ass:non-trivial} and Assumption~\ref{ass:weak} hold and fix any ordering $\iota$ of the nodes.
    Let $\calF$ be a CRSF in the support of \eqref{eq:proba_CRSF_oriented},
    and let $\calF_1, \calF_2, \dots, \calF_{\kappa_{\iota}(\calF)}$ be the deterministic decomposition induced by $\iota$; see Remark~\ref{rem:shape}.

    Denote by $\calT_\ell$ the random subgraphs output by \cyclepopping{}$(G,\mathsf{W},\alpha,\iota)$.
    For all $1\leq \ell \leq \kappa_{\iota}(\calF)$, we have
\begin{align}
    \mathbb{P}(\calT_\ell = \calF_\ell) &= \mathbb{P}(\calT_1 = \calF_1, \dots,\calT_\ell = \calF_\ell) \nonumber \\
    &=   \frac{\det(\mathsf{\Delta}_{\overline{\calV_\ell}})}{\det(\mathsf{\Delta})} \times\prod_{e\in \calF_\ell}w_{e} \times \prod_{\substack{\text{\rm oriented}\\\text{\rm cycle } c \subseteq \calF_\ell}} \Big(1- \cos\theta(c)\Big),
    \label{e:correctness_determinantal}
\end{align}
where $\calV_\ell$ is the set of nodes in $\calF_\ell$.
\end{proposition}
The first equality in \eqref{e:correctness_determinantal} is a consequence of Remark~\ref{rem:shape}.
When $\ell=\kappa_{\iota}(\calF)$, the second equality shows that \cyclepopping{} indeed produces samples from \eqref{eq:proba_CRSF_oriented}, upon noting that $\calF_{\kappa_{\iota}(\calF)} = \calF$ and $\det\mathsf{\Delta}_{\emptyset} = 1$.
The extra factor $2$ in front of each cycle weight in \eqref{eq:proba_CRSF} comes from the fact that each cycle of a CRSF can be obtained with two orientations in an oriented CRSF.
Remark that backtracks are popped with probability $1$.
Finally, note that the resulting measure does not depend on the node ordering $\iota$. 
Hence, in the sequel, we often omit the dependence on $\iota$.

\subsection{A last-exit decomposition} 
\label{sec:proof_correctness}

This section revisits the proof sketch in \citep[Section 2.3]{Kassel15}, fleshing out all the technical details.
A closely related approach was used by \citet{Marchal99} to prove the correctness of Wilson's algorithm for uniformly sampling spanning trees; see also \citep{AvGaud2018}.

With the notation of Proposition~\ref{prop:correctness},
\begin{equation}
    \mathbb{P}(\calT_1 = \calF_1, \dots,  \calT_\ell = \calF_\ell) = \mathbb{E}\prod_{i=1}^{\ell} 1_{\calT_i = \calF_i} = \mathbb{E} \prod_{i=1}^{\ell} 1_{\calT_i = \calF_{i-1}\cup \calJ_i},
    \label{e:expectation_of_indicators}
\end{equation}
where $\calF_1$ is a lasso, and for $\ell\geq 2$, $\calF_i = \calF_{i-1} \cup \calJ_i$, $\calJ_i$ being either a branch or a lasso.
Now, fix $1\leq i\leq \ell$, and write $\calJ_i = (x_0^i, \dots, x_{k_i-1}^i)$. 
On $\{\calT_{i-1} = \calF_{i-1}\}$, the event $\{\calT_i = \calF_{i-1} \cup \calJ_i\}$ decomposes as a disjoint union
\begin{align}
    \{\calT_{i-1} = \calF_{i-1}\}\cap\{\calT_i = \calF_{i-1} \cup \calJ_i\} 
    &= \{\calT_{i-1} = \calF_{i-1}\} \cap \left[\bigcup_{\bfm \in \mathbb{N}^{k_i-1}} \calE_{\bfm}^i\right],
    \label{e:decomposition_as_disjoint_union}
\end{align}
where for $\bfm \in \mathbb{N}^{k-1}$, the event $\calE^i_{\bfm}$ is defined by a last-exit decomposition; see~\citep[Proposition 4.6.4]{LaLi10}. 
In words, we walk an $m_1$-step loop based as $x_0^i$ without hitting $\calV_{i-1}$, pop the loop, and walk one step to $x_1^i$, \emph{never to return} to $x_0^i$; then we walk an $m_2$-loop based at $x_1^i$ without hitting $\calV_{i-1}$, pop the loop, and walk one step to $x_2^i$, never to return to $x_1^i$, etc.\ until we accept the cycle in $\calJ_i$ if the latter is a lasso, or we hit $\calV_{i-1}$ if $\calJ_i$ is a branch. 
Formally, let $m_1^* = m_1$ and $m_j^* = m_{j-1}^* + m_j + 1$ for $j\geq 2$. 
We first treat the case of $\calJ_i = (x_0^i, \dots, x_{k_i-1}^i)$ being a lasso, so that $k_i\geq 2$ and $x_j^i = x_{k_i}^i$ for some $j<k_i$. 
Denoting by $c_i$ the cycle of $\calJ_i$, we let 
\begin{align}
    \calE_{\bfm}^i
    &= 
    \overbrace{\{ X^{x_0^i}_{m_1^*} = x_0^i \text{ and } m_1^* < \tau^{x_0^i}_{\mathrm{stop}, \calV_{i-1}}\}}^{\text{loops of $m_1^*$ steps based at $x_0^i$}} \cap \overbrace{\{
    X^{x_0^i}_{m_1^*+1} = x_1^i \}}^{\text{step from $x_0^i$ to $x_1^i$ }} \nonumber\\
    &\cap\{
    X^{x_0^i}_{m_2^*} = x_1^i \text{ and } m_2^* < \tau^{x_0^i, (\text{after }m_1^* + 1)}_{\mathrm{stop}, \calV_{i-1}\cup \{x_0^i\}} \}\cap 
    \{X^{x_0^i}_{m_2^* + 1} = x_2^i \} \nonumber\\
    &\vdots\nonumber\\
    &\cap\{
    X^{x_{0}^i}_{m_{k_i-1}^*} = x_{k_i-2}^i \text{ and } m_{k_i-1}^* < \tau^{x_0^i, (\text{after }m^*_{k_i-2} + 1)}_{\mathrm{stop},\calV_{i-1}\cup \{x_0^i, \dots, x_{k_i-3}^i\}} \} \nonumber\\
    &\cap \{
    X^{x_{0}^i}_{m_{k_i-1}^* + 1} = x^i_{k_i-1} \}\nonumber\\
    &\cap \underbrace{\{B^{x_0^i}_{c_i, m_{k_i-1}^* + 1} = 1\}}_{\text{$c_i$ is accepted}}.
    \label{eq:fund_event}
\end{align}
When $\calJ_i$ is a branch rather than a lasso, the last row of \eqref{eq:fund_event} is simply omitted: we have hit $\calV_{i-1}$ at $x^i_{k_i-1}$.

Repeatedly plugging \eqref{e:decomposition_as_disjoint_union} into \eqref{e:expectation_of_indicators}, we obtain 
\begin{equation}
    \mathbb{P}(\calT_1 = \calF_1, \dots,  \calT_\ell = \calF_\ell) = \mathbb{E}\left[\prod_{i=1}^{\ell} \sum_{\bfm \in \mathbb{N}^{k_i-1}} 1_{\calE_{\bfm}^i}\right],
\end{equation}
upon noting that, by convention, we set $\calF_0 = \emptyset$ and thus $\calT_0=\calF_0$ almost surely.
By independence of the $\ell$ involved SRWs and countable additivity, it comes
\begin{align}
    \label{e:decomposition_in_fund_event}
    \mathbb{P}(\calT_1 = \calF_1, \dots,  \calT_\ell = \calF_\ell) = \prod_{i=1}^{\ell} \left[\sum_{\bfm \in \mathbb{N}^{k_i-1}} \mathbb{P}(\calE_{\bfm}^i)\right].
\end{align}
We are thus led to computing $\mathbb{P}(\calE_{\bfm}^i)$, for $1\leq i\leq \ell$, and $\bfm\in \mathbb{N}^{k_i}$.

Again, we start with the case where $\calJ_i$ is a lasso. 
The event of the last row of \eqref{eq:fund_event} is independent from the rest. 
Then we use Lemma~\ref{lem:stopping_time_after_m} and apply the Markov property to the chain $$
Z_{n}^{x_0^i} = (X_n^{x_0^i}, B_{1, c_1}^{x_0^i}, \dots,  B_{1, c_d}^{x_0^i}),
$$
where $c_1, \dots, c_d$ are the cycles of the graph $G$.
Remembering that $\mathbb{Q}_x$ denotes the law of the SRW $(X_n)$ on $G$ with initial distribution $\delta_x$, we obtain,
\begin{align}
    &\mathbb{P}(\calE_{\bfm}^i)\nonumber\\
    &=  
    \mathbb{Q}_{x_0^i}( \{X_{m_1^*} = x_0^i \text{ and } m_1 < \tau_{\mathrm{stop}, \calV_{i-1}}\})  p_{x_0^i x_1^i}\nonumber\\
    &\quad \ \mathbb{Q}_{x_1^i}(\{
    X_{m_2} = x_1^i \text{ and } m_2 < \tau_{\mathrm{stop}, \calV_{i-1}\cup \{x_0^i\}} \}) p_{x_1^i x_2^i} \nonumber\\
    &\quad \vdots\nonumber\\
    &\quad \times\mathbb{Q}_{x_{k_i-2}^i}(\{
    X_{m_{k_i-1}} = x_{k_i-2}^i \text{ and } m_{k_i-1} < \tau_{\mathrm{stop},\calV_{i-1}\cup \{x_0^i, \dots, x_{k_i-3}^i\}}\}) p_{x^i_{k_i-2} x^i_{k_i-1}}\nonumber\\
    &\quad \times \alpha(c_i).
    \label{eq:fund_event_markovied}
\end{align}
The case of $\calJ_i$ being a branch is similar, resulting in the omission of the factor $\alpha(c)$ in \eqref{eq:fund_event_markovied}.

Now, we recognize in \eqref{eq:fund_event_markovied} evaluations of the Green generating function of the chain, which we introduce before continuing the computation of \eqref{eq:fund_event_markovied}.

\subsection{The Green generating function\label{s:Green_generating_fct}}

We begin by computing the probability that $(X_n)$, starting from node $x$, makes a based loop of $k$ steps ($k\geq 2$), while staying in the complement of a subset of nodes $\calA$ and popping (i.e.\ erasing, rejecting) all the cycles along its path,
\begin{align}
    g_k(x,\calA) \triangleq \mathbb{Q}_{x}\Big(X_{k} = x \text{ and } k < \tau_{\to \calA} \wedge \min_{ \text{ cycle }c} \tau_{\circlearrowleft c}\Big),
    \label{eq:Green_k}
\end{align}
with $x \in \calV \setminus \calA$.
For definiteness, we write $g_1(x,\calA) = 0$ and $g_0(x,\calA) = 1$.
For an oriented based loop $\gamma = (x_0, x_1, \dots, x_{k-1}, x_0)$, where $k\geq 2$, write 
\begin{equation}
    q(\gamma) = p_{x_0x_1} \dots p_{x_{k-1}x_0},\label{eq:loop_weight}
\end{equation}
that is, the probability that the simple random walk starting at $x_0$ walks exactly on this closed path.\footnote{
For measures over loops in graphs and loop soups, we refer to \citep{Sznitman2012,LeJan11,lawler2004brownian,Symanzik69}.
}
In particular, $q(\gamma) = q(\bar{\gamma})$ for any based loop $\gamma$, in light of the definition of $p_{xy}$ in \eqref{eq:P}.
It is conventional to set $q(\gamma) = 1$ when $\gamma$ consists of only one point.
Note that based loops are not necessarily cycles, and may contain cycles, backtracks being popped almost surely in Algorithm~\ref{a:cyclepopping}.
\begin{notation}[set of based loops]
    Let $\calA$ be a subset of $\calV$ and $x\in\calV \setminus \calA$. 
    Denote by $\BLoops(\overline{\calA},x)$ the set of oriented based loops in $\calV \setminus \calA$ whose base is $x$, including trivial loops of only one point.
\end{notation}
\begin{notation}[cycles within a based loop]\label{not:pieces_loop}
    Let $\gamma$ be a loop based at $x_0$. 
    We now go over the oriented path defined by $\gamma$ and chronologically erase all the cycles encountered. 
    Note that these (oriented) cycles may occur more than once. 
    We collect the cycles, with multiplicity, in a multiset\footnote{A multiset is a set endowed with a function endoding element multiplicity.} we call $\cycles(\gamma)$.
    Figure~\ref{fig:heaps} illustrates a based loop, for which the multiset contains 3 cycles, one of which has multiplicity two.\footnote{
        Combinatorists might prefer defining a \emph{heap} \citep{krattenthaler2006theory}, as shown in Figure~\ref{fig:heap_from_loop}, but we do not need the additional sophistication for the moment.
    }
\end{notation}
\begin{figure}[h]
    \begin{subfigure}[b]{0.5\linewidth}
        \centering
    \begin{tikzpicture}[scale =1.5]
        \def\x{0.7};
        \def\y{0.7};
        \def\s{0.03};

        \node (P11) at (0,0) {\textbullet};
        \node  at (-0.3,0) {$x$};
        \node (P12) at (1,0) {\textbullet};
        \node (P13) at (2,0) {\textbullet};

        \node (P21) at (0+\x,0+\y) {\textbullet};
        \node (P22) at (1+\x,0+\y) {\textbullet};
        \node (P23) at (2+\x,0+\y) {\textbullet};

        \node (P31) at (0+2*\x,0+2*\y) {\textbullet};
        \node (P32) at (1+2*\x,0+2*\y) {\textbullet};
        \node (P33) at (2+2*\x,0+2*\y) {\textbullet};

        \draw [very thick, draw=gray, opacity=0.2] (P11.center)--(P12.center);
        \draw [very thick, draw=gray, opacity=0.2] (P12.center)--(P13.center);
        \draw [very thick, draw=gray, opacity=0.2] (P11.center)--(P21.center);
        \draw [very thick, draw=gray, opacity=0.2] (P12.center)--(P22.center);
        \draw [very thick, draw=gray, opacity=0.2] (P13.center)--(P23.center);
        \draw [very thick, draw=gray, opacity=0.2] (P21.center)--(P22.center);
        \draw [very thick, draw=gray, opacity=0.2] (P22.center)--(P23.center);
        \draw [very thick, draw=gray, opacity=0.2] (P21.center)--(P31.center);
        \draw [very thick, draw=gray, opacity=0.2] (P22.center)--(P32.center);
        \draw [very thick, draw=gray, opacity=0.2] (P23.center)--(P33.center);
        \draw [very thick, draw=gray, opacity=0.2] (P31.center)--(P32.center);
        \draw [very thick, draw=gray, opacity=0.2] (P32.center)--(P33.center);
        %
        \draw [magenta,->,very thick] (P11)--($(P12) + (0,1*\s)$);
        \draw [magenta,->,very thick] ($(P12)+ (0,1*\s)$)-- ($(P22)+ (0,2*\s)$);
        \draw [magenta,->,very thick] ($(P22)+ (0,2*\s)$)--($(P21)+ (0,3*\s)$);
        \draw [magenta,->,very thick] ($(P21)+ (0,3*\s)$)--($(P11)+ (0,4*\s)$);
    
        \draw [orange,->,very thick] ($(P11)+ (0,4*\s)$) to ($(P12)+ (0,5*\s)$);
        \draw [dotted, draw=gray, opacity=0.2] (P12)--($(P12)+ (0,5*\s)$);

        \draw [teal,->,very thick] ($(P12)+ (0,5*\s)$) to ($(P13)+ (0,6*\s)$);
        \draw [teal,->,very thick] ($(P13)+ (0,6*\s)$) to ($(P23)+ (0,7*\s)$);
        \draw [teal,->,very thick] ($(P23)+ (0,7*\s)$) to ($(P22)+ (0,8*\s)$);
        \draw [teal,->,very thick] ($(P22)+ (0,8*\s)$) to ($(P12)+ (0,9*\s)$);

        \draw [teal,->,very thick] ($(P12)+ (0,9*\s)$) to ($(P13)+ (0,10*\s)$);
        \draw [teal,->,very thick] ($(P13)+ (0,10*\s)$) to ($(P23)+ (0,11*\s)$);
        \draw [dashed, draw=gray, opacity=0.9] (P23)--($(P23)+ (0,11*\s)$);

        \draw [teal,->,very thick] ($(P23)+ (0,11*\s)$) to ($(P22)+ (0,12*\s)$);
        \draw [dashed, draw=gray, opacity=0.9] (P22)--($(P22)+ (0,12*\s)$);

        \draw [teal,->,very thick] ($(P22)+ (0,12*\s)$) to ($(P12)+ (0,13*\s)$);
        \draw [dashed, draw=gray, opacity=0.9] (P12)--($(P12)+ (0,13*\s)$);

        \draw [orange,->,very thick] ($(P12)+ (0,13*\s)$) to ($(P13)+ (0,14*\s)$);
        \draw [dashed, draw=gray, opacity=0.9] (P13)--($(P13)+ (0,15*\s)$);
        \draw [orange,->,very thick] ($(P13)+ (0,14*\s)$) to ($(P23)+ (0,15*\s)$);
        \draw [dashed, draw=gray, opacity=0.9] (P23)--($(P23)+ (0,15*\s)$);

        \draw [orange,->,very thick] ($(P23)+ (0,15*\s)$) to ($(P33)+ (0,16*\s)$);
        \draw [dashed, draw=gray, opacity=0.9] (P33)--($(P33)+ (0,16*\s)$);

        \draw [orange,->,very thick] ($(P33)+ (0,16*\s)$) to ($(P32)+ (0,17*\s)$);
        \draw [dashed, draw=gray, opacity=0.9] (P32)--($(P32)+ (0,17*\s)$);

        \draw [orange,->,very thick] ($(P32)+ (0,17*\s)$) to ($(P31)+ (0,18*\s)$);
        \draw [dashed, draw=gray, opacity=0.9] (P31)--($(P31)+ (0,18*\s)$);

        \draw [orange,->,very thick] ($(P31)+ (0,18*\s)$) to ($(P21)+ (0,19*\s)$);
        \draw [dashed, draw=gray, opacity=0.9] (P21)--($(P21)+ (0,19*\s)$);

        \draw [orange,->,very thick] ($(P21)+ (0,19*\s)$) to ($(P11)+ (0,20*\s)$);
        \draw [dashed, draw=gray, opacity=0.9] (P11)--($(P11)+ (0,20*\s)$);
    \end{tikzpicture}
    \caption{A based loop $\gamma$ based at $x$. \label{fig:loop_with cycles}}
    \end{subfigure}
    \hfill
    \begin{subfigure}[b]{0.5\linewidth}
        \centering
        \begin{tikzpicture}[scale = 1.5]
            \def\x{0.7};
            \def\y{0.7};
            \def\s{2.5};
            \node  at (-0.3,0) {$x$};
            \node (P11_shift2) at (0,0+\s) {\textbullet};
            \node (P12_shift2) at (1,0+\s) {\textbullet};
            \node (P13_shift2) at (2,0+\s) {\textbullet};

            \node [orange] at (1+2*\x,0+2*\y+0.5*\s) {$c_3$};
            \node [teal] at (1+2*\x,0+2*\y+0.5*\s -2.5/2) {$c_2$};
            \node [magenta] at (1+2*\x - 1.,0+2*\y+0.5*\s -2.5) {$c_1$};

            \node (P21_shift2) at (0+\x,0+\y+\s) {\textbullet};
            \node (P23_shift2) at (2+\x,0+\y+\s) {\textbullet};
    
            \node (P31_shift2) at (0+2*\x,0+2*\y+\s) {\textbullet};
            \node (P32_shift2) at (1+2*\x,0+2*\y+\s) {\textbullet};
            \node (P33_shift2) at (2+2*\x,0+2*\y+\s) {\textbullet};
            \node (P12_shift1) at (1,0+0.5*\s) {\textbullet};
            \node (P13_shift1) at (2,0+0.5*\s) {\textbullet};
            \node (P22_shift1) at (1+\x,0+\y+0.5*\s) {\textbullet};
            \node (P23_shift1) at (2+\x,0+\y+0.5*\s) {\textbullet};


            \node (P11) at (0,0) {\textbullet};
            \node (P12) at (1,0) {\textbullet};    
            \node (P21) at (0+\x,0+\y) {\textbullet};
            \node (P22) at (1+\x,0+\y) {\textbullet};
    
            \draw [dotted,draw=blue] (P11.center)--(P11_shift2.center);
            \draw [dotted, draw=blue, opacity=0.5] (P21.center)--(P21_shift2.center);

            \draw [dotted, draw=blue] (P12.center)--(P12_shift2.center);
            \draw [dotted, draw=blue, opacity=0.5] (P22.center)--(P22_shift1.center);

            \draw [dotted, draw=blue] (P23_shift1.center)--(P23_shift2.center);
            \draw [dotted, draw=blue] (P13_shift1.center)--(P13_shift2.center);

            \draw [magenta,->,very thick] (P11)--(P12);
            \draw [magenta,->,very thick] (P12)--(P22);
            \draw [magenta,->,very thick, opacity=0.8] (P22)--(P21);
            \draw [magenta,->,very thick] (P21)--(P11);
            \draw [teal,->,very thick,style=double] (P12_shift1)--(P13_shift1);
            \draw [teal,->,very thick,style=double] (P13_shift1)--(P23_shift1);
            \draw [teal,->,very thick,style=double, opacity=0.8] (P23_shift1)--(P22_shift1);
            \draw [teal,->,very thick,style=double] (P22_shift1)--(P12_shift1);
            \draw [orange,->,very thick] (P11_shift2)--(P12_shift2);
            \draw [orange,->,very thick] (P12_shift2)--(P13_shift2);
            \draw [orange,->,very thick] (P13_shift2)--(P23_shift2);
            \draw [orange,->,very thick] (P23_shift2)--(P33_shift2);
            \draw [orange,->,very thick] (P33_shift2)--(P32_shift2);
            \draw [orange,->,very thick] (P32_shift2)--(P31_shift2);
            \draw [orange,->,very thick] (P31_shift2)--(P21_shift2);
            \draw [orange,->,very thick] (P21_shift2)--(P11_shift2);
    
        \end{tikzpicture}
        \caption{The heap associated with $\gamma$. \label{fig:heap_from_loop}}

    \end{subfigure}

    \caption{Based loop as a heap of cycles.
    In Figure~\ref{fig:loop_with cycles}, we display a loop $\gamma$ based at $x$ where we visualized the steps by a small vertical shift.
    Colors allow to visualize $\cycles(\gamma) = \{c_1,(c_2)^2, c_3\}$.
    On the right-hand side, in Figure~\ref{fig:heap_from_loop}, we display the heap associated with $\gamma$.
    It is constituted of the following pieces: $c_1$, $c_2$ and $c_3$. 
    Double arrows indicate multiplicity $2$.
\label{fig:heaps}}
\end{figure}
Thanks to this notation, we can write the holonomy \eqref{e:holonomy} of a based loop as 
\begin{equation}
    \hol(\gamma) = \prod_{c\in \cycles(\gamma)} \hol(c) = \prod_{c\in \cycles(\gamma)}  \exp(-\rmi\theta(c)).   \label{eq:hol_decomposition} 
\end{equation}
\begin{figure}[b!]
    \begin{subfigure}[b]{0.5\linewidth}
        \centering
        \resizebox{\linewidth}{!}{
        \begin{tikzpicture}
            \node (A) at (0,0) {\textbullet};
            \node at (0,-0.25) {$x_0$};
            \node (B) at (2,1) {\textbullet};
            \node at (2,1+0.25) {$x_4$};
            \node at (2,0.4) {$c_1$};
            \node (C) at (4,0) {\textbullet};
            \node at (4,-0.25) {$x_5$};
            \node (D) at (6,0) {\textbullet};
            \node at (6,-0.25) {$x_6$};
            \path [->, ultra thick] (C) edge  (A);
            \draw [ultra  thick, draw=gray, opacity=0.2] (C.center)--(A.center);
            \node (A1) at (-0.5 ,1) {\textbullet};
            \node  at (-0.5 - 0.25 ,1 + 0.25) {$x_1$};
            \node  (A2) at (0.5,1) {\textbullet};
            \node  at (0.5 + 0.25 ,1 + 0.25) {$x_2$};
            \path [->, ultra thick] (A) edge  (A1);
            \draw [thick, draw=gray, opacity=0.2] (A.center)--(A1.center);
            \path [->, ultra thick] (A1) edge  (A2);
            \draw [thick, draw=gray, opacity=0.2] (A1.center)--(A2.center);
            \path [->, ultra thick] (A2) edge  (A);
            \draw [thick, draw=gray, opacity=0.2] (A2.center)--(A.center);
            \node  at (0,0.65) {$c_0$};
            \path [->, ultra thick] (A) edge  (B);
            \draw [thick, draw=gray, opacity=0.2] (A.center)--(B.center);
            \path [->, ultra thick](B) edge (C);
            \draw [thick, draw=gray, opacity=0.2] (B.center)--(C.center);
            \path [->, ultra thick](C) edge (D);
            \draw [thick, draw=gray, opacity=0.2] (C.center)--(D.center);
            \path [->, ultra thick](D) edge (C);
            \draw [thick, draw=gray, opacity=0.2] (D.center)--(C.center);
        \end{tikzpicture}
            }
            \caption{$\gamma= (x_0, x_1, x_2, x_0, x_4, x_5, x_6, x_5, x_0)$}
    \end{subfigure}
    \begin{subfigure}[b]{0.5\linewidth}
        \resizebox{\linewidth}{!}{
        \centering
        \begin{tikzpicture}
            \node (A) at (0,0) {\textbullet};
            \node at (0,-0.25) {$x_0$};
            \node (B) at (2,1) {\textbullet};
            \node at (2,1+0.25) {$x_4$};
            \node at (2,0.4) {$c_1$};
            \node (C) at (4,0) {\textbullet};
            \node at (4,-0.25) {$x_5$};
            \node (D) at (6,0) {\textbullet};
            \node at (6,-0.25) {$x_6$};
            \path [->, ultra thick] (C) edge  (A);
            \draw [thick, draw=gray, opacity=0.2] (C.center)--(A.center);
            \node (A1) at (-0.5 ,1) {\textbullet};
            \node  at (-0.5 - 0.25 ,1 + 0.25) {$x_1$};
            \node  (A2) at (0.5,1) {\textbullet};
            \node  at (0.5 + 0.25 ,1 + 0.25) {$x_2$};
            \path [->, ultra thick] (A) edge  (A2);
            \draw [thick, draw=gray, opacity=0.2] (A.center)--(A2.center);
            \path [->, ultra thick] (A2) edge  (A1);
            \draw [thick, draw=gray, opacity=0.2] (A2.center)--(A1.center);
            \path [->, ultra thick] (A1) edge  (A);
            \draw [thick, draw=gray, opacity=0.2] (A1.center)--(A.center);
            \node  at (0,0.65) {$\bar{c}_0$};
            \path [->, ultra thick] (A) edge  (B);
            \draw [thick, draw=gray, opacity=0.2] (A.center)--(B.center);
            \path [->, ultra thick](B) edge (C);
            \draw [thick, draw=gray, opacity=0.2] (B.center)--(C.center);
            \path [->, ultra thick](C) edge (D);
            \draw [thick, draw=gray, opacity=0.2] (C.center)--(D.center);
            \path [->, ultra thick](D) edge (C);
            \draw [thick, draw=gray, opacity=0.2] (D.center)--(C.center);
        \end{tikzpicture}
        }
        \caption{$\gamma^\prime = (x_0, x_2, x_1, x_0, x_4, x_5, x_6, x_5, x_0)$}    \end{subfigure}
    \caption{Two loops based at $x_0$, namely, $\gamma,\gamma^\prime \in\BLoops(\overline{\calA},x_0)$. We have $\cycles(\gamma)= \{c_0,c_1\}$ with $c_0 = (x_0,x_1,x_2,x_0)$ and $c_1 = (x_0, x_4, x_5, x_0)$. Similarly, $\cycles(\gamma^\prime)= \{\bar{c}_0,c_1\}$ The loop weights are $q(\gamma) = q(\gamma^\prime)$ and their holonomies satisfy $\hol(\gamma) + \hol(\gamma^\prime) = \hol(c_1)\Re\hol(c_0)$. \label{fig:popped_loops}}
\end{figure}
Now, recall that when a cycle $c$ is formed at time $n$, it is popped independently from the past, with probability $\mathbb{P}(B_{c,n} = 0) = \cos \theta(c)$.
Armed with these definitions, we formulate the following instrumental lemma.
\begin{lemma}[popping probabilities from holonomies]
    The probability \eqref{eq:Green_k} that the SRW starting at $x$ is back at $x$ in $k$ steps ($k\geq 2$), before hitting $\calA$ and before accepting any cycle  along its path, is
    \begin{align}\label{eq:based-loops}
    g_k(x,\calA)
        &=  \sum_{\substack{\gamma \in\BLoops(\overline{\calA},x)\nonumber\\|\gamma| = k}} q(\gamma) \prod_{c\in \cycles(\gamma)} \cos\theta(c) \nonumber\\
        &= \sum_{\substack{\gamma \in\BLoops(\overline{\calA},x)\\|\gamma| = k}}q(\gamma) \prod_{c\in \cycles(\gamma)}  \exp(-\rmi\theta(c)).
    \end{align}
\end{lemma}
\begin{proof}
    The equality \eqref{eq:based-loops} is obtained by realizing that, for an oriented based loop $\gamma$ of length $k$ containing an oriented simple cycle $c\in \cycles(\gamma)$, there is a distinct oriented based loop $\gamma'$ of length $k$ containing the cycle $\bar{c}$ (of opposite orientation) and keeping the same orientation for the remaining simple cycles. 
    Thus, the imaginary part of  $\hol(c)$ and $\hol(\bar{c})$ simply cancel out by summing the terms corresponding to $\gamma$ and $\gamma'$; see Figure~\ref{fig:popped_loops} for an illustration.
    This also trivially applies when $c$ is a backtrack since $\cos \theta(c) = 1$ in this case.
\end{proof}
Next, we replace each term by its definition in \eqref{eq:based-loops}, and we obtain
\begin{align}
    g_k(x,\calA) &=  \sum_{x_1, \dots, x_{k-1}\in \calV\setminus \calA} p_{xx_1} \dots p_{x_{k-1}x} \cdot \phi_{x x_1} \dots \phi_{x_{k-1}x}\nonumber\\
    &=  \left(\left(\mathsf{\Pi}_{\overline{\calA}}\right)^k\right)_{xx}\nonumber.
\end{align}

We collect the contributions of the popped loops based at $x$ of \emph{all possible lengths}, into the so-called Green generating function
\begin{align}
    G(t,x,x;\calA) &\triangleq \sum_{k=0}^{+\infty} t^k g_k(x,\calA), \quad t\in(0,1],
    \label{e:G_as_a_sum_over_k}
\end{align}
where $g_k(x,\calA)$ is defined in \eqref{eq:Green_k}.
Note that, to be consistent with most references, like \citep{Marchal99}, we use here the notation $G(t,x,x^\prime;\calA)$, though we will always take $x^\prime = x$.
In particular, by using \eqref{eq:based-loops},
\begin{align}
    G(t,x,x;\calA) =  \sum_{k=0}^{+\infty} t^k \left(\left(\mathsf{\Pi}_{\overline{\calA}}\right)^k\right)_{xx}
    = \left((\I - t \mathsf{\Pi}_{\overline{\calA}})^{-1}\right)_{xx},
    \label{e:green_as_neumann_series}
\end{align}
where the convergence of the Neumann series is guaranteed for all $t\in (0,1]$ under the conjunction of Assumption~\ref{ass:non-trivial} and Assumption~\ref{ass:non-trivial-sign-flipped} since the spectrum of $\Pi$ lies in $(-1,1)$ in this case as a consequence of Proposition~\ref{prop:Pi}.
The series also converges if Assumption~\ref{ass:non-trivial} and Assumption~\ref{ass:weak} hold, as a consequence of Remark~\ref{rem:ass2}.

Note that the term $k=0$ above is equal to one, and the term $k=1$ vanishes since $(\mathsf{\Pi})_{xx} = 0$ for all $x\in \calV$. 
We refer to \citet{PiTa18} for an interpretation of the Green generating function in terms of spanning forests.

Now, under Assumption~\ref{ass:non-trivial} and Assumption~\ref{ass:weak}, we consider in more detail the value of the Green generating function in $t=1$. 
Applying \eqref{eq:Pi} to \eqref{e:green_as_neumann_series},
\begin{equation}
        G(1,x,x;\calA) = \deg(x) \left((\mathsf{\Delta}_{\overline{\calA}})^{-1}\right)_{xx}.\label{eq:Green}
\end{equation}
As a side remark, \eqref{eq:Green} is the average number of visits of $x$ that \cyclepopping{} does before accepting any cycle and before hitting $\calA$, as can be seen by summing the probability in \eqref{eq:Green_k}.
\subsection{Putting it all together}
\label{sec:all_together}

Armed with the Green generating function \eqref{eq:Green}, we continue the computation of \eqref{eq:fund_event_markovied}, which becomes
\begin{align}
    \sum_{\bfm \in \mathbb{N}^{k_i-1}} & \mathbb{P}(\calE_{\bfm}^i) \nonumber\\
    &= \alpha(c_i) \prod_{j=0}^{k_i-2} p_{x_j^i x^i_{j+1}} G(1,x^i_j,x^i_j; \calV_{i-1}\cup\{x^j_0, \dots, x_{j-1}^i\}), \label{eq:last_step_before_determinantal_assumption}\\
    &= \alpha(c_i) \left[\deg(x^i_0)\left(\mathsf{\Delta}_{\overline{\calV_{i-1}}}\right)^{-1}_{x_0^i x_0^i}\right]  p_{x_0^i x_1^i}\quad\times \dots \times\nonumber\\ 
    &\quad\times \left[\deg(x_{k_i-2})\left(\mathsf{\Delta}_{\overline{\calV_{i-1}\setminus\{x_0^i, \dots, x_{k_i-3}^i\}}}\right)^{-1}_{x_{k_i-2} x_{k_i-2}}\right]  p_{x_{k_i-2} x_{k_i-1}}.
    \label{eq:before_telescoping}
\end{align}
for a lasso, and the same equality without $\alpha(c_i)$ for a branch.

Finally, we note that, for $\calA\subset \calV$ and $x \in \calV\setminus \calA$,
\[
    \left((\mathsf{\Delta}_{\overline{\calA}})^{-1}\right)_{x x} = \frac{\det(\mathsf{\Delta}_{\overline{\calA\cup\{x\}}})}{\det(\mathsf{\Delta}_{\overline{\calA}})}.
\]
In particular, \eqref{eq:before_telescoping} is a telescopic product, and 
\[
    \sum_{\bfm \in \mathbb{N}^{k_i-1}} \mathbb{P}(\calE_{\bfm}^i) = w_{x^i_0 x^i_1} \dots w_{x^i_{k_i-2} x^i_{k_i-1}} 
    \times  \frac{ \det \mathsf{\Delta}_{\overline{\calV_{i-1}\cup\{x_0^i, \dots, x_{k_i-2}^i\}}}} { \det \mathsf{\Delta}_{\overline{\calV_{i-1}}} } \times \alpha(c_i),
\]
for a lasso, and the same without $\alpha(c_i)$ for a branch.
Plugging back into the decomposition \eqref{e:decomposition_in_fund_event}, we obtain \eqref{e:correctness_determinantal}, which concludes the proof. 

\subsection{Law of the number of steps to complete \cyclepopping{}}
Let $T$ be the number of steps to complete \cyclepopping{}. 
We characterise the law of $T$ through its moment generating function (MGF). 
\begin{proposition}\label{prop:moments}
    Let $t \in [-1,1]$. If Assumption~\ref{ass:non-trivial} and Assumption~\ref{ass:weak} hold, we have
    \begin{equation}
        \E[t^T] = t^{n}\det\left(\I + (1-t) \mathsf{\Pi} (\I - \mathsf{\Pi})^{-1}\right)^{-1} = t^{n}\frac{\det(\I - \mathsf{\Pi})}{\det(\I - t\mathsf{\Pi})},\label{e:MGF_Pi}
    \end{equation}
    where $\mathsf{\Pi}$ is defined in \eqref{eq:Pi}.
\end{proposition}
\begin{proof}
    In light of the proof of correctness in Section~\ref{sec:proof_correctness}, this result is simply obtained as follows: the weight $w_e$ of each edge followed by the SRW is formally replaced by $tw_e$. This accounts for one power of $t$ for each step taken by the SRW. 
    Then, we have
    \begin{align}
        \E[t^T] &= \sum_{\substack{\calF \text{ CRSF}\\ \text{oriented}}} \frac{\prod_{e\in \calF} tw_{e} \times \prod_{c \in \calF}(1- \cos\theta(c))}{\det(\mathsf{D})\det(\I - t \mathsf{D}^{-1}(\mathsf{W}\odot \mathsf{\Phi}))}\nonumber\\
        &=\sum_{\substack{\calU \text{ CRSF}}} \frac{\prod_{e\in \calU} tw_{e} \times \prod_{c \in \calU}2(1- \cos\theta(c))}{\det(\mathsf{D} - t \mathsf{W}\odot \mathsf{\Phi})}\nonumber\\
        &= \frac{t^n}{\det(\mathsf{D} - t \mathsf{W} \odot \mathsf{\Phi})} \det(\mathsf{D} - \mathsf{W} \odot \mathsf{\Phi}) = \frac{t^n}{\det(\I - t \mathsf{\Pi})} \det(\I - \mathsf{\Pi}),\label{eq:ratio_det}
    \end{align}
    where we used the generalized matrix-tree theorem of \citet{FORMAN199335} at the next-to-last equality, which is simply the normalization of \eqref{eq:proba_CRSF}.
    Note that $\I + \mathsf{\Pi}$ is invertible in the light of Remark~\ref{rem:ass2} and Proposition~\ref{prop:Pi}.
    This completes the proof.
\end{proof}
\begin{proposition}\label{prop:mean_var_using_cumulants}
    Under Assumption~\ref{ass:non-trivial} and Assumption~\ref{ass:weak}, it holds that 
    \[
        \E[t^T] = t^{n} \exp \sum_{k=1}^{+\infty} \frac{(1-t)^k}{k} \Tr\left(\mathsf{\Pi}(\I -\mathsf{\Pi})^{-1} \right)^k.
    \]
    In particular, we have the following identities
        \[
            \E[T] = n + \Tr\left(\mathsf{\Pi}(\I -\mathsf{\Pi})^{-1} \right) 
            \text{ and }
            \var[T] = \Tr\left(\mathsf{\Pi}(\I -\mathsf{\Pi})^{-1}\right) + \Tr\left(\mathsf{\Pi}(\I -\mathsf{\Pi})^{-1} \right)^2,
        \]
        as well as a formula for the expected parity $\E[(-1)^T] = \det(\mathsf{\Pi} - \I)/\det(\mathsf{\Pi} + \I)$ of the number of steps.
\end{proposition}
\begin{proof}
    These identities are merely obtained by differentiating the cumulant generating function $\log \E[e^{\alpha T}] = \sum_{k\geq 1} \kappa_k(T)\alpha^k / k!$ with respect to $\alpha$ such that $\alpha <0$.
    This generating function is obtained by taking the logarithm of \eqref{e:MGF_Pi} evaluated at $t=e^\alpha$, which gives $\log \E[e^{\alpha T}] = n \alpha - \Tr \log\left(\mathbb{I} + (1-e^\alpha)\mathsf{\Pi}(\I -\mathsf{\Pi})^{-1}\right)$.
    Lastly, the formula for $\E[(-1)^T]$ is obtained by taking $t=-1$ in Proposition~\ref{prop:moments}.
\end{proof}
Recall that $\I - \mathsf{\Pi}$ is non-singular in light of Proposition~\ref{prop:Pi}.
Moreover, letting $\mathsf{\Delta_{N}} = \mathsf{D}^{-1/2}\mathsf{\Delta}\mathsf{D}^{-1/2}$ be the so-called \emph{normalized} magnetic Laplacian, we see that $\E[T] = \Tr\mathsf{\Delta_{N}}^{-1}$.
The expected running time is thus intuitively low if the least eigenvalue of the normalized magnetic Laplacian is not close to zero.
The variational formula for this least eigenvalue,
\[
    \lambda_{\min}(\mathsf{\Delta_{N}}) = \min_{\|\mathsf{v}\|_2 = 1} \mathsf{v}^*\mathsf{\Delta_{N}}\mathsf{v} = \min_{\|\mathsf{v}\|_2 = 1} \sum_{x\sim y} w_{xy}\left|\frac{\mathsf{v}_x}{\sqrt{\deg(x)}} - \frac{\phi_{yx}\mathsf{v}_y}{\sqrt{\deg(y)}}\right|^2,
\]
shows that $\lambda_{\min}(\mathsf{\Delta_{N}})$ relates to the inconsistency the $\Uone$-connection graph, see \citep{FanBar22}.
\subsection{Multi-type spanning forests}
A multi-type spanning forest $\calU$ (MTSF) is a spanning graph where each connected component is either a tree or a cycle-rooted tree, as defined by \citet{kenyon2019}, who also proves that the measure on MTSFs given by
\begin{equation}
    \mu_{\mathrm{MTSF}}(\calU) = \frac{q^{\rho(\calU)}}{\det(\mathsf{\Delta} + q \I)}  \Big(\prod_{e\in \calU} w_e\Big)    \prod_{\substack{\text{non-oriented}\\\text{cycle }c \subseteq \calU }}
    \Big(2 - 2\cos \theta(c)\Big).   \label{eq:MTSF_measure} 
\end{equation}  
is determinantal.
This measure \eqref{eq:MTSF_measure} has been used to build sparsifiers of the regularized magnetic Laplacian in \citep{FanBar22}, and we are thus interested in the computational complexity of sampling from \eqref{eq:MTSF_measure}.
As a simple extension of the previous discussions,  Proposition~\ref{prop:mean_var_MTSF} below gives the mean and variance of the time to sample from \eqref{eq:ST_measure} with a variant of \cyclepopping{}, denoted by \cyclepopping{}$_q$ to emphasize\footnote{Note that \cyclepopping{} is simply \cyclepopping{}$_0$.} the dependence of the parameter $q \geq 0$.
It is briefly described below and we refer the reader to \citep{FanBar22}.

As in Section~\ref{sec:SF_case}, define the auxiliary graph $G_r$ where each node of $G$ is connected to an auxiliary root $r\notin\calV$, with an edge weight $q > 0$. 
Then, under Assumption~\ref{ass:non-trivial} and Assumption~\ref{ass:weak}, \cyclepopping{}$_q$ proceeds as \cyclepopping{} on $G_r$ (see Section~\ref{sec:formalization}), by adding a branch or lasso at each iteration, but starting from $\calF_0 = \{r\}$.
This way, the SRWs are also absorbed upon reaching $r$.
Finally, the spanning subgraph of $G_r$ produced by \cyclepopping{}$_q$ is postprocessed by removing $r$ and all the edges of which $r$ is an endpoint.
The resulting subgraph of $G$ is then an MTSF distributed as \eqref{eq:ST_measure}.
\begin{proposition}\label{prop:mean_var_MTSF}
    Let Assumption~\ref{ass:non-trivial} and Assumption~\ref{ass:weak} hold. Let $q>0$ and let $T$ denote the number of steps to sample a MTSF with \cyclepopping{}$_q$ according to \eqref{eq:MTSF_measure}.
    We have
        \[
            \E[T] = \Tr\left((\mathsf{D} + q\I)(\mathsf{\Delta} + q \I)^{-1}\right)
        \]
        and 
        \[
            \var[T] =  \Tr\left(\mathsf{M}^{(\mathsf{\Phi})}_q (\mathbb{I}- \mathsf{M}^{(\mathsf{\Phi})}_q)^{-1}\right) + \Tr\left(\mathsf{M}^{(\mathsf{\Phi})}_q (\mathbb{I}- \mathsf{M}^{(\mathsf{\Phi})}_q)^{-1}\right)^2,
        \]
        with $\mathsf{M}^{(\mathsf{\Phi})}_q = (\mathsf{D}+q\mathbb{I})^{-1}(\mathsf{W}\odot \mathsf{\Phi})$.
\end{proposition}
The proof of Proposition~\ref{prop:mean_var_MTSF} follows the same lines as the proof of Proposition~\ref{prop:mean_var_SF}, and we thus omit it.
\section{Proof of correctness beyond the determinantal case}
\label{s:Poisson}
We now extend the proof that Algorithm~\ref{a:cyclepopping} samples from the measure \eqref{eq:proba_CRSF_non_det}, by removing any assumption on the form of the cycle weight function $\alpha$.
In particular, we do not assume \eqref{eq:determinantal_cycle_weights}, and the probability measure \eqref{eq:proba_CRSF_non_det} is not necessarily determinantal over the graph edges.
Furthermore, we generalize a bit the framework by considering mesures over oriented Cycle Rooted Spanning Forests (oCRSF), which are spanning subgraphs for which each connected component has exactly one oriented cycle or exactly one backtrack. 
We define the measure
\begin{equation}
    \label{eq:proba_CRSF_non_det_oriented}
    \mu_{\mathrm{oCRSF},\alpha}(\calF)
    = \frac{1}{Z_\alpha} \prod_{e \text{ edge of } \calF} p_e \times
    \prod_{\substack{c \text{ oriented cycle }\\{\text{or backtrack of } \calF }}}
    \alpha(c),
\end{equation}
By assigning a weight zero to backtracks and forgetting cycle orientation, we can then obtain a sample from \eqref{eq:proba_CRSF_non_det}.
\citet{KK2017} already proved the correctness of Algorithm~\ref{a:cyclepopping} to sample from \eqref{eq:proba_CRSF_non_det_oriented} using a stacks-of-cards argument that we formalize in Section~\ref{sec:prs} in framework of partial rejection sampling.
But first, by adapting Section~\ref{sec:proof_correctness}, we provide a different proof, more in line with the seminal proof of \cite{Marchal99} for spanning trees.
In particular, a by-product of this new proof is the law of the running time of the algorithm.
\subsection{Loop measures\label{s:loop_measures}}
We follow \cite{LaLi10} in defining so-called \emph{loop measures}, with a slight adaptation to allow for cycles to be accepted with non-zero probability.

For each oriented cycle $c$ of the graph (including backtracks), we have a weight
$
    \alpha(c) \in [0,1],
$
which is interpreted as the probability to accept $c$.
To avoid trivial measures, we take the following assumption which implies Assumption~\ref{ass:non-trivial} when $\alpha(\cdot) = 1 - \cos \theta(\cdot)$.
\begin{assumption}[Non-trivial cycle weight]\label{ass:non-trivial-weights}
    There is at least one cycle $c_\star$ such that $\alpha(c_\star) > 0$. 
\end{assumption}
We first define a measure on based loops by
\begin{equation}
    \mu_\alpha(\gamma) =  
    \begin{cases}
        q(\gamma) \prod_{c\in \cycles(\gamma)}(1 - \alpha(c)) & \text{ if } |\gamma|\geq 2\\
        1 & \text{ if } |\gamma| = 0
    \end{cases},
    \label{eq:mu_alpha}
\end{equation}
where $q(\gamma)$ is the product of the transition probabilities along $\gamma$, as defined in \eqref{eq:loop_weight}, and $\cycles(\gamma)$ is the multiset of popped cycles defined in Notation~\ref{not:pieces_loop}.
Note that $\mu_\alpha$ is not necessarily a probability measure.
An immediate property of the measure \eqref{eq:mu_alpha} is that it behaves well under concatenation.
\begin{lemma}\label{lem:concat}
    Let $\gamma_1$ and $\gamma_2$ be two loops based at $x\in \calV$, and  $\gamma_1\circ \gamma_2$ the concatenation of $\gamma_1$ and $\gamma_2$.
    Then $\mu_\alpha(\gamma_1\circ \gamma_2) = \mu_\alpha(\gamma_1) \mu_\alpha(\gamma_2)$.
\end{lemma}
Let $x\in \calV\setminus \calA$ with $\calA\subset \calV$, define the Green generating function by
\[
    G_\alpha(t,x,x;\calA) = \sum_{\substack{\gamma \in\BLoops(\overline{\calA},x)}} t^{|\gamma|}\mu_\alpha(\gamma), \quad t\in (0,1],
\]
where the sum goes over loops based at $x$ and included in the complement of $\calA$.
Under Assumption~\ref{ass:weak}, this definition reduces to the Green function defined in Section~\ref{s:Green_generating_fct} when $\alpha(\cdot) = 1 - \cos \theta(\cdot)$.
\begin{lemma}\label{lem:Green_function_loops}
    Let $\calA\subset \calV$ be a subset of nodes and let $x\in \calV\setminus \calA$. 
    Under Assumption~\ref{ass:non-trivial-weights}, $
    G_\alpha(t,x,x;\calA) < +\infty
    $.
\end{lemma}
\begin{proof}
    Consider first the case $\calA = \emptyset$ and $t=1$.
    Let $c_\star$ such that $\alpha(c_\star) > 0$. 
    By \eqref{eq:Green_k}, it comes
    \begin{align*}
        G_\alpha(1,x,x;\emptyset) &= \sum_{k=0}^{+\infty}\mathbb{Q}_{x}\Big(X_{k} = x \text{ and } k < \min_{ \text{ cycle }c} \tau_{\circlearrowleft c}\Big)\\
        &\leq \sum_{k=0}^{+\infty}\mathbb{Q}_{x}\Big(X_{k} = x \text{ and } k <  \tau_{\circlearrowleft c_\star}\Big).
    \end{align*}
    The quantity on the right-hand side equals
    \[
        1 + \mathbb{Q}_{x}\Big( \tau^{ (\text{after }1)}_{\to \{x\}} <  \tau_{\circlearrowleft c_\star}\Big)
    \]
    which is finite. 
    To conclude, it is readily checked that 
    $$
    G_\alpha(t,x,x;\calA) \leq  G_\alpha(1,x,x;\calA) \leq G_\alpha(1,x,x;\emptyset)
    $$
    since $t\in (0,1]$ and $\BLoops(\overline{\calA},x)$ is a subset of $\BLoops(\mathcal{V},x)$.
\end{proof}
We shall show that non-determinantality does not prevent to use a telescoping argument in the vein of \eqref{eq:before_telescoping}.
\begin{notation}[number of revisits of based point]\label{not:d(gam)}
    For a loop $\gamma$ based at $x$, denote by $d(\gamma)$ the number of times $\gamma$ comes back to $x$ after its start.
\end{notation}
We start with two lemmas. 
For example, $d((x,y,x,y,x)) = 2$.
\begin{lemma}\label{lem:power_sum}
    Let $\gamma \in\BLoops(\overline{\calA},x)$ and $d(\gamma)$ the number of visits of $x$ as in Notation~\ref{not:d(gam)}. 
    Let $t\in(0,1]$.
    Under Assumption~\ref{ass:non-trivial-weights}, for any integer $k\geq 1$, we have
    \begin{equation}
        \sum_{\substack{\gamma \in\BLoops(\overline{\calA},x) : d(\gamma)=k}} t^{|\gamma|}\mu_{\alpha}(\gamma) = \left(\sum_{\substack{\gamma \in\BLoops(\overline{\calA},x) : d(\gamma)=1}} t^{|\gamma|}\mu_\alpha(\gamma)\right)^k.\label{eq:power_sum}
    \end{equation}
\end{lemma}
\begin{proof}
    This follows from distributing the $k$ sums on the right-hand side and using Lemma~\ref{lem:concat}.
\end{proof}
To understand the impact of forgetting the base point of a based loop, we need to be precise about sets of unbased loops and their cardinality.
Let $\Loops$ be the set of (unbased) loops in the graph $G$.
The marginal measure on $\Loops$ corresponding to \eqref{eq:mu_alpha} is
\begin{equation}
     \mu_\alpha([\gamma]) = \sum_{\gamma\in [\gamma]}\mu_\alpha(\gamma) =  \mathcal{N}_{[\gamma]} \mu_\alpha(\gamma),\label{eq:measure_forget_base}
\end{equation}
where the multiplicative factor $\mathcal{N}_{[\gamma]}$ is the number of representatives of the equivalence class $[\gamma]$.
Note that $\mathcal{N}_{[\gamma]}$ is not necessarily equal to $|\gamma|$.
For instance, if we take a loop $\gamma = [(x,y,x,y,x)]$ made of two identitical backtracks, we have $\mathcal{N}_{[\gamma]} = 2$ by counting the number of cyclic permutations. However, the length of this loop is $|[\gamma]| = 4 = |\gamma|$.

To clarify the meaning of $\mathcal{N}_{[\gamma]}$, we introduce the following definition: a based loop is \emph{primitive} if it cannot be written as $\gamma_0^m$ for any integer $m\geq 2$ and any other based loop $\gamma_0$.
Also, we say that an unbased loop is primitive if one of its based representatives is primitive.
For any representative $\gamma$ of a loop $[\gamma]$, we have $\gamma = \gamma_0^m$ for some $m\geq 1$ and a primitive based loop $\gamma_0$, possibly equal to $\gamma$. 
Upon defining $\mult(\gamma) = m$, we have  
\begin{equation}
    \mathcal{N}_{[\gamma]} = |\gamma|/\mult(\gamma),\label{e:nb_rep}
\end{equation}
i.e., the number of representatives of $[\gamma]$ is the length of $\gamma_0$.
\begin{lemma}[Number of $x$-based representatives of an unbased loop]\label{lem:number_loops}
     Let $[\gamma]$ be an unbased loop containing $x$.
     Denote by $d_x([\gamma])$ the number of visits of $x$ after its start by any based representative of $[\gamma]$ based at $x$. 
     Let $\#_{[\gamma],x}$ be the number of \emph{representatives of $[\gamma]$ based at $x$}.
      For  any representative $\gamma$ of $[\gamma]$ that is based at $x$, we have
     \begin{itemize}
        \item[$(i)$]  $d_x([\gamma]) = d(\gamma)$  and
        \item[$(ii)$] $\#_{[\gamma],x} = \mathcal{N}_{[\gamma]}d_x([\gamma])/|\gamma|$.
     \end{itemize}   
\end{lemma}
\begin{proof}
    Let $\gamma$ be a representative of $[\gamma]$ based at $x$ and recall $d(\gamma)$ is the number of times $\gamma$ comes back to $x$ after its start as in Notation~\ref{not:d(gam)}.
    The proof of $(i)$ is straightforward.
    To show $(ii)$, we consider two cases.
    $(a)$ If $\gamma$ is not primitive, i.e., if $\gamma$ can be written has $\gamma_0^{m}$ for some $m\geq 2$ and a primitive based loop $\gamma_0$, there is only one \emph{representative of $[\gamma]$ based at $x$}, namely $\#_{[\gamma],x} = 1$.
    Furthermore, the number of visits of $x$ by the based loop $\gamma$ after its start is  $d_x([\gamma]) = m$, whereas the number of representatives of $[\gamma]$ is $\mathcal{N}_{[\gamma]} = |\gamma| / m $; see \eqref{e:nb_rep}.
    Thus, by a direct substitution, we also have $\mathcal{N}_{[\gamma]}d_x([\gamma])/|\gamma| = 1$, and this yields $(ii)$.
    $(b)$ Otherwise, $\gamma$ is  primitive and we have $|\gamma| = \mathcal{N}_{[\gamma]}$ as a consequence of \eqref{e:nb_rep}, whereas the number of \emph{representatives of $[\gamma]$ based at $x$} is equal to the number of visits of $x$ by $\gamma$ after its start, namely $\#_{[\gamma],x} = d_x([\gamma])$. Recalling that  $\mathcal{N}_{[\gamma]}/|\gamma| = 1$, we have shown $(ii)$.
    This completes the proof.
\end{proof}
\subsection{An exponential formula for the Green generating function}
\label{s:gen_Green}
Define the measure on nontrivial unbased loops
\begin{equation}
    m_\alpha([\gamma]) = \begin{cases}
        \mu_{\alpha}([\gamma])/|\gamma| &\text{ if } |\gamma| \geq 2\\
         0 &\text{ if }  |\gamma| = 0
    \end{cases}.\label{eq:exact_measure}
\end{equation}
As anticipated in Section~\ref{sec:intro}, the expression \eqref{eq:exact_measure} can be simplified as follows: using \eqref{e:nb_rep}, we have that $m_\alpha([\gamma]) = \mu_\alpha(\gamma)/\mult(\gamma)$, where $\gamma$ is any representative of $[\gamma]$.
The presence of loop multiplicity is emphasized in \citet[section 2.2]{lejan:hal-03655583}.

Lemma~\ref{lem:power_sum} and Lemma~\ref{lem:number_loops} imply that the Green generating function can be expressed as an exponential involving this measure $m_\alpha$; see Lemma~\ref{lem:logG} below which is a key ingredient for the proof of correctness in Section~\ref{sec:proof_beyond_det}.
\begin{proposition}[Green exponential formula]\label{lem:logG}
    Let $t\in (0,1]$ and $\mathcal{A}\subset \mathcal{V}$. Let $x\in\mathcal{V}\setminus\calA$. 
    Denote by $\Loops(\bar{\calA})$ the set of unbased loops in the complement of a subgraph $\calA$ of $G$.
    Under \ref{ass:non-trivial-weights},  
    $$
        G_\alpha(t,x,x;\calA) =  \exp\left(\sum_{[\gamma]\in \Loops(\overline{\calA}) : x\in [\gamma]} t^{|\gamma|}m_\alpha([\gamma])\right).
    $$
\end{proposition}
Proposition~\ref{lem:logG} may be seen as a generalization of \citep[Lemma 9.3.2]{LaLi10} to measures involving cycle weights\footnote{A similar result was announced in \citep[Section 2.5]{Kassel15} but remained unpublished (personal communication by A. Kassel).}; the proof technique is similar.
\begin{proof}
    We begin by considering the sum of the measures of loops based at $x$ that visit $x$ once after their start.
    The quantity
    \begin{equation}
        \sum_{\substack{\gamma \in\BLoops(\overline{\calA},x) : d(\gamma)=1}} \mu_\alpha(\gamma) = 
        \mathbb{Q}_{x}\Big(\tau^{(\text{after }1)}_{\to \{x\}} < \tau_{\to \calA} \wedge \min_{ \text{ cycle }c} \tau_{\circlearrowleft c} \Big) < 1 \label{eq:proba_back_once}      
    \end{equation}
    is the probability of coming back to $x$ before hitting $\mathcal{A}$ and before accepting any cycle.
    Now, with the help of \eqref{eq:power_sum}, we calculate the sum of the measures of loops based at $x$ that visit $x$ exactly $k$ times after their start.
    By summing \eqref{eq:power_sum} over all possible numbers $k$ of visits, including the case $k=0$ for which no loop occurs, we find
    \begin{equation}
        G_\alpha(t,x,x;\calA) = \left( 1 - \sum_{\substack{\gamma \in\BLoops(\overline{\calA},x) : d(\gamma)=1}} t^{|\gamma|}\mu_\alpha(\gamma)\right)^{-1}.\label{e:Green_inverse_success_proba}
    \end{equation}
    Armed with this identity, we now give the expression of $\log G_\alpha(t,x,x;\calA)$.
    Recalling that $
    \mu_\alpha([\gamma]) = \mathcal{N}_{[\gamma]} \mu_\alpha(\gamma)$ by definition, and using Lemma~\ref{lem:number_loops} for turning a sum over unbased loops into a sum over its based representatives, we find
    \begin{align*}
        \sum_{\substack{[\gamma]\in \Loops(\overline{\calA}) \\ x\in [\gamma]}} t^{|\gamma|}\frac{\mu_\alpha([\gamma])}{|\gamma|} 
        =
        \sum_{\substack{\gamma \in\BLoops(\overline{\calA},x) }} t^{|\gamma|}\frac{\mu_\alpha(\gamma)}{d(\gamma)} 
        &= 
        \sum_{k = 1}^{+\infty} \frac{1}{k}\sum_{\substack{\gamma \in\BLoops(\overline{\calA},x)\\ d(\gamma) = k }}  t^{|\gamma|}\mu_\alpha(\gamma)\\
        &= 
        \sum_{k = 1}^{+\infty} \frac{1}{k}
         \left(\sum_{\substack{\gamma \in\BLoops(\overline{\calA},x) \\ d(\gamma)=1}} t^{|\gamma|}\mu_\alpha(\gamma)\right)^k.
    \end{align*}
    Note that we used \eqref{eq:power_sum} to establish the last equality.
    The result follows by using the Taylor series $ \log(1 - s)^{-1} = \sum_{k=1}^{+\infty} s^k/k$ for $s\in [0,1)$.
\end{proof}
\subsection{Putting it all together \label{sec:proof_beyond_det}}
In the proof of Proposition~\ref{prop:correctness}, we first used the form of the weights when rewriting \eqref{eq:last_step_before_determinantal_assumption} as \eqref{eq:before_telescoping}. 
We thus start from \eqref{eq:last_step_before_determinantal_assumption}, and note that, at that stage, we had in particular proven the following. 
Let $\calF$ in the support of \eqref{eq:proba_CRSF_non_det_oriented}, and $x_1, \dots, x_n$ enumerate the nodes of $G$, in the unique possible order they were added by Algorithm~\ref{a:cyclepopping} to construct $\calF$, see Remark~\ref{rem:shape}.
Then plugging \eqref{eq:last_step_before_determinantal_assumption} into \eqref{e:decomposition_in_fund_event}, we have
\begin{equation}
    \mathbb{P}(\calT = \calF) =  p_{x_1 x_2} \dots p_{x_{n-1} x_n}  \prod_{i=1}^{n}G_\alpha(1,x_i,x_i;\mathcal{B}_i)  \prod_{c \in \calF} \alpha(c),
    \label{e:reprise}
\end{equation}
where we introduced the growing size subsets of nodes $\mathcal{B}_1 = \emptyset$, $\mathcal{B}_2 = \{x_1\}$, $\mathcal{B}_3 = \{x_1,x_2\}$, \dots,  $\mathcal{B}_{n} = \{x_1, \dots, x_{n-1}\}$.
Note that $G_\alpha(1,x_n,x_n;\mathcal{B}_n) = 1$.

Consider now the mutually disjoint sets of loops
$\rmE_1 = \{[\gamma] \text{ s.t. } x_1\in [\gamma]\}$,
$\rmE_2 = \{[\gamma] \text{ s.t. } x_1\notin [\gamma] \text{ and } x_2\in [\gamma]\}$
$\rmE_3 = \{[\gamma] \text{ s.t. } x_1\notin [\gamma] \text{ and } x_2\notin [\gamma]\text{ and } x_3\in [\gamma]\}$, \dots, whereas $\rmE_n$ is the trivial loop containing $x_n$.
In the more compact notation introduced in Section~\ref{s:loop_measures}, for $1\leq i \leq n $, we let
\begin{equation}
    \rmE_i = \left\{[\gamma]\in \Loops(\overline{\mathcal{B}_i}) \text{ such that } x_i\in [\gamma]\right\} \text{ with }\mathcal{B}_i = \{x_1,\dots, x_{i-1}\}.\label{eq:partition_of_unbased_loops}
\end{equation}
In particular, $\rmE_1, \dots, \rmE_{n}$ are a partition of the entire set of loops $\Loops$.
For simplicity, by slightly generalizing \eqref{eq:exact_measure}, we define another measure over loops including the dependence on $t\in(0,1]$, as
\[
    m_{\alpha,t}([\gamma]) = t^{|\gamma|}m_{\alpha}([\gamma]),
\]
and we write $m_{\alpha,t}(\rmE) = \sum_{[\gamma]\in \rmE} m_{\alpha,t}([\gamma])$ for any subset of unbased loops $\rmE$.
Using $n$ times the exponential formula of Lemma~\ref{lem:logG}, and the additivity of the measure $m_{\alpha,t}$, it comes
\begin{equation}
    G_\alpha(t,x_1,x_1;\mathcal{B}_1) \dots G_\alpha(t,x_{n},x_{n};\mathcal{B}_{n})= e^{m_{\alpha,t}(\rmE_1) + \dots + m_{\alpha,t}(\rmE_{n})} = e^{m_{\alpha,t}(\Loops)}.
    \label{eq:chain_rule}
\end{equation}
Equation~\eqref{eq:chain_rule} generalizes the telescoping product of determinants of Section~\ref{sec:all_together}.
In particular, for $t=1$, \eqref{eq:chain_rule} applied to \eqref{e:reprise} shows that $\mathbb{P}(\calT=\calF)$ is \eqref{eq:proba_CRSF_non_det_oriented}, and that the normalizing factor in \eqref{eq:proba_CRSF_non_det_oriented} is
\begin{equation}
    Z_\alpha^{-1} = e^{m_\alpha(\Loops)}.\label{eq:proba_CRSF_non_det_full}
\end{equation}    
This concludes the proof of correctness.
\subsection{Interpretation as a Poisson point process on unbased oriented loops \label{sec:Poisson_over_loops}}
The proof of Proposition~\ref{prop:moments} applies \emph{mutatis mutandis} for general weights, generalizing its result to the identity
\begin{equation}
    \mathbb{E}[t^T] = t^n e^{\sum_{\substack{[\gamma]\in \Loops}}(t^{|\gamma|} - 1)\frac{\mu_{\alpha}([\gamma])}{|\gamma|}} \text{ for } t\in (0,1].\label{eq:ratio_det_Poisson}
\end{equation}
Interestingly, \eqref{eq:ratio_det_Poisson} is a Laplace transform w.r.t.\ a Poisson point process.
Indeed, under the (inhomogeneous) Poisson process on $\Loops$ with intensity $m_\alpha$, for any non-negative function $f$ on $\Loops$,
\begin{equation}
    \E_{\calX\sim\text{\rm Poisson}(m_\alpha,\Loops)}e^{-\sum_{[\gamma]\in \calX}f([\gamma])} = e^{\sum_{\substack{[\gamma]\in \Loops}}\left(e^{-f([\gamma])}-1\right)m_\alpha([\gamma])},
    \label{e:laplace}
\end{equation}
see\footnote{
    We note that a result analogous to \eqref{e:laplace} is obtained in \citep[Eq 4.1]{LeJan11} for continuous loop measures, for which each node also comes with a continuous holding time.
} e.g. \citep[Thm 3.9]{LP2017}.
Letting $t\in(0,1]$ and applying \eqref{e:laplace} to 
$
    f:[\gamma] \mapsto \log(t^{-|\gamma|})
$
yields
\[
    \E_{\calX\sim\text{\rm Poisson}(m_\alpha,\Loops)} \left[t^{n + \sum_{[\gamma]\in \calX} | \gamma |}\right] = \mathbb{E}[t^T],
\]
where the right-hand side is given by \eqref{eq:ratio_det_Poisson}.
This justifies our initial claim in \eqref{eq:law_of_T_Poisson}, which we rephrase here.
\begin{corollary}[law of the running time of \cyclepopping{}]\label{corol:law_of_T}
    Let $T$ be the number of (Markov chain) steps to complete Algorithm~\ref{a:cyclepopping}.
    Then,  under Assumption~\ref{ass:non-trivial-weights}, 
    \begin{equation*}
        T \stackrel{\text{(law)}}{=} n + \sum_{\substack{[\gamma]\in \calX}}|\gamma| \text{ with } \calX\sim\text{\rm Poisson}(m_\alpha, \Loops),
    \end{equation*}
    where $m_\alpha$ is the measure over loops given in \eqref{eq:exact_measure}.
\end{corollary}
In the light of Corollary~\ref{corol:law_of_T}, we now briefly discuss the complexity of \cyclepopping{}, and compare it to a well-known algebraic algorithm for sampling determinantal point processes, due to \cite*{HKPV06} and usually referred to as the HKPV algorithm, after the initials of the authors.
\begin{itemize}
    \item 
    On the one hand, the law of $T$ given in Corollary~\ref{corol:law_of_T} indicates that if the cycles are very consistent, i.e., $\cos \theta(c)\approx 1$, the intensity $m_\alpha([\gamma])$ of the Poisson process of loops is large.
    In other words, the number of popped loops is expected to be large, which tends to slow down \cyclepopping{}.
    In the determinantal case, we can derive from Proposition~\ref{prop:mean_var_using_cumulants} the simple upper bound $\mathbb{E}[T] \leq n/ \lambda_{\min}(\mathsf{\Delta_{N}})$ with $\mathsf{\Delta_{N}} = \mathsf{D}^{-1/2}\mathsf{\Delta}\mathsf{D}^{-1/2}$.
    It is interesting to observe at this point the appearance of the least eigenvalue of the normalized magnetic Laplacian, which plays a central role in the Cheeger inequality \citep{Bandeira} and controls the frustration of the connection graph.
    \item 
    On the other hand, as discussed in \citep{FanBar22}, if the cycles are very consistent, the least eigenvalue of the Laplacian $\lambda_{\min}(\mathsf{\Delta})$ -- or of its normalized version $\lambda_{\min}(\mathsf{\Delta_{N}})$ -- is expected to be close to zero. 
    In this case, the \textsc{HKPV} sampler is expected to be error prone since the computation of the correlation kernel of the determinantal process essentially necessitates to solve a linear system of the type $\mathsf{\Delta} \bm{x} = \bm{b}$. 
    This matrix being ill-conditioned, we expect numerical errors on the correlation kernel to affect \textsc{HKPV}.
    In particular, the output of HKPV with a slightly perturbed kernel might not be a CRSF, in contrast with \cyclepopping{}, which outputs a CRSF by construction.
\end{itemize}

As we discussed under Proposition~\ref{prop:law_of_T}, the characterization by \cite{Marchal99} of the law of the running time of Wilson's algorithm for spanning trees has practical consequences on the choice of the root.
We now examine a similar consequence of Corollary~\ref{corol:law_of_T} for cycle-rooted spanning forests.
\subsection{Expectation of the number of steps and cycle times}
We formalize here the intuition that \cyclepopping{} is fast when cycles of large weights are well-spread in the graph.
In view of \eqref{eq:ratio_det_Poisson}, the expectation of the number of steps to complete \cyclepopping{} is
\begin{equation}
    \mathbb{E}[T] = \sum_{x\in \calV} G_\alpha(1,x,x;\emptyset).\label{e:sum_G_xx}
\end{equation}
Expression \eqref{e:sum_G_xx} -- which naturally reduces to $\E[T] =  \Tr((\I -\mathsf{\Pi})^{-1})$ in the determinantal case --  has a clear probabilistic interpretation as the sum over all the nodes $x$ of the inverse of the probability that \emph{the walker starting from $x$ accepts a cycle no later than its first return to $x$}, or in other words, an \emph{escape-to-a-cycle} probability.
This is in complete analogy with the rooted spanning tree case of Section~\ref{s:CyclePopping_for_STs} where the role of the roots is now played by the inconsistent cycles of the graph.
\begin{proposition}[Escape-to-a-cycle probability]\label{prop:inverse_cycle probability}
    Let $x\in \calV$. Under the assumptions of Corollary~\ref{corol:law_of_T},
    it holds that $G_\alpha(1,x,x;\emptyset) = 1/\rho_x$, where
    \begin{align}
        \rho_x &= \mathbb{Q}^x\left( \min_{c}\tau_{\circlearrowleft c} \leq \tau^{ (\text{after }1)}_{\to \{x\}}\right)
        = 1 - \sum_{\substack{\gamma \in\BLoops(\calV,x) : d(\gamma)=1}}\mu_\alpha(\gamma).
    \end{align}
    Here, $\BLoops(\calV,x)$ is the set of based loops based at $x$, $d(\gamma)$ is the number of visits of the base point of $\gamma$ after its start, and $\mu_\alpha(\gamma)$ is the probability \eqref{eq:mu_alpha} that the SRW $(X^x_n)$ walks along the based loop $\gamma$ and pops all cycles in $\cycles(\gamma)$.
\end{proposition}
An intuitive consequence of \eqref{e:sum_G_xx} and Proposition~\ref{prop:inverse_cycle probability} is that \cyclepopping{} is expected to be fast in a $\Uone$-connection graph where the set of inconsistent cycles can be \emph{hit} from any node in the graph with a large probability.
\begin{proof}
    By Definitions \eqref{eq:Green_k} and \eqref{e:G_as_a_sum_over_k}, and dominated convergence, the Green function reads
    \[
        G_\alpha(1,x,x;\emptyset) =  \sum_{k=0}^{+\infty} \mathbb{Q}^x\left(X_{k}=x\text{ and } k < \min_{c}\tau_{ \circlearrowleft c}\right) = \mathbb{E} R_x,
    \]
    where 
    $$
        R_x = \sum_{k=0}^{+\infty} 1_{\left\{X^x_{k}=x\text{ and } k < \min_{c}\tau^x_{\circlearrowleft c}\right\}}
    $$
    is the number of returns of $(X_k^x)$ to its starting point before a cycle is accepted.
    At this point, we recall from Section~\ref{s:defining_the_markov_chain} the definition of the Markov chain 
    $$
        Z^x_{n} = (X^x_{n}, B^x_{c_1,n}, \dots, B^x_{c_{d}, n}).
    $$
    Remark that $\mathbb{Q}^x(  \min_{c}\tau_{\circlearrowleft c} > \tau^{(\text{after }1)}_{\to \{x\}})$ is the probability of first return before a cycle is accepted, i.e., a special case of \eqref{eq:proba_back_once} with $\mathcal{A}= \emptyset$.
    Denote the probability of the complementary event by $\rho_x = \mathbb{Q}^x(  \min_{c}\tau_{\circlearrowleft c} \leq \tau^{(\text{after }1)}_{\to \{x\}})$. 
    Using the Markov property for $Z^x_{n}$, we have the following law
    \[
        \mathbb{P}(R_x = k) = (1 - \rho_x)^{k} \rho_x,\quad k=0,1,\dots,
    \]
    so that $R_x$ is a geometric random variable with success parameter $ \rho_x $, whose expectation is $ (1-\rho_x)/\rho_x$.
    Thus, we find $G_\alpha(1,x,x;\emptyset) = 1 + (1-\rho_x)/\rho_x = 1/\rho_x$.
    Now, by using \eqref{e:Green_inverse_success_proba} with $\calA = \emptyset$ and $t=1$, we also have 
    \[
        G_\alpha(1,x,x;\emptyset) = \left( 1 - \sum_{\substack{\gamma \in\BLoops(\calV,x) : d(\gamma)=1}}\mu_\alpha(\gamma)\right)^{-1}.
    \]
    Recalling the definition of $q(\gamma)$ in \eqref{eq:loop_weight} in terms of the product of the transition probabilities along the edges of $\gamma$, it is easy to see that 
    \begin{align*}
        \mathbb{Q}^{x}\left(\tau^{(\text{after }1)}_{\to \{x\}}  <  \min_{c}\tau_{\circlearrowleft c}\right) &= \sum_{\substack{\gamma \in\BLoops(\calV,x) : d(\gamma)=1}}q(\gamma) \prod_{c\in \cycles(\gamma)}(1 - \alpha(c))\nonumber\\
         &= \sum_{\substack{\gamma \in\BLoops(\calV,x) : d(\gamma)=1}}\mu_\alpha(\gamma),
    \end{align*}
    in the light of the definition of $\mu_\alpha(\gamma)$ in \eqref{eq:mu_alpha}.
    This completes the proof.
\end{proof}
\subsection{A verbose \cyclepopping{} also samples a Poisson point process on loops}
In this subsidiary section, we seize the opportunity to write down the details of an intriguing construction that \cite[Section 8.3]{LeJan11} defined with fewer details and in the absence of $\Uone$-connection; see also \citep[Remark 18]{lejan:hal-03655583} for continuous loops. 
This is not directly related to the analysis of the running time of Wilson's algorithm, but it is more of a side observation that, as hinted by Corollary~\ref{corol:law_of_T} and upon suitably randomizing the popped based loops, \cyclepopping{} also implicitly gives a sample of a Poisson process of unbased loops.
The key to understand this construction is how the based loops popped by \cyclepopping{} are randomly partitioned to yield a Poisson sample with intensity \eqref{eq:exact_measure}.

Consider a sample output of Algorithm~\ref{a:cyclepopping} and let $x_1, \dots, x_n$ be the vertices of the graph as they are visited by the algorithm.
The loops erased by Algorithm~\ref{a:cyclepopping} form a tuple of -- possibly empty -- based loops
\[
    (\gamma_{1}, \gamma_{2},  \dots, \gamma_{n}) \in \BLoops(\overline{\mathcal{B}_1},x_1) \times \BLoops(\overline{\mathcal{B}_2},x_2) \times \dots \times \BLoops(\overline{\mathcal{B}_{n}},x_{n})
\]
where $\mathcal{B}_1 = \emptyset$, $\mathcal{B}_2 = \{x_1\}$, $\mathcal{B}_3 = \{x_1,x_2\}$, \dots, $\mathcal{B}_{n} = \{x_1, \dots, x_{n-1}\}$ in the notations of \eqref{e:reprise}.
Note that the $n$-th entry $\gamma_n$ is always the trivial loop containing only $x_n$, which satisfies $\mu_\alpha(\gamma_n) = 1$.
It is easy to see that collecting the popped loops and the output of Algorithm~\ref{a:cyclepopping} defines a random pair $(\Gamma,\calT)$, where 
\begin{align}
    \Gamma = (\Gamma_1, \dots, \Gamma_{n})\label{eq:Gamma_tuple}
\end{align}
is a random tuple of loops in  $\BLoops(\overline{\mathcal{B}_1},x_1) \times \dots \times \BLoops(\overline{\mathcal{B}_{n}},x_{n})$ and $\calT$ is a random CRSF.
Their joint law\footnote{In reference to physics, \citet[Section 8]{LeJan11} describes this coupling as a \emph{supersymmetry}.} is
\begin{align}
    \mathbb{P}(\Gamma = (\gamma_{1}, \dots, \gamma_{n}), \calT = \calF) = \mu_\alpha(\gamma_{1}) \dots \mu_\alpha(\gamma_{n}) \times \prod_{e\in \calF} p_e \prod_{c\in \calF} \alpha(c). \label{eq:coupling}
\end{align}
In analogy with the so-called \emph{verbose mode} in computing, we define a verbose \cyclepopping{} which keeps tracks of more details of the run by outputting the erased loops in addition to a CRSF.
\begin{definition}
    We call \verbosecyclepopping{} the variant of \cyclepopping{} which outputs the random pair $(\Gamma,\calT)$ constituted of a tuple of popped based loops $\Gamma$ as given in \eqref{eq:Gamma_tuple} and a CRSF $\calT$.
\end{definition}
By marginalizing, i.e., by summing the first factor of \eqref{eq:coupling} over all possible based loops, we recover the normalization \eqref{eq:proba_CRSF_non_det_full}, namely,
\begin{align*}
    \sum_{\gamma_1\in \BLoops(\overline{\mathcal{B}_1},x_1)} \dots \sum_{\gamma_n \in \BLoops(\overline{\mathcal{B}_{n}},x_{n})}\mu_\alpha(\gamma_{1}) \dots \mu_\alpha(\gamma_{n})  &=  \prod_{i=1}^{n}G_\alpha(1,x_i,x_i;\mathcal{B}_i)\\
    &= \exp m_\alpha(\Loops) = Z_\alpha^{-1},
\end{align*}
where we used Lemma~\ref{lem:Green_function_loops} to express the Green function as a sum over based loops, and  the notation for the normalization  \eqref{eq:proba_CRSF_non_det_full} of the measure on CRSFs.
In other words, by forgetting the erased loops, we indeed find that Algorithm~\ref{a:cyclepopping} yields $\calF$ with probability $\mathbb{P}(\calT = \calF)$ as given in \eqref{e:reprise}.

Now, by inserting $e^{- m_\alpha(\Loops)} e^{ m_\alpha(\Loops)}=1$ in \eqref{eq:coupling} and by using the partition of the set of unbased loops in \eqref{eq:chain_rule}, we write
\begin{align}
    \mathbb{P}(\Gamma = (\gamma_{1}, \dots, \gamma_{n-1}), \calT = \calF) &= \mu_\alpha(\gamma_{1}) e^{-m_{\alpha}(\rmE_1)} \dots \mu_\alpha(\gamma_{n}) e^{-m_{\alpha}(\rmE_{n})}\times\nonumber\\
    &\times  Z_\alpha^{-1} \prod_{e\in \calF} p_e \prod_{c\in \calF} \alpha(c), \label{eq:coupling_with_based_loops}
\end{align}
where one factor is trivial since $\mu_\alpha(\gamma_{n})=1$ by definition of trivial loops and $m_{\alpha}(\rmE_{n})=0$ by \eqref{eq:exact_measure}.
We shall see that a randomization of $\Gamma$ actually yields a Poisson point process of loops with intensity measure $m_\alpha$.
\begin{proposition}[Coupling between Poisson process and CRSFs]\label{prop:coupling_Poisson_CRSFs}
    Under Assumption~\ref{ass:non-trivial-weights}, let $(\Gamma, \calT)$ be the output of \verbosecyclepopping{}.
    Denote by $\mathcal{L} = \{[\delta_1]^{\nu_1}$, \dots, $[\delta_\ell]^{\nu_\ell}\}$ a set of loops with multiplicity and $\calF$ a CRSF.
    Let $\calX \sim \text{\rm Poisson}(m_\alpha, \Loops)$
    We have 
    \begin{align*}
        &\mathbb{P}(\forgetorder(\forgetbase(\randomsplit(\Gamma))) = \mathcal{L}, \calT = \calF)\\
        &= \mathbb{P}(\calX =\mathcal{L}) \times
        \mu_{\mathrm{CRSF},\alpha}(\calF),
    \end{align*}
    where the measure over CRSFs is given in \eqref{eq:proba_CRSF_non_det_oriented}.
\end{proposition}
The precise definition of $\randomsplit$ is given in the proof of Proposition~\ref{prop:coupling_Poisson_CRSFs}.
\begin{proof}
Consider the first factor in \eqref{eq:coupling_with_based_loops} with the based loop $\gamma_1 \in \BLoops(\overline{\mathcal{B}_1},x_1)$.
To avoid trivialities, we assume $|\gamma_1|\geq 2$.
We now describe how the based loop $\gamma_1$ can be randomly split into a set of unbased loops that provides a Poisson sample of unbased loops in $\rmE_1$, as suggested by \citet[Section 8.3]{LeJan11} or \citet[Section 9.4]{LaLi10} in a slightly different setting.

Suppressing indices for a moment,  say that a loop $\gamma$ based at $x$ visits its base point exactly $n_{x}$ times after its start, i.e., $d(\gamma)= n_{x}$ as defined in Notation~\ref{not:d(gam)}.
Thus, this oriented loop is the concatenation of $n_{x}$ loops,
\begin{equation}
    \gamma=\eta_1 \circ \dots \circ \eta_{n_{x}},\label{eq:decomp_gamma}
\end{equation}
each of them visiting $x$ only once of its start: $d(\eta_i) = 1$ for $1\leq i \leq n_x$. \\

\noindent\emph{Trivial multiplicities.}
For simplicity, we assume that these loops are all distinct; the case of general multiplicities is treated in the next paragraph.
The main idea that we describe below is to define a decomposition of $\gamma$ into loops thanks to a random partition of the set of basic $n_x$ loops which is identified to $\{1, \dots, n_x\}$.

Recall that a partition of $\{1, \dots, n_x\}$ is an unordered collection of disjoint sets, called blocks, such that their union is $\{1, \dots, n_x\}$.
Let $\{\calS_1, \dots, \calS_k\}$ be a partition of $\{1, \dots, n_x\}$. 
Naturally, the block sizes satisfy $|\calS_1|+\dots+|\calS_k|=n_x$.
The blocks can be put in \emph{order of appearance}: the smallest element of $\calS_1$ is smaller than the smallest element of $\calS_2$, etc.
Let $\bm\calS$ be a random partition with law
\begin{equation}
    \mathbb{P}(\bm\calS = \{\calS_1, \dots, \calS_k\}) = \frac{(|\calS_1|-1)!  \dots (|\calS_k|-1)!}{n_x!},\label{eq:partition_proba}
\end{equation}
and denote by $|\bm\calS| = k$ its number of blocks.\footnote{
    The reader may recognize the law on partitions that arises in the so-called \emph{Chinese restaurant process}. 
}
Define the random ordered partition $\bm{\calS^{o}}$ by taking the blocks in order of appearance, and then uniformly permuting the $k$ blocks.
Formally, the probability to obtain a given ordered partition is
\begin{equation}
    \mathbb{P}(\bm{\calS^{o}} = (\calS_1^{o}, \dots, \calS_k^{o})) = \frac{1}{k!}\mathbb{P}(\bm\calS = \{\calS^{o}_1, \dots, \calS^{o}_k\}).
    \label{eq:partition_proba_ordered}
\end{equation}
At this point, we define the ``exchangeable random order'' of block sizes in the terminology of \cite[Section 2]{pitman2006combinatorial}. 
Let $M$ be a random tuple whose entries are the sizes of the blocks of $\bm{\calS^{o}}$.
Its law is obtained by summing \eqref{eq:partition_proba_ordered} over all possible ordered partitions, say of given sizes $(m_1,\dots,m_k)$.
By inspection of \eqref{eq:partition_proba}, all these partitions have the same probability. 
Hence, this summation amounts to multiply \eqref{eq:partition_proba_ordered} by a multinomial factor as follows
\begin{align}
    \mathbb{P}(M = (m_1,\dots,m_k)) &= \frac{n_x!}{m_1! \dots m_k!}\frac{1}{k!} \frac{(m_1-1)!  \dots (m_k-1)!}{n_x!} \nonumber\\
    &= \frac{1}{k!} \frac{1}{m_1 \cdots m_k}.\label{eq:law_of_M}
\end{align}
That \eqref{eq:law_of_M} is well-normalized is a consequence of \citep[Exersise 9.1]{LaLi10}. 
Next, draw $M$ according to \eqref{eq:law_of_M} and concatenate the based loops $\eta_1$, \dots, $\eta_{n_{x}}$ in \eqref{eq:decomp_gamma} to obtain
\[
     \gamma = \overbrace{\delta_1}^{\text{$m_1$ visits of $x$}} \circ \dots \circ \overbrace{\delta_k}^{\text{$m_k$ visits of $x$}},
\]
with number of visits of $x$ respectively equal to $d(\delta_1) = m_1$, \dots, $d(\delta_k) = m_k$; e.g., the first loop being 
$
    \delta_1 = \eta_1 \circ \dots \circ \eta_{m_1}.
$
The corresponding random function reads $\randomsplit: \gamma\mapsto (\delta_1, \dots, \delta_k)$.
Considering a nontrivial based loop $\gamma$, \eqref{eq:law_of_M} gives
\[
    \mathbb{P}(\randomsplit(\gamma) = (\delta_1, \dots, \delta_k) ) = \frac{1}{k!} \frac{1}{m_1 \cdots m_k}.
\]
For any ordered sequence of based loops $(\delta_1, \dots, \delta_k)$, we also define the function 
$$
    \forgetbase: (\delta_1, \dots, \delta_k) \mapsto ([\delta_1], \dots, [\delta_k]),
$$ 
which simply maps any list of based loops to the most of their unbased equivalence class.

Coming back to the measure \eqref{eq:coupling_with_based_loops}, we are ready to write the law of 
$$
\forgetbase(\randomsplit(\Gamma_1)).
$$
Let $\gamma_1$ be a loop based at the common base $x_1$ of every loop in $\Gamma_1$.
Using the definition of the unbased loop measures \eqref{eq:measure_forget_base} and \eqref{eq:exact_measure}, as well as the random tuple of popped based loops \eqref{eq:Gamma_tuple}, we find
\begin{align}
    \mathbb{P}\big(\forgetbase &(\randomsplit(\Gamma_1)) = ([\delta_1], \dots, [\delta_k])\big)\nonumber\\
    &= \mathbb{P}\left( \randomsplit(\Gamma_1) \in {\forgetbase}^{-1}\left(([\delta_1], \dots, [\delta_k])\right) \right) \nonumber\\
    &= \sum_{\delta'_1, \dots, \delta'_k} \frac{ e^{-m_{\alpha}(\rmE_1)}}{k!} \prod_{i = 1}^{k} \frac{\mu_\alpha(\delta_i)}{d(\delta_i)} 1_{\delta'_i \in \BLoops(\calV,x_1)} 1_{[\delta'_i]=[\delta_i]}\nonumber\\    
    &= \frac{ e^{-m_{\alpha}(\rmE_1)}}{k!} \prod_{i = 1}^{k}\#_{[\delta_i],x_1} \frac{\mu_\alpha(\delta_i)}{d(\delta_i)}\nonumber\\ 
    &= \frac{ e^{-m_{\alpha}(\rmE_1)}}{k!} \prod_{i = 1}^{k} m_\alpha([\delta_i]),
    \label{eq:proba_split_before_forgetting_order}
\end{align}
where $\#_{[\gamma],x_1}$ is be the number of \emph{representatives of $[\gamma]$ based at $x_1$}; as given in  Lemma~\ref{lem:number_loops}.
Hence, if $\calX_1\sim\text{\rm Poisson}(m_\alpha, \rmE_1)$, the final probability is obtained by multiplying by the number of permutations of $k$ distinct objects,
\begin{align}
    &\mathbb{P}\Big(\forgetorder(\forgetbase(\randomsplit(\Gamma_1))) = \{[\delta_1], \dots, [\delta_k]\} \Big) \nonumber\\
    &= m_\alpha([\delta_1])\dots m_\alpha([\delta_k])e^{-m_{\alpha}(\rmE_1)}\nonumber\\
     &= \mathbb{P}(\calX_1 = \{[\delta_1], \dots, [\delta_k]\} ).\label{eq:proba_split_after_forgetting_order}
\end{align}
Repeating the construction for each factor in the joint distribution \eqref{eq:coupling_with_based_loops}, we recognize the law of the Poisson process $\text{\rm Poisson}(m_\alpha, \Loops)$.\\

\noindent\emph{General multiplicities.}
    In order to derive the previous results, we assumed that the loops in the decomposition \eqref{eq:decomp_gamma} where pairwisely distinct.
    We can relax this assumption.
    Denote by $\{[\delta_1]^{\nu_1}, \dots, [\delta_\ell]^{\nu_\ell}\}$ a multiset where $[\delta_1]$, \dots, $[\delta_\ell]$ appear respectively $\nu_1$, \dots, $\nu_\ell$ times with $\nu_1 + \dots  + \nu_\ell = k$.
    Now, the change in the derivation happens at the ``$\forgetorder$'' step between \eqref{eq:proba_split_before_forgetting_order} and \eqref{eq:proba_split_after_forgetting_order}, where one has to multiply \eqref{eq:proba_split_before_forgetting_order} by the number of ordered sequences yielding the same unordered sequence, i.e.,
    $
        \frac{k!}{\nu_1! \dots \nu_\ell!}.
    $
    Thus, the corresponding generalization of \eqref{eq:proba_split_before_forgetting_order} reads
    \begin{align*}
        &\mathbb{P}\Big(\forgetorder(\forgetbase(\randomsplit(\Gamma_1))) = \{[\delta_1]^{\nu_1}, \dots, [\delta_\ell]^{\nu_\ell}\} \Big)\nonumber \\ &= \frac{m_\alpha([\delta_1])^{\nu_1}\dots m_\alpha([\delta_\ell])^{\nu_\ell}}{\nu_1 ! \dots \nu_\ell !}e^{-m_{\alpha}(\rmE_1)}\\
        &=\mathbb{P}(\calX_1 = \{[\delta_1]^{\nu_1}, \dots, [\delta_\ell]^{\nu_\ell}\} ).
    \end{align*}
    The last equality can be shown as follows. Denote by $\mathcal{N}_{\calS}(\calX)$ the random number of points (loops) 
    of $\calX \sim \text{\rm Poisson}(m_\alpha, \Loops)$ in the set $\calS\subseteq \Loops$.
    By using the properties of Poisson Point Processes, we have
    \begin{align*}
        &\mathbb{P}\Big(\mathcal{N}_{\{[\delta_1]\}}(\calX_1) = \nu_1 \text{ and } \dots \text{ and } \mathcal{N}_{\{[\delta_\ell]\}}(\calX_1) = \nu_k \text{ and } \mathcal{N}_{\rmE_1\setminus\{[\delta_1], \dots, [\delta_\ell]\}}(\calX_1) = 0\Big)\\
        & = \frac{m_\alpha([\delta_1])^{\nu_1}e^{-m_\alpha([\delta_1])}}{\nu_1!} \dots \frac{m_\alpha([\delta_\ell])^{\nu_k}e^{-m_\alpha([\delta_\ell])}}{\nu_\ell!} e^{-m_\alpha(\rmE_1\setminus\{[\delta_1], \dots, [\delta_\ell]\})},
    \end{align*}
    which yields the desired result since $\{[\delta_1]\} \cup  \dots \cup \{[\delta_\ell]\} \cup \rmE_1\setminus\{[\delta_1], \dots, [\delta_\ell]\} = \rmE_1$.
    This concludes the proof of Proposition~\ref{prop:coupling_Poisson_CRSFs}.
\end{proof}
\subsection{Interpretation of $\cyclepopping{}$ with heaps of cycles \label{sec:heaps}}
One of our initial motivations was to understand the normalization constant \eqref{eq:proba_CRSF_non_det_full} as expressed in somewhat compact form ``$\sum_{\calC} \mathbb{P}\left(\mathrm{ pop }\ \mathcal
{C}\right)$'' in \citep[p938]{KK2017}, where $\calC$ denotes all possible collections of cycles which can be popped.
In this section, we explain how one can fall back onto the expression of \cite{KK2017} using the combinatorial notion of \emph{heaps of pieces} \citep{viennot2006heaps}, which are here simply \emph{heaps of oriented cycles} \citep[Chapter 5b]{ViennotCourse} popped by $\cyclepopping{}$.
\begin{definition}[concurrent cycles]\label{def:concurrent}
    Let $[c]$ and $[c^\prime]$ be two unbased cycles. 
    We say that $[c]$ is concurrent to $[c^\prime]$, namely $[c] \mathcal{R} [c^\prime]$, if $[c]$ and $[c']$ share at least one node. 
\end{definition}
The definition of a heap abstracts the definition of $\cycles(\gamma)$ in Notation~\ref{not:pieces_loop}.
Loosely speaking, \emph{node-disjoint consecutive} cycles commute in a heap.
Although, as emphasized by Viennot, one possible definition follows from the construction of the partially commutative monoid of \citet{cartier1969problemes}, we give below a less formal definition.
\begin{definition}[labelled heap]
    A labelled heap of cycles $H$ is a partially ordered set of unbased cycles endowed with an order index, namely $\{([c_1],k_1)$, \dots, $([c_h],k_h)\}$ with $k_1,\dots, k_h$ positive integers and $\{[c_1], \dots, [c_h]\}$ a multiset of cycles.
    The set of pieces of $H$ \citep[Definition 2.2]{krattenthaler2006theory} is simply the set of distinct cycles in $\{[c_1], \dots, [c_h]\}$.
    The partial order is defined as follows.
    If $[c] \mathcal{R} [c^\prime]$ and $k \leq k^\prime$, we write $([c],k) \preceq_0 ([c^\prime],k^\prime)$.
    This partial order is then extended by transitivity for nonconcurrent cycles.
    First, when $[c] \mathcal{R} [c^\prime]$, then $\preceq$ is given by $\preceq_0$.
    Second, for $[c] \slashed{\mathcal{R}} [c^\prime]$,
    we write $([c],k) \preceq ([c^\prime],k^\prime)$ if there is a sequence $([c_i],k_i)$ for $1\leq i \leq j$ for some integer $j\geq 1$ such that $([c],k) \preceq_0 ([c_1],k_1) \preceq_0 \dots \preceq_0  ([c_j],k_j) \preceq_0  ([c^\prime],k^\prime)$.
    For compactness, we often leave the index of $([c],k)$ implicit and simply write $[c]$.
\end{definition}
The empty heap, associated with the empty set of pieces, is denoted by $\emptyset$.
The following definitions formalize the intuition that several labellings correspond to the same heap.
\begin{definition}[isomorphism of labelled heaps]\label{def:isom_heap}
    Two labelled heaps of cycles $H_1 = \{([c_1],k_1)$, \dots, $([c_h],k_h)\}$ and $H_2 = \{([c_1],\ell_1)$, \dots, $([c_h],\ell_h)\}$ are isomorphic if they share the same multiset of cycles and for all cycles $[c]$ and $[c^\prime]$ such that $[c]\mathcal{R}[c^\prime]$, we have $k \leq k^\prime$ if and only if $\ell \leq \ell^\prime$.
\end{definition}
\begin{definition}[heap]
    A heap of cycles is an equivalence class of labelled heaps of cycles under isomorphism.
\end{definition}
The construction of Viennot also includes a formal definition of \emph{putting a heap on top of another heap} that we briefly sketch here for self-containedness.
Given two heaps $H_1$ and $H_2$, the superposition $H_1 \circ H_2$ of $H_2$ on $H_1$ has for set of pieces the union of the pieces of $H_1$ and $H_2$, and the partial order is given by an intuitive composition rule inheriting the partial order of $H_1$ and $H_2$; see \citep[Definition 2.5]{krattenthaler2006theory} or \citep[Definition 2.7]{viennot2006heaps} for a rigorous definition.
\begin{definition}[pyramid]
    The element $[c]$ of a heap is \emph{maximal} if there is no $[c^\prime]$ such that $[c]\preceq [c^\prime]$. 
     A heap is a pyramid if it contains exactly one maximal piece.
\end{definition}
We refer to Figure~\ref{fig:heap_from_loop} for an illustration of a pyramid where $c_3$ is the maximal piece.
For further use, denote by $\mathcal{P}(\overline{\calA},x)$ the set of pyramids with pieces in the complement of $\calA$ and such that node $x$ belongs to the maximal piece.
Note that an empty heap is not a pyramid. 
It will be convenient to define 
\begin{equation}
    \mathcal{P}_\emptyset(\overline{\calA},x) = \{\emptyset\} \cup \mathcal{P}(\overline{\calA},x),\label{eq:set_pyramids_cup_empty_heap}
\end{equation}
where $\emptyset$ is the empty heap.
\subsubsection{MGF of $T$ as a ratio of sums over heaps}
Define the weight of a heap as 
\begin{equation}
    w(H) = \prod_{[c]\in H} w([c]),\label{eq:weight_of_heap}
\end{equation}
where $w([c])\geq 0$ is a weight function over unbased oriented cycles. 
Furthermore, we set the weight of the empty heap to be $w(\emptyset)=1$.
\begin{proposition}\label{prop:conv_sum_over_heaps}
    Consider a (finite) graph and heaps of cycles with weights as given in \eqref{eq:weight_of_heap}.
    If $|w([c])| < 1$ for every cycle $[c]$, the series $\sum_{H}w(H)$ summing all heap weights is finite.
\end{proposition}
\begin{proof}
    In a finite graph, there is a finite number of oriented cycles $\{[c_1], \dots, [c_d]\}$, and therefore the set of available pieces for heaps has finite cardinality. 
    Each of these oriented cycles can of course appear several times in a heap.
    We now find a finite upper bound on $\left|\sum_{H}w(H)\right|$ by summing over all possible subsets of oriented cycles and over all cycle multiplicities $\nu_1, \dots, \nu_d \geq 0$.
    This overcounting amounts to neglecting the concurrency of cycles.
    As a consequence of the triangle inequality, we have
    \begin{align*}
        \left|\sum_{H}w(H)\right| &\leq \sum_{\nu_1 = 0, \dots, \nu_d = 0 }^{+\infty} |w([c_1])|^{\nu_1} \dots |w([c_d])|^{\nu_d}\\
        &= \left(1 - |w([c_1])|\right)^{-1} \dots \left(1 - |w([c_d])|\right)^{-1}.
    \end{align*}
\end{proof}
In what follows, we are going to set e.g.\ $w([c]) = \mu_\alpha(c)$ with $c$ any based representative of $[c]$ and where the based loop measure is given in \eqref{eq:mu_alpha}.
This is licit since the expression of $\mu_\alpha(c)$ does not depend on the base point.
\begin{proposition}[decomposition of a heap in pyramids]\label{prop:pyramids_from_heap}
    Let $H$ be a heap of cycles in a graph of $n$ nodes.
    Let $x_1, \dots, x_{n}$ be an ordering of the nodes of the graph.
    There exists a unique sequence $(P_{x_1}, \dots, P_{x_{n-1}})$, indexed by the nodes, of heaps with disjoint sets of elements -- which are either the empty heap or a pyramid -- such that $H= P_{x_1} \circ \dots \circ  P_{x_{n-1}}$ and such that if $P_{x_i}$ $(1\leq i \leq n-1)$ is a pyramid then $x_i$ belongs to its maximal piece.
\end{proposition}
\begin{proof}
    The proof simply gives a partition of $H$ as a composition of disjoint pyramids or empty heaps.
    Consider the first node in the ordering $x_1$.  
    If no cycle in the heap contains $x_1$, then $P_{x_1} = \emptyset$.
    Otherwise, denote the maximum element in the heap containing $x_1$ by $[c(x_1)]$, which is necessarily unique up to multiplicity since two elements of the heap containing $x_1$ are always related by the partial order ``$\preceq$'', i.e., cycles containing $x_1$ are concurrent; see Definition~\ref{def:concurrent}.
    There exists a pyramid in the heap $H$ with $[c(x_1)]$ as its maximal element, which includes all $[c^\prime]$ so that $[c^\prime] \preceq [c(x_1)]$.
    Denote by $P_{x_1}$ this pyramid.
    By design, any piece containing $x_1$ should be in this pyramid, since this cycle is related by the partial order to  $[c(x_1)]$.
    Now, in the wording of \citet[proof of Theorem 4.4]{krattenthaler2006theory}, we push  $P_{x_1}$ downwards, namely, we subtract it from $H$ to yield $H_{1}$.
    Next, consider the second node $x_2$ in the ordering. 
    If no cycle of the heap contains $x_2$, then $P_{x_2}$ is the  empty heap.
    Otherwise, let $[c(x_2)]$ be the unique maximum piece (up to multiplicity) containing $x_2$ in $H_{1}$.
    Let $P_{x_2}$ be the pyramid with $[c(x_2)]$ as maximum element and containing all pieces in $H_2$ dominated by $[c(x_2)]$.
    By construction, this pyramid does not contain $x_1$.
    Going on over the nodes in the ordering, we build the desired ordered sequence of pyramids.
\end{proof}
There is a connection between the cycles popped by \cyclepopping{} and heaps of cycles à la Viennot, which was presumably already summarized in an oral presentation of \cite{Marchal_heaps} in the case of Wilson's algorithm, i.e., when all the cycles are popped with probability $1$. 
To our knowledge, a similar connection between pyramids and loops was pointed out by \citet{GISCARD2021112305} and \citet[Appendix A]{helmuth2016loop}.
An illustration is given in Figure~\ref{fig:loops_cycles_pyramid}.
\begin{definition}[chronological labelling]\label{def:chrono_labelling}
    Consider a loop based at $x$. 
    We define the corresponding \emph{labelled heap of cycles with chronological labelling} as follows.
    The oriented loop is visited by starting at $x$ along the oriented edges.
    The first encountered (based) cycle $c$ defines element $([c],1)$, i.e., the first element of the heap. 
    Now, let $([\tilde{c}], \ell)$ be the last added element to the growing heap. 
    Let $c^\prime$ be the next cycle in the visit of the loop. 
    If $c^\prime$ has no common node with any cycle of the heap, then we add the element $([c^\prime], 1)$ to the heap. 
    Otherwise, let $([c_0], \ell_0)$ be the maximum element with a common node with $c^\prime$. Define the index of $c^\prime$ to be $\ell^\prime = \ell_0 + 1$. 
    Put $([c^\prime], \ell^\prime)$ into the heap.
    The construction terminates when the loop is entirely visited. 
\end{definition}
\begin{lemma}[Pyramid and based loops]\label{lem:pyramid_equiv_loop}
    Let $\gamma$ be a nontrivial loop based at $x$. 
    \begin{itemize}
        \item[(i)] The cycles determined by the loop erasure defined in Notation~\ref{not:pieces_loop} form a labelled pyramid where the maximal element contains $x$ and for which the labelling is given by the chronological order of popping in Definition~\ref{def:chrono_labelling}.
        \item[(ii)] A pyramid together with a node $x$ in its maximal element determines a loop based at $x$.
    \end{itemize}
\end{lemma}
\begin{proof}
    Part $(i)$ of the lemma is trivial in the light of Definition~\ref{def:chrono_labelling}. 
    For $(ii)$, we construct the based loop by adding edges iteratively to a growing path.
    We start this path from $x$ and follow the oriented edge emanating from $x$ in a \emph{minimal} cycle containing $x$ (corresponding to the first edge visited by the loop). 
    This cycle is unique up to multiplicity since any two concurrent cycles are ordered.
    The orientation of this cycle determines the edge $xx_1$ which is added to the path; see Figure~\ref{fig:pyramid_with_growing_loop} for an illustration.
    We go on with $x_1$, and we consider the minimal cycle in the pyramid which has not yet been completely visited. The orientation of this cycle determines the edge $x_1x_2$.
     This procedure will eventually visit all the cycles since all concurrent cycles are ordered. The maximal cycle started at $x$ will be visited last.
\end{proof}
\begin{figure}[t]
    \begin{subfigure}[b]{0.8\linewidth}
        \centering
    \begin{tikzpicture}[scale =1.5]
        \def\x{0.7};
        \def\y{0.7};
        \def\s{0.03};

        \node (P11) at (0,0) {\textbullet};
        \node  at (-0.3,0) {$x$};
        \node (P12) at (1,0) {\textbullet};
        \node (P13) at (2,0) {\textbullet};
        \node (P21) at (0+\x,0+\y) {\textbullet};
        \node (P22) at (1+\x,0+\y) {\textbullet};
        \node (P23) at (2+\x,0+\y) {\textbullet};
        \node (P32) at (1+2*\x,0+2*\y) {\textbullet};

        \node[orange]  at (1.5,-0.3) {$c_1$};
        \node[teal]  at (1+\x+0.1,\y +0.4) {$c_2$};
        \node[magenta]  at (0.8,0.4) {$c_3$};

        \node  at ($(P13) + (0.3,0)$) {$x_2$};
        \node  at ($(P12) + (0,-0.3)$) {$x_1$};
        \node  at ($(P22) + (0.3,0)$) {$x_4$};
        \node  at ($(P32) + (0.3,0)$) {$x_5$};
        \node  at ($(P21) + (-0.3,0)$) {$x_7$};

        \node (P31) at (0+2*\x,0+2*\y) {\textbullet};
        \node (P32) at (1+2*\x,0+2*\y) {\textbullet};
        \node (P33) at (2+2*\x,0+2*\y) {\textbullet};
        \draw [very thick, draw=gray, opacity=0.2] (P11.center)--(P12.center);
        \draw [very thick, draw=gray, opacity=0.2] (P12.center)--(P13.center);
        \draw [very thick, draw=gray, opacity=0.2] (P11.center)--(P21.center);
        \draw [very thick, draw=gray, opacity=0.2] (P12.center)--(P22.center);
        \draw [very thick, draw=gray, opacity=0.2] (P13.center)--(P23.center);
        \draw [very thick, draw=gray, opacity=0.2] (P21.center)--(P22.center);
        \draw [very thick, draw=gray, opacity=0.2] (P22.center)--(P23.center);
        \draw [very thick, draw=gray, opacity=0.2] (P21.center)--(P31.center);
        \draw [very thick, draw=gray, opacity=0.2] (P22.center)--(P32.center);
        \draw [very thick, draw=gray, opacity=0.2] (P23.center)--(P33.center);
        \draw [very thick, draw=gray, opacity=0.2] (P31.center)--(P32.center);
        \draw [very thick, draw=gray, opacity=0.2] (P32.center)--(P33.center);
        \draw [magenta,->,very thick] (P11)--(P12);
        \draw [magenta,->,very thick] (P12)-- (P22);
        \draw [magenta,->,very thick] (P22)--(P21);
        \draw [magenta,->,very thick] (P21)--(P11);
        \draw [orange,<->,very thick] (P12)--(P13);
        \draw [teal,<->,very thick] (P22)--(P32);
    \end{tikzpicture}
    \caption{Loop $\gamma = (x,x_1,x_2,x_1,x_4,x_5,x_4,x_7,x)$. }
    \end{subfigure}
    \hfill
    \begin{subfigure}[b]{0.8\linewidth}
        \centering
        \begin{tikzpicture}[scale = 1.5]
            \def\x{0.7};
            \def\y{0.7};
            \def\s{1.};
            \node (P11) at (0,0) {};
            \node  at ($(-0.3,0)+(0,0+\s)$) {$x$};
            \node (P12) at (1,0) {\textbullet};
            \node (P13) at (2,0) {\textbullet};
            \node (P21) at (0+\x,0+\y) {};
            \node (P22) at (1+\x,0+\y) {\textbullet};
            \node (P32) at (1+2*\x,0+2*\y) {\textbullet};
            \node (P21) at (0+\x,0+\y) {};
            \node (P22) at (1+\x,0+\y) {\textbullet};
            \node (P23) at (2+\x,0+\y) {};
            \node (P31) at (0+2*\x,0+2*\y) {};
            \node (P33) at (2+2*\x,0+2*\y) {};
            \node (P11_shift2) at ($(P11)+(0,0+\s)$) {\textbullet};
            \node  at ($(P12)+(0,0+\s) + (0.3,0)$) {$x_1$};
            \draw [dotted, draw=blue, opacity=0.5] (P11.center)--(P11_shift2.center);
            \node (P12_shift2) at ($(P12)+(0,0+\s)$) {\textbullet};
            \draw [dotted, draw=blue, opacity=0.5] (P12.center)--(P12_shift2.center);
            \node (P22_shift2) at ($(P22)+(0,0+\s)$) {\textbullet};
            \draw [dotted, draw=blue, opacity=0.5] (P22.center)--(P22_shift2.center);
            \node (P21_shift2) at ($(P21)+(0,0+\s)$) {\textbullet};
            \draw [dotted, draw=blue, opacity=0.5] (P21.center)--(P21_shift2.center);
            \node  at ($(P13) + (0.3,0)$) {$x_2$};
            \node  at ($(P12) + (0,-0.3)$) {$x_3 = x_1$};
            \node  at ($(P22) + (0,\s) + (0.3,0)$) {$x_4$};
            \node  at ($(P32) + (0.3,0)$) {$x_5$};
            \node  at ($(P22) + (0.6,0)$) {$x_6 = x_4$};
            \node  at ($(P21)+ (0,\s) + (-0.3,0)$) {$x_7$};
            \node[magenta]  at ($(P22)+ (0,\s) + (-0.6,0.3)$) {\small Maximal element};
            \node[magenta]  at ($(0.8,0.35) + (0,\s)$) {$(c_3,2)$};
            \node[orange]  at (1.7,0.2) {$(c_1,1)$};
            \node[teal]  at (1+\x+0.8,\y +0.3) {$(c_2,1)$};
            \draw [magenta,->,very thick] (P11_shift2)--(P12_shift2);
            \draw [magenta,->,very thick] (P12_shift2)-- (P22_shift2);
            \draw [magenta,->,very thick] (P22_shift2)--(P21_shift2);
            \draw [magenta,->,very thick] (P21_shift2)--(P11_shift2);
            \draw [orange,<->,very thick] (P12)--(P13);
            \draw [teal,<->,very thick] (P22)--(P32);

            \draw [very thick, draw=gray, opacity=0.2] (P11.center)--(P12.center);
            \draw [very thick, draw=gray, opacity=0.2] (P12.center)--(P13.center);
            \draw [very thick, draw=gray, opacity=0.2] (P11.center)--(P21.center);
            \draw [very thick, draw=gray, opacity=0.2] (P12.center)--(P22.center);
            \draw [very thick, draw=gray, opacity=0.2] (P13.center)--(P23.center);
            \draw [very thick, draw=gray, opacity=0.2] (P21.center)--(P22.center);
            \draw [very thick, draw=gray, opacity=0.2] (P22.center)--(P23.center);
            \draw [very thick, draw=gray, opacity=0.2] (P21.center)--(P31.center);
            \draw [very thick, draw=gray, opacity=0.2] (P22.center)--(P32.center);
            \draw [very thick, draw=gray, opacity=0.2] (P23.center)--(P33.center);
            \draw [very thick, draw=gray, opacity=0.2] (P31.center)--(P32.center);
            \draw [very thick, draw=gray, opacity=0.2] (P32.center)--(P33.center);
            \
        \end{tikzpicture}
        \caption{Pyramid associated with $\gamma$ based at $x$.\label{fig:pyramid_with_growing_loop}}
    \end{subfigure}
    \caption{Loop  $\gamma$ based at $x$ (Left-hand side) and the corresponding pyramid (Right-hand side) with $x$ contained in the maximal element. Following the oriented edges of the loops, the orange backtrack $c_1$ is popped first and the green backtrack $c_2$ is popped second, whereas the magenta cycle $c_3$ is popped third.
    The labelled heap given by the chronological order is $\{(c_3,2), (c_2, 1), (c_1,1)\}$.    
    In this heap, the two consecutive backtracks are not concurrent and, 
    thus, $(c_2, 1)$ cannot be compared to $(c_1,1)$.
    Nonetheless, the maximum element $(c_3,2)$ is such that $(c_3,2) \succ (c_1,1)$ and $(c_3,2) \succ (c_2,1)$.
    An isomorphic labelled heap is $\{(c_3,4), (c_2, 1), (c_1,1)\}$; see Definition~\ref{def:isom_heap}.
\label{fig:loops_cycles_pyramid}}
\end{figure}
Let $\calA$ be a proper subset of the nodes and let $x\in \overline{\calA}$. 
We define the weight of an unbased cycle as
$
    w_{t,\alpha}([c]) = t^{|c|}\mu_\alpha(c),
$
for $t\in (0,1]$ and with $c$ any representative of $[c]$.
Similarly the weight of a heap reads $w_{t,\alpha}(H) = \prod_{[c]\in H}w_{t,\alpha}([c])$.
Thus, by Lemma~\ref{lem:pyramid_equiv_loop}, the Green generating function of Lemma~\ref{lem:Green_function_loops} can be rewritten as a sum over pyramids
\[
    G_{\alpha}(t,x, x;\calA) = \sum_{P \in \mathcal{P}_\emptyset(\overline{\calA},x)
    }  w_{t,\alpha}(P),
\]
where the sum is over the empty heap (corresponding to the trivial loop at $x$ with $\mu_\alpha$-measure $1$) and any pyramid with pieces in the complement of $\calA$ and $x$ in its maximal piece; see \eqref{eq:set_pyramids_cup_empty_heap}.
Note that this series converges as a consequence of Proposition~\ref{prop:conv_sum_over_heaps}.
Now, let $x_0, \dots, x_{n-1}$ be an ordering of the nodes and define $\mathcal{B}_0 = \emptyset$, $\mathcal{B}_1 = \{x_0\}$, \dots, $\mathcal{B}_{n-1}= \{x_0,x_1, \dots, x_{n-2}\}$.
We calculate the product of Green generating functions as in \eqref{eq:chain_rule}
 \begin{equation*}
    \prod_{\ell=0}^{n-1}G_\alpha(t,x_{\ell},x_{\ell};\mathcal{B}_{\ell}) = \sum_{
        P_0 \in \mathcal{P}_\emptyset(\overline{\mathcal{B}_{0}},x_0)} \dots \sum_{
            P_{n-1} \in \mathcal{P}_\emptyset(\overline{\mathcal{B}_{n-1}},x_{n-1})} w_{t,\alpha}(P_0\circ \dots \circ P_{n-1}).\label{eq:Green_pyramid}
 \end{equation*}
By Proposition~\ref{prop:pyramids_from_heap}, for a fixed node ordering, any heap can be written as a unique superposition of pyramids and empty heaps $P_0\circ \dots \circ P_{n-1}$.
Hence, the sum at the right-hand side of the sum above has exactly one term for each possible heap.
Therefore,
\begin{equation}
    G_\alpha(t,x_{0},x_{0};\mathcal{B}_{0})\dots G_\alpha(t,x_{n-1},x_{n-1};\mathcal{B}_{n-1}) = \textstyle\sum\nolimits_{\text{heap }H}  w_{t,\alpha}(H),\label{eq:Green_heap}
\end{equation}
which is independent of the chosen ordering. 
By taking $t=1$ in \eqref{eq:Green_heap}, we recover the expression of the normalization of the measure \eqref{eq:proba_CRSF_non_det}, which was written ``$\sum_{\calC} \mathbb{P}\left(\mathrm{ pop }\ \calC\right)$'' in \citep[p938]{KK2017}.
Now we use the same strategy as for the derivation of \eqref{eq:ratio_det_Poisson}. 
The moment generating function (MGF) of the sampling time $T$ is then given by
\begin{equation}
    \mathbb{E}[t^T] = t^n \frac{\sum_{\text{heap }H}  w_{t,\alpha}(H)}{\sum_{\text{heap }H}  w_{1,\alpha}(H)}.\label{eq:ratio_det_heap}
\end{equation}
\subsubsection{MGF of $T$ as a ratio of determinants}
At this point, we want to show that expression \eqref{eq:ratio_det_heap} can be written as a ratio that generalizes the ratio of determinants appearing in \eqref{e:MGF_Pi} in the determinantal case.
To do so, we first need to define trivial heaps.
\begin{definition}[Trivial heap of cycles]
    A trivial heap of cycles is a heap whose pieces are all unbased oriented cycles which are not concurrent.
\end{definition}
In other words, a trivial heap of cycles $\calC$ is a set of non-concurrent cycles, and the cardinality of this set is necessarily finite.
Note that the empty heap is a trivial heap.
Now, we leverage a key result relating the generating function of heaps of pieces to trivial heaps, namely Corollary 4.5 in \citep{krattenthaler2006theory}, Lemma 5 in \citep{jerrum2021fundamentals} or Remark 9 in \citep{fredes2023combinatorial},
\begin{equation}
    \sum_{\text{heap }H}  w_{t,\alpha}(H) = \left(\sum_{\substack{\text{trivial heap} \\ \calC}} (-1)^{|\calC|}  w_{t,\alpha}(\calC)\right)^{-1},\label{eq:heap_generating_fct}
\end{equation}
where $|H|$ denotes the total number of pieces of $H$.
The reader can find an example of derivation of the identity \eqref{eq:heap_generating_fct} in the case of heaps with a set of pieces of cardinality $4$ at page 6 of \citep{jerrum2021fundamentals}.
By a direct substitution of \eqref{eq:heap_generating_fct} into \eqref{eq:ratio_det_heap}, we have
\begin{equation}
    \mathbb{E}[t^T] = t^n \frac{\sum_{\calC} (-1)^{|\calC|}  \prod_{[c]\in \calC}  \mu_\alpha(c)}{\sum_{\calC} (-1)^{|\calC|} \prod_{[c]\in \calC} t^{|c|}\mu_\alpha(c)},\label{eq:ratio_det_heap_}
\end{equation}
where $|c|$ denotes here the number of edges in the cycle $c$. 
Noticeably, when $\mu_\alpha$ corresponds to the determinantal case, the expression \eqref{eq:ratio_det_heap_} reduces to the ratio of determinants \eqref{eq:ratio_det}, as a consequence of the expansion of the determinant over permutations and the factorization of permutations over cycles.
\subsubsection{MGF of $T$ as a Poisson process of pyramids}
In order to draw a parallel with the Poisson point process of loops of Section~\ref{sec:Poisson_over_loops}, we rephrase \eqref{eq:ratio_det_heap_} thanks to a combinatorial result of \citet[Proposition 5.10]{viennot2006heaps}, namely 
\[
    \log\sum_{\text{heap }H}  w(H) = \sum_{\substack{
        P\text{ pyramid }}} \frac{1}{|P|} w(P).
\]
for any weight $w(H) = \prod_{c \in H} w(c)$ such that $w(c)\geq 0$, and where $|P|$ is the number of elements of the pyramid $P$.
Consequently, \eqref{eq:ratio_det_heap} reads
\begin{equation}
    \mathbb{E}[t^T] = t^n \exp\left(\sum_{\substack{
        P\text{ pyramid }}}(t^{n_e(P)}-1)\frac{w_{1,\alpha}(P)}{|P|}\right),\label{eq:pyramid_Poisson}
\end{equation}
with $n_e(P)= \sum_{c\in P} |c|$ being the number of edges in the pyramid of cycles.
The same reasoning as in Section~\ref{sec:Poisson_over_loops} yields
\begin{equation*}
    T \stackrel{\text{(law)}}{=} n + \sum_{\substack{P\in \mathcal{P}}}n_e(P)
\end{equation*}
where $\mathcal{P}$ is a Poisson process over pyramids of cycles with intensity $P\mapsto w_{1,\alpha}(P)/|P|$.
\section{\cyclepopping{} as Partial Rejection Sampling\label{sec:prs}}
%
We present a slight variant of \cyclepopping{}, formulated as an instance of the Partial Rejection Sampling (PRS) framework in \citep*[Algorithm 1]{guo2019uniform} and \citep*{jerrum2021fundamentals}.
This approach extends the original stacks-of-cards approach of \citep{Wilson96} for sampling uniform spanning trees. 
It only differs from \cyclepopping{} in the sense that the order in which cycles are popped is here deterministic.
Advantages of this new formulation of Wilson's algorithm as an instance of PRS are that 1) we obtain a compact expression for the law of the number of cycles popped during the algorithm, using heaps of cycles, and 2) it sheds light on the similarity of cycle-rooted spanning forests to other graph-based structures such as sink-free edge orientations \citep{Cohn2002}, which can all be sampled using (extremal) PRS.

Consider a vector $\bv = (v_1, \dots, v_n)$ associated with the \emph{vertices} of $G$ ordered arbitrarily from $1$ to $n$, and where the entry $v_i$ is valued in the finite set $\calN_i$ of neighbours of node $i$, for all $1\leq i \leq n$.
Note that because we only have one outgoing edge from each node, $\bv$ necessarily determines an \emph{oriented} CRSF, which we denote by $\calF(\bv)$.
We now add variables associated to the cycles in the graph, remembering that we include backtracks as licit oriented cycles.
Consider the \emph{oriented cycles} $c_1, \dots, c_p$ of the graph $G$, arbitrarily ordered from $1$ to $p$, and denote by $\bb = (b_1, \dots, b_p)$ a vector with entries in $\{0,1\}$.
We anticipate that $b_\ell = 1$ will be interpreted as ``$c_\ell$ is accepted'' as in Section~\ref{s:formal_def_cycle_popping}.
Now, define $\calD$ as the set of all possible vectors of the type $\bu = (\bv,\bb)$, namely,discrete set
\begin{equation}
        \calD = \calN_1 \times \dots \times \calN_n \times \{0,1\}^p. 
        \label{eq:product_space}
\end{equation}

In the \emph{stacks-of-cards} narrative, we imagine a vector $\bu$ in $\calD$ to be a set of visible face-up cards, one card on top of each node and cycle.
We shall stack several such sets of cards on top of each other, and describe a sampling algorithm as iteratively discarding constraint-violating subsets of visible cards, revealing the cards immediately below.
Formally, we add $p$ constraints on the tuples in $\calD$, one for each oriented cycle.
For $1\leq \ell\leq p$, the $\ell$-th constraint is
\begin{equation}
    \varphi_\ell(\bu) = (c_\ell  \notin \calF(\bv))\text{ or }( c_\ell  \in \calF(\bv) \text{ and } b_\ell = 1 ) \in \{\textsc{true}, \textsc{false}\}. \label{eq:constraint_ell}
\end{equation}
In words, $\varphi_\ell(\bv,\bb) = \textsc{true}$ if, and only if, either the oriented cycle $c_\ell$ is absent or $c_\ell$ is present with $b_\ell = 1$.

Omit temporarily the subscript $\ell$ for simplicity. By examining each constraint $\varphi$ in \eqref{eq:constraint_ell}, we observe that this constraint \emph{effectively} depends on a subset of the 
$$n_{\rm{var}} \triangleq n + p$$
variables, denoted by $\{u_{i_1}, \dots, u_{i_a}\}$ for some integer $1\leq a \leq n_{\rm{var}}$.
Define the scope of $\varphi$ by the tuple of distinct indices $$\scp(\varphi) = (i_1, \dots, i_a),$$ where $i_j\in \{1,\dots, n_{\rm{var}}\}$ for all $1\leq j \leq a$.
Let $\calS = (i_1, \dots, i_a)$ be an arbitrary tuple of distinct indices where each of them is valued in $\{1, \dots, n_{\rm{var}}\}$.
For compactness, we use the following notation
\begin{equation}
    \bu_\calS = (u_{i_1}, \dots, u_{i_a}), \label{eq:subset_tuple}
\end{equation}
to denote the tuple obtained from $\bu$ by selecting the entries belonging to $\calS$.
Thus, by examining \eqref{eq:constraint_ell}, we see that there exists a $\widetilde{\varphi}$ such that 
$$
    \varphi(\bu) = \widetilde{\varphi}(\bu_{\scp(\varphi)}).
$$
Also, in this particular case of the constraint over cycles \eqref{eq:constraint_ell}, for a given constraint $\widetilde{\varphi_\ell}$, there is thus a unique assignement of variables which violates $\widetilde{\varphi_\ell}$.
At this point, we are ready to state the conjunction of constraints
\begin{equation}
    \Phi(\bu) = \widetilde{\varphi_{1}}(\bu_{\scp(\varphi_1)}) \text{ and } \dots \text{ and } \widetilde{\varphi_{p}}(\bu_{\scp(\varphi_p)}), \label{eq:conjunction}
\end{equation}
that the oriented CRSF associated with $\bu$ has to satisfy.
Importantly, as a consequence of the definition of the $\varphi_\ell$'s in \eqref{eq:constraint_ell}, this conjunction of constraints has an important property called \emph{extremality}; see  \citep{jerrum2021fundamentals} and references therein.
\begin{definition}[Extremal formula]
    A formula of the form \eqref{eq:conjunction} is called \emph{extremal} if: for all pairs of constraints $1\leq \ell < k \leq p$, if $\scp(\varphi_\ell)$ has a non-empty intersection with $\scp(\varphi_k)$, then $\varphi_\ell$ and $\varphi_k$ cannot be both false.
\end{definition}
In other words, in the extremal case, when two constraints are both violated, then their scopes are disjoint.
This is clearly the case in \eqref{eq:conjunction} when $\varphi_\ell$ is given in \eqref{eq:constraint_ell} since all cycles in a oriented CRSF are node-disjoint.
We are now ready to describe how the sampling algorithm works.

Define $n_{\rm{var}}$ mutually \emph{independent} random variables -- associated with nodes and oriented cycles -- 
$$
    \mathbf{U} \triangleq (V_1,\dots,V_n,C_1,\dots,C_p) \sim \mathbb{P}_{U_1\times \dots \times U_{n_{\rm{var}}}},
$$
valued in the Cartesian product $\calD \triangleq \calN_1 \times \dots \times \calN_n \times \{0,1\}^p$, as follows.
Let
$$
\mathbb{P}_{V_i}(v) = p_{i v} \text{ for all } v\in \calN_i \text{ (categorical)}\text{ and } \mathbb{P}_{C_\ell} (1) = \alpha(c_\ell) \text{ (Bernoulli)},
$$ where $p_{\cdot}$ and $\alpha(\cdot)$ are the probabilities appearing in the measure $\mu_{\mathrm{oCRSF},\alpha}$ given in \eqref{eq:proba_CRSF_non_det_oriented}.
The law of $\mathbf{U}$ is the \emph{product distribution}
$$
\mathbb{P}_{U_1 \times \dots \times U_{n_{\rm{var}}}}(\mathbf{u}) \triangleq \mathbb{P}_{U_1}(u_1) \dots \mathbb{P}_{U_{n_{\rm{var}} }}(u_{n_{\rm{var}} }),
$$
where the product index in $\mathbb{P}_{U_1 \times \dots \times U_{n_{\rm{var}}}}$ is a reminder of the product nature of the law.
Restating our goal with the new notation, we aim to sample from the probability measure \eqref{eq:proba_CRSF_non_det_oriented}, which reads
\begin{equation}
    \mu(\bu) = Z^{-1}  \mathbb{P}_{U_1 \times \dots \times U_{n_{\rm{var}}}}(\mathbf{u})\cdot 1(\Phi(\bu)= \textsc{true}). \label{eq:target_distribution}
\end{equation}
Note that the normalization constant reads 
$$
    Z = \mathbb{P}_{U_1 \times \dots \times U_{n_{\rm{var}}}} (\Phi(\mathbf{U})=\textsc{true}).
$$
The PRS algorithm, stated as Algorithm~\ref{a:ordered_prs}, samples from \eqref{eq:target_distribution} under a condition on the constraints known as \emph{extremality}.
\begin{algorithm}[b!]
    \begin{itemize}
            \item  Sample $\mathbf{u}$ from $\mathbb{P}_{U_1 \times \dots \times U_{n_{\rm{var}}}}$.
            \item While $\Phi(\mathbf{u})= \textsc{false}$,
            \begin{itemize}
            \item For all $1\leq \ell \leq p$ such that  $\varphi_{\ell}(\bu)= \textsc{false}$, choose $\ell^\star_\sigma$ such that $\sigma(\ell^\star_\sigma)$ is minimal. 
            Resample all variable in $\scp(\varphi_{\ell^\star_\sigma})$ w.r.t.\ the product measure $\prod_{i \in \scp(\varphi_{\ell^\star_\sigma})} \mathbb{P}_{U_i}$.
            \end{itemize}
            \item Output $\mathbf{u}$.
    \end{itemize}
    \caption{Extremal PRS$(\Phi,U_1, \dots, U_{n_{\rm{var}}}; \sigma)$}
    \label{a:ordered_prs}
\end{algorithm}
\subsection{The PRS algorithm in the extremal case}\label{sec:PRS_algo_extremal}
In short, given an ordering $\sigma$ of the $p$ constraints  in the conjunction \eqref{eq:conjunction}, the algorithm resamples -- at each step -- only the variables in the scope of the first violated constraint w.r.t.\ $\sigma$, until all the constraints are satisfied.
%
To formalize this, we define a \emph{random} function which samples \emph{independently} the variables in the scope of a violated constraint.
This sampling function takes as input a tuple $\bu$ and returns a random tuple valued in the cartesian product $\calD$ (see \eqref{eq:product_space}). 
Now, it is convenient to introduce a notation for the indices of the violated constraints. 
To do so, we define 
\begin{equation}
    \bad(\bu) = \{\ell \in \{1, \dots, p\} \text{ such that } \varphi_\ell(\bu)= \textsc{false}\}.
\end{equation}
A direct consequence of the extremality of $\Phi$ is that, the scopes $\scp(\varphi_{\ell})$ for all $\ell\in \bad(\bu)$ are pairwise disjoint.
Now, for all $j\in \{1,\dots ,{n_{\rm{var}}} \}$, we define the $j$-th entry of $\res(\bu,\sigma)$ as follows
\begin{equation}
    \res(\bu, \sigma)_j = \begin{cases}
        U_j \sim \mathbb{P}_{U_j} &\text{if } j \in \scp(\varphi_{\ell^\star_\sigma}) \text{ with } \ell^\star_\sigma = \mathrm{argmin}\{\sigma(\ell):  \ell \in  \bad(\bu)\},\\
        u_j &\text{otherwise},
    \end{cases}\label{eq:resample}
\end{equation}
and so that the $U_j$'s are sampled \emph{independently}.

At this point, it is convenient to define an \emph{auxiliary graph}.
\begin{definition}\label{def:extremal}
    We call the \emph{auxiliary graph} associated with $\res(\cdot,\sigma)$ the following directed graph.
    The nodes of the graph are the points in the discrete set $\mathcal{D}$.
    For any pair of (not necessarily distinct) nodes $\bu,\bu'$, there is a \emph{directed edge} $\bu \to_\sigma \bu^\prime$ if $\bu^\prime$ is in the range of $\res(\bu,\sigma)$.
\end{definition} 
Note that this auxiliary graph can have self-loops and that its nodes are card configurations or equivalently oriented CRSFs of the form $\calF(\bu)$.

Thus, for $\bu^\prime$ such that $\bu \to_\sigma \bu^\prime$, we have
\begin{equation}
    \mathbb{P}(\bu^\prime = \res(\bu,\sigma)) =  \prod_{i\in \scp(\varphi_{\ell^\star_\sigma(\bu)})} \mathbb{P}_{U_i}(u^\prime_i).
\end{equation}
To describe the run of Algorithm~\ref{a:ordered_prs}, we define a Markov chain $(\bE^k_\sigma)_{k\geq 0}$ with state space $\calD$.
Let the initialization be
\begin{equation}
    \bE^0_\sigma \sim \mathbb{P}_{U_1 \times \dots \times U_{n_{\rm{var}}}}. \label{eq:init_MC}
\end{equation}
Then, for all $k\geq 1$, we define 
\begin{equation}
    \bE^k_\sigma = \res(\bE^{k-1}_\sigma, \sigma). \label{eq:step_k_MC}
\end{equation}
The chain $(\bE^k_\sigma)_{k\geq 0}$ stops at $N = \min\{k\geq 0 \text{ such that } \bad(\bE^k_\sigma) = \emptyset\}$, that is, upon reaching a point in the set $\{\bu^\star \text{ such that } \Phi(\bu^\star)=\textsc{true}\}$.
Remark that this Markov chain depends on $\sigma$ -- as indicated by the subscript $_\sigma$ -- due to \eqref{eq:step_k_MC}.
\begin{lemma}
    The random variable $N$ is finite almost surely.
\end{lemma}
\begin{proof}
    The number of resamplings in the classical rejection sampling is finite with probability one. Therefore, since the latter dominates $N$, the random variable $N$ is finite almost surely. 
\end{proof}

\begin{figure}[h]
    \begin{subfigure}[b]{0.45\linewidth}
        \centering
    \begin{tikzpicture}[>={Classical TikZ Rightarrow[]}]

        \node (V1) at (0,-1) {\textbullet};
        \node (V2) at (1*0.5,-1) {\textbullet};
        \node (V3) at (2*0.5,-1) {\textbullet};
        \node (V4) at (3*0.5,-1) {\textbullet};
        \node (V5) at (4*0.5,-1) {\textbullet};
        \node (V6) at (5*0.5,-1) {\textbullet};
        \node (V7) at (6*0.5,-1) {\textbullet};
        \node (V8) at (7*0.5,-1) {\textbullet};

        \node (w1) at (0,-1+0.2) {$v_1$};
        \node (w2) at (1*0.5,-1+0.2) {$v_2$};
        \node (w3) at (2*0.5,-1+0.2) {$v_3$};
        \node (w4) at (3*0.5,-1+0.2) {$v_4$};
        \node (w5) at (4*0.5,-1+0.2) {$v_5$};
        \node (w6) at (5*0.5,-1+0.2) {$v_6$};
        \node (w7) at (6*0.5,-1+0.2) {$v_7$};
        \node (w8) at (7*0.5,-1+0.2) {$v_8$};

        \node[right] at (8*0.5,-1) {\text{graph}};

        \draw [very thick, draw=gray, opacity=0.2] (V1.center)--(V2.center);
        \draw [very thick, draw=gray, opacity=0.2] (V2.center)--(V3.center);        
        \draw [very thick, draw=gray, opacity=0.2] (V3.center)--(V4.center);
        \draw [very thick, draw=gray, opacity=0.2] (V4.center)--(V5.center);
        \draw [very thick, draw=gray, opacity=0.2] (V5.center)--(V6.center);
        \draw [very thick, draw=gray, opacity=0.2] (V6.center)--(V7.center);
        \draw [very thick, draw=gray, opacity=0.2] (V7.center)--(V8.center);

        \node[gray] (P10) at (0,0) {};
        \node[gray] (P20) at (1*0.5,0) {};
        \node[gray] (P30) at (2*0.5,0) {};
        \node[gray] (P40) at (3*0.5,0) {};
        \node[gray] (P50) at (4*0.5,0){} ;
        \node[gray] (P60) at (5*0.5,0) {};
        \node[gray] (P70) at (6*0.5,0) {};
        \node[gray] (P80) at (7*0.5,0) {};

        \node[gray] at (1.5*0.5,0.3) {$0$};
        \node[gray] at (4.5*0.5,0.3) {$0$};
        \node[gray] at (6.5*0.5,0.3) {$1$};

        \node[right] (A0) at (8*0.5,0) {$\mathbf{e}^0$};

        \draw [-latex, thick,draw=gray] (P10.center) -- (P20.center);
        \draw [latex-latex, thick,draw=gray,style=double,blue] (P20.center) -- (P30.center);
        \draw [-latex, thick,draw=gray] (P40.center) -- (P30.center);
        \draw [latex-latex, thick,draw=gray,style=double] (P50.center) -- (P60.center);
        \draw [latex-latex, thick,draw=gray,style=double] (P70.center) -- (P80.center);

    \node (P11) at (0,1) {};
    \node (P21) at (1*0.5,1) {};
    \node (P31) at (2*0.5,1) {};
    \node (P41) at (3*0.5,1) {};
    \node (P51) at (4*0.5,1){} ;
    \node (P61) at (5*0.5,1) {};
    \node (P71) at (6*0.5,1) {};
    \node (P81) at (7*0.5,1) {};
    \node[right]  (A1) at (8*0.5,1) {$\mathbf{e}^1$};

    \draw [latex-latex, thick,draw=gray,style=double,magenta] (P11.center) -- (P21.center);
    \draw [latex-latex, thick,draw=gray,style=double] (P31.center) -- (P41.center);
    \draw [latex-latex, thick,draw=gray,style=double] (P51.center) -- (P61.center);
    \draw [latex-latex, thick,draw=gray,style=double] (P71.center) -- (P81.center);

    \node[gray] at (0.5*0.5,1+0.3) {$0$};
    \node[gray] at (2.5*0.5,1+0.3) {$0$};
    \node[gray] at (4.5*0.5,1+0.3) {$0$};
    \node[gray] at (6.5*0.5,1+0.3) {$1$};

    \node (P12) at (0,2) {};
    \node (P22) at (1*0.5,2) {};
    \node (P32) at (2*0.5,2) {};
    \node (P42) at (3*0.5,2) {};
    \node (P52) at (4*0.5,2){} ;
    \node (P62) at (5*0.5,2) {};
    \node (P72) at (6*0.5,2) {};
    \node (P82) at (7*0.5,2) {};
    \node[right] (A2) at (8*0.5,2) {$\mathbf{e}^2$};

    \draw [-latex, thick,draw=gray] (P12.center) -- (P22.center);
    \draw [-latex, thick,draw=gray] (P22.center) -- (P32.center);
    \draw [latex-latex, thick,draw=gray,style=double,teal] (P32.center) -- (P42.center);
    \draw [latex-latex, thick,draw=gray,style=double] (P52.center) -- (P62.center);
    \draw [latex-latex, thick,draw=gray,style=double] (P72.center) -- (P82.center);

    \node[gray] at (2.5*0.5,1+1+0.3) {$0$};
    \node[gray] at (4.5*0.5,1+1+0.3) {$0$};
    \node[gray] at (6.5*0.5,1+1+0.3) {$1$};

        \node (P13) at (0,3) {};
        \node (P23) at (1*0.5,3) {};
        \node (P33) at (2*0.5,3) {};
        \node (P43) at (3*0.5,3) {};
        \node (P53) at (4*0.5,3){} ;
        \node (P63) at (5*0.5,3) {};
        \node (P73) at (6*0.5,3) {};
        \node (P83) at (7*0.5,3) {};
        \node[right] (A3) at (8*0.5,3) {$\mathbf{e}^3$};

        \draw [-latex, thick,draw=gray] (P13.center) -- (P23.center);
        \draw [-latex, thick,draw=gray] (P23.center) -- (P33.center);
        \draw [-latex, thick,draw=gray] (P33.center) -- (P43.center);
        \draw [-latex, thick,draw=gray] (P43.center) -- (P53.center);
        \draw [latex-latex, thick,draw=gray,style=double,red] (P53.center) -- (P63.center);
        \draw [latex-latex, thick,draw=gray,style=double] (P73.center) -- (P83.center);

        \node[gray] at (4.5*0.5,3+0.3) {$0$};
        \node[gray] at (6.5*0.5,3+0.3) {$1$};

        \node (P14) at (0,4) {};
        \node (P24) at (1*0.5,4) {};
        \node (P34) at (2*0.5,4) {};
        \node (P44) at (3*0.5,4) {};
        \node (P54) at (4*0.5,4){} ;
        \node (P64) at (5*0.5,4) {};
        \node (P74) at (6*0.5,4) {};
        \node (P84) at (7*0.5,4) {};
        \node[right] (A4) at (8*0.5,4) {$\mathbf{e}^4$};

        \draw [-latex, thick,draw=gray] (P14.center) -- (P24.center);
        \draw [-latex, thick,draw=gray] (P24.center) -- (P34.center);
        \draw [-latex, thick,draw=gray] (P34.center) -- (P44.center);
        \draw [-latex, thick,draw=gray] (P44.center) -- (P54.center);
        \draw [-latex, thick,draw=gray] (P54.center) -- (P64.center);
        \draw [-latex, thick,draw=gray] (P64.center) -- (P74.center);

        \draw [latex-latex, thick,draw=gray,style=double] (P74.center) -- (P84.center);

        \draw [->, thick] (A0) -- (A1);
        \draw [->, thick] (A1) -- (A2);
        \draw [->, thick] (A2) -- (A3);
        \draw [->, thick] (A3) -- (A4);

        \node[gray] at (6.5*0.5,4+0.3) {$1$};

    \end{tikzpicture}
    \caption{Path in auxilliary graph. }
    \end{subfigure}
    \hfill
    \begin{subfigure}[b]{0.45\linewidth}
        \centering
        \begin{tikzpicture}[>={Classical TikZ Rightarrow[]}]

            \node (C1) at (0,-1.) {};
            \node (C2) at (1*0.5,-1.) {};
            \node (C3) at (2*0.5,-1.) {};
            \node (C4) at (3*0.5,-1.) {};
            \node (C5) at (4*0.5,-1.){} ;
            \node (C6) at (5*0.5,-1.) {};
            \node (C7) at (6*0.5,-1.) {};
            \node (C8) at (7*0.5,-1.) {};

            \node[gray]  at (0+0.25,-1.+0.2) {$\varphi_1$};
            \node[gray]  at (1*0.5+0.25,-1.+0.2) {$\varphi_2$};
            \node[gray]  at (2*0.5+0.25,-1.+0.2) {$\varphi_3$};
            \node[gray]  at (3*0.5+0.25,-1.+0.2) {$\varphi_4$};
            \node[gray]  at (4*0.5+0.25,-1.+0.2){$\varphi_5$} ;
            \node[gray]  at (5*0.5+0.25,-1.+0.2) {$\varphi_6$};
            \node[gray]  at (6*0.5+0.25,-1.+0.2) {$\varphi_7$};

            \node[gray]  at (0+0.25,-1.-0.2) {$1$};
            \node[gray]  at (1*0.5+0.25,-1.-0.2) {$2$};
            \node[gray]  at (2*0.5+0.25,-1.-0.2) {$3$};
            \node[gray]  at (3*0.5+0.25,-1.-0.2) {$4$};
            \node[gray]  at (4*0.5+0.25,-1.-0.2){$5$} ;
            \node[gray]  at (5*0.5+0.25,-1.-0.2) {$6$};
            \node[gray]  at (6*0.5+0.25,-1.-0.2) {$7$};
            
            \node[right,gray] at (7*0.5,0.15) {\text{heap $\mathcal{H}$}};

            \draw [latex-latex, thick,draw=gray,style=double] (C1.center) -- (C2.center);
            \draw [latex-latex, thick,draw=gray,style=double] (C2.center) -- (C3.center);
            \draw [latex-latex, thick,draw=gray,style=double] (C3.center) -- (C4.center);
            \draw [latex-latex, thick,draw=gray,style=double] (C4.center) -- (C5.center);
            \draw [latex-latex, thick,draw=gray,style=double] (C5.center) -- (C6.center);
            \draw [latex-latex, thick,draw=gray,style=double] (C6.center) -- (C7.center);
            \draw [latex-latex, thick,draw=gray,style=double] (C7.center) -- (C8.center);
            \node[right,gray] at (7*0.5-0.1,-1.) {\text{ $2$-cycles}};
            \node[right,gray] at (7*0.5-0.1,-1.2) {\text{ $\sigma$}};

            \node[gray] (A2) at (1*0.5,0) {\textbullet};
            \node[gray] (A3) at (2*0.5,0) {\textbullet};
            \draw [latex-latex, thick,draw=gray,style=double,blue] (A2.center) -- (A3.center);

            \node[gray] (B1) at (0*0.5,0.2) {\textbullet};
            \node[gray] (B2) at (1*0.5,0.2) {\textbullet};
            \draw [latex-latex, thick,draw=gray,style=double,magenta] (B1.center) -- (B2.center);

            \node[gray] (C3) at (2*0.5,0.2*1) {\textbullet};
            \node[gray] (C4) at (3*0.5,0.2*1) {\textbullet};
            \draw [latex-latex, thick,draw=gray,style=double,teal] (C3.center) -- (C4.center);

            \node[gray] (D5) at (4*0.5,0) {\textbullet};
            \node[gray] (D6) at (5*0.5,0) {\textbullet};
            \draw [latex-latex, thick,draw=gray,style=double,red] (D5.center) -- (D6.center);
    
            \draw[gray] (-0.05,-0.2) -- (3.5,-0.2) -- (3.5,0.5) -- (-0.05,0.5) -- (-0.05,-0.2);
        \end{tikzpicture}
        \caption{Corresponding heap.\label{fig:pyramid_with_growing_loop}}
    \end{subfigure}
    \caption{Example of a path in the auxilliary graph $\be^0\to_{\sigma} \be^1$, \dots, $\be^{3} \to _\sigma \be^{4}$ in the case of the line graph with $8$ nodes.
    By inspection, $\Phi(\be^{4})$ is true.
    Here, there are $7$ oriented cycles which are ordered by $\sigma$ from left to right (increasing order).
    On the left-hand side, each cycle is depicted with the associated value $b\in \{0,1\}$; see \eqref{eq:constraint_ell}, namely when a $2$-cycle comes with a value $1$, it satisfies the constraint. 
    On the right-hand side, we depict the heap of resampled scopes. Colors help visualizing the construction of the heap during the resampling process. 
\label{fig:split}}
\end{figure}
\subsection{Defining a heap by splitting a path in an auxiliary graph}
Let $\be^0\to_{\sigma} \be^1$, \dots, $\be^{k^\star - 1} \to _\sigma \be^{k^\star}$ be a deterministic path in the auxiliary graph with the endpoint being a valid assignement, i.e., $\be^{k^\star} = \bu^\star$ with $\bad(\bu^\star) = \emptyset$.
Thus, by definition, $\Phi(\bu^\star)$ is true.
It is now convenient to reorganise this path by highlighting the sequence of resampled scopes.
Hence, we introduce a splitting function which splits a deterministic path with a valid assignement as endpoint into a heap of violated scopes \citep{viennot2006heaps,krattenthaler2006theory} and a valid assignment, namely,
\begin{equation}
    \spt(\be^0, \dots, \be^{k^\star} = \bu^\star) = (\calH, \bu^\star).\label{eq:split}
\end{equation}
Here, the heap $\calH$ is a combinatorial object which collects the resampled scopes in the run of Algorithm~\ref{a:ordered_prs}; we refer the reader back to Section~\ref{sec:heaps} for more details about heaps. 
The basic idea is that $\calH$ is a heap whose pieces are of the form
\begin{equation}
    \bm{\pi}=\be_{\scp(\varphi_{\ell})}, \label{eq:piece_of_heap}
\end{equation}
for some $1\leq \ell\leq p$ and for some $ \be \in \calD$  such that  $\widetilde{\varphi_{\ell}}(\be_{\scp(\varphi_{\ell})})=\textsc{false}$.
A piece associated with the variables $\be_{\scp(\varphi)}$ (violating $\widetilde{\varphi}$)  can be put on top of a piece associated with $\be_{\scp(\varphi^\prime)}$ (violating $\widetilde{\varphi}^\prime$) in a heap only if $\scp(\varphi) \cap \scp(\varphi^\prime) \neq \emptyset$.
The main difference with Definition~\ref{def:concurrent} is that our definition is here a bit more general since it can be used for the case of any extremal formula.

The splitting function \eqref{eq:split} is simply constructed as follows. 
Let us start by an empty heap and construct the heap step by step.
Consider the $j$-th step of the stopped path for $0 \leq j \leq k^\star -1$. 
By examining the entries of $\be^{j}$ which differ from $\be^{j+1}$, we can determine the piece $\bm{\pi}^j$ to be added to the heap $\calH$. 
The last configuration of the path $\be^{k^\star}$ is equal to the valid configuration $\bu^\star$. 
We refer the reader to Figure~\ref{fig:split} for an illustration.

Before going further, it is natural to associate a weight to a heap $\calH$ given by the product of the probabilities of each of its pieces (resampled scopes) as follows:
\begin{equation}
    w(\calH) = \prod_{\substack{\bm{\pi}\in \calH }} \prod_{i}\mathbb{P}_{U_{s_i}}(\pi_i).    \label{eq:heap_weight}
\end{equation}
The following proposition is a straightforward consequence of the definition of a heap and of extremality as given in Definition~\ref{def:extremal}.
\begin{proposition}\label{prop:bijective}    
    Consider the extremal formula \eqref{eq:conjunction}.
    Let $\calH$ be a heap of violated constraints and denote by $k^\star$ the number of pieces in $\calH$. Let $\bu^\star\in\calD$ satisfying $\Phi(\bu^\star) = \textsc{true}$.
    Let $\sigma$ be an ordering of the constraints and let $(\bE^{k}_\sigma)_{k\geq 0}$ be a Markov chain initialized by \eqref{eq:init_MC} and with the transition rule \eqref{eq:step_k_MC} using $\res(\cdot,\sigma)$.
    Denote by $\mathbb{P}$ the law of this Markov chain. 
    There exists a unique path $\be^0\to_{\sigma} \be^1$, \dots, $\be^{k^\star - 1} \to _\sigma \bu^\star$ in the auxiliary graph induced by $\sigma$ such that  
    \begin{itemize}
        \item[(i)]$\spt(\be^0, \dots, \bu^\star) = (\calH, \bu^\star)$ and,
        \item[(ii)]  $\mathbb{P}(\bE^0_\sigma = \be^0 \text{ and }  \dots  \text{ and }\bE^{k^\star}_\sigma = \bu^\star)  = w(\calH) \prod_{i=1}^{n_{\rm var}} \mathbb{P}_{U_i}(u^\star_{i})\neq 0.$
    \end{itemize}
\end{proposition}
\begin{proof} 
    To begin, we prove (i) by reconstructing the path in the auxiliary graph defined in  Section~\ref{sec:PRS_algo_extremal}.

    First, we recover the order in which violated scopes are resampled by doing a \emph{forward pass} on $\calH$.
    Consider the minimal pieces of $\calH^0 \triangleq \calH$, i.e., the pieces which have no other piece below them.
    By extremality (Definition~\ref{def:extremal}), each piece is uniquely associated to the scope of a violated constraint. 
    Thus, all minimal pieces correspond to violated constraints which have disjoint scopes.
    Among the minimal pieces of $\calH^0$, we take the piece corresponding with the violated constraint having the smallest $\sigma$-value. 
    This determines the first scope resampled which is of the form \eqref{eq:piece_of_heap}.
    For convenience, we name this unique piece 
    $$
        \bm{\pi}^1 = \bff^{0}_{\scp(\varphi_{q^1})} \text{ for some index } 1 \leq q^1 \leq p \text{ and some tuple } \bff^{0}\in \calD.
    $$
    Now, remove this piece from $\calH^0$, namely, define $\calH^1$ so that $\calH^0 = \bm{\pi}^1 \circ \calH^1$.  
    Let us pauze a moment for a remark.
    By extremality, removing $\bm{\pi}^1$ does not affect the presence of the other minimal pieces of $\calH^0$ which therefore  are also minimal in $\calH^1$.
    However, $\calH^1$ can contain other minimal pieces which were not minimal in $\calH^0$.
    This closes our remark.
    We then proceed similarly with $\calH^1$.
    Thus, we select the piece of $\calH^1$ with minimal $\sigma$-value among its minimum pieces.
    This determines a unique piece that we name 
    $$
    \bm{\pi}^2 = \bff^{1}_{\scp(\varphi_{q^2})} \text{ for some index } 1 \leq q^2 \leq p \text{ and some tuple } \bff^{1}\in \calD.
    $$
    We next repeat this operation until the remaining heap is empty.    
    Note that the number of pieces in $\calH$ is necessarily equal to $k^\star$.
    At this point, we have retrieved the order in which the scopes were resampled, namely $\varphi_{q^1}$, \dots, $\varphi_{q^{k^\star}}$ and the values taken by the variables in the scope of the violated constraints. Note that the indices $q^1, \dots , q^{k^\star}$ may not all be distinct.

    Second, to determine the sequence $\be^0, \be^1, \dots, \be^{k^\star} = \bu^\star$, we go over the pieces in the reverse order, i.e., from the top to the bottom of the heap: $\bm{\pi}^{k^\star}, \dots, \bm{\pi}^1$.
    Naturally, we initialize $\be^{k^\star} = \bu^\star$.
    Consider $\bm{\pi}^{t}$ for some $1\leq t \leq k^\star$.
    Define $\be^{t}$ such that for all $1\leq i \leq n_{\rm var}$
    \begin{equation}
        e^{t}_i =
        \begin{cases}
            p^t_i & \text{ if } i \in \scp(\varphi_{q^{t}}),\\
            e^{t+1}_{i} & \text{ if } i \in \{1,\dots,n_{\rm var}\}\setminus \scp(\varphi_{q_{t}}).
        \end{cases}   
        \label{eq:backward_substitution}
    \end{equation}
    Importantly, the \emph{backward substitution} \eqref{eq:backward_substitution} guarantees that $\be^{t}$ violates $\varphi_{q^{t}}$ and the constraints which were also violated by $\be^{t+1}$, but that no other extra constraint is violated, as a consequence of extremality\footnote{Here is a key difference with what might happen in a non-extremal formula such as in the case of sampling independent sets in graphs in its first formulation of \citep[Section 4.1]{jerrum2021fundamentals} as conjunction of constraints on the edges.
    In that case, the \emph{backward substitution} \eqref{eq:backward_substitution} could make other constraints false}.
    This makes sure that the probability that $\res(\be^{t}, \sigma) = \be^{t+1}$ is not zero.
    This defines $\be^0, \be^1, \dots, \be^{k^\star}$ unambiguously, proving $(i)$.

    We now prove (ii).  Let the path $\be^0\to_{\sigma} \be^1$, \dots, $\be^{k^\star - 1} \to _\sigma \bu^\star$ be as determined in (i).
    By the Markov property, we have
    \begin{align}
         &\mathbb{P}\left(\bE^0_\sigma = \be^0 ,  \dots , \bE^{k^\star}_\sigma = \bu^\star \right) \nonumber\\
        &  = \mathbb{P}(\bE^0_\sigma = \be^0) \mathbb{P}(\bE^1_\sigma = \be^1|\bE^0_\sigma = \be^0)  \dots \mathbb{P}(\bE^{k^\star}_\sigma = \bu^\star|\bE^{k^\star - 1}_\sigma = \be^{k^\star - 1}). \label{eq:Markov_on_run_prs} 
    \end{align}
    At this point, we recall two facts. 
    First, the initialization \eqref{eq:init_MC} gives
    $
        \mathbb{P}(\bE^0_\sigma = \be^0) = \prod_{i=1}^{n_{\rm var}} \mathbb{P}_{U_i}(e^0_{i}).
    $
    Second, the transition of the Markov chain \eqref{eq:step_k_MC} yields
    \begin{align*}
        \mathbb{P}(\bE^l_\sigma = \be^l|\bE^{l-1}_\sigma = \be^{l-1})& =     \mathbb{P}(\be^l = \res(\be^{l-1},\sigma))\\
        & = \prod_{i\in \scp(\varphi_{\ell^\star_\sigma(\be^{l-1})})} \mathbb{P}_{U_i}(e^{l}_i),
    \end{align*}
    for $l\geq 1$.
    By grouping the factors on the right-hand side of \eqref{eq:Markov_on_run_prs}  belonging to the valid assignment $\bu_{\star}$ and the factors corresponding to resampled scopes in the weight of the heap, we find 
    \begin{align}
        \mathbb{P}\left(\bE^0_\sigma = \be^0 ,  \dots , \bE^{k^\star}_\sigma = \bu^\star \right)= w(\calH)  \prod_{i=1}^{n_{\rm var}} \mathbb{P}_{U_i}(u^\star_{i}),\label{eq:coupling_heap_valid}
    \end{align}
    where $w(\calH)$ is defined in \eqref{eq:heap_weight}.
\end{proof}

\subsection{Putting it all together}
Let $\calH$ be a heap of resampled scopes containing $k^\star$ pieces and let $\bu^\star$ be such that $\Phi(\bu^\star)= \textsc{true}$.
Now, we calculate
\begin{align}
   \mathbb{P}\left(\spt(\bE^0_\sigma, \dots, \bE^{N}_\sigma) = (\calH, \bu^\star)\right).\label{eq:proba_split_preimage}
\end{align}
Since $\spt$ simply adds up resampled scopes and since one scope is resampled at the time, it is clear that \eqref{eq:proba_split_preimage} vanishes on the complement of the event $\{N = k^\star\}$.
Thus, recalling that the Markov chain realizes paths in the auxiliary graph associated with $\sigma$, we have
\begin{align*}
    &\mathbb{P}\left(\spt(\bE^0_\sigma, \dots, \bE^{k^\star}_\sigma) = (\calH, \bu^\star)\right) \\
    &=  \smashoperator{\sum_{\substack{\bff^0, \dots, \bff^{k^\star-1}\text{ s.t. }\\ \bff^0\to_{\sigma} \bff^1, \dots, \bff^{k^\star - 1} \to _\sigma \bu^\star}}}\mathbb{P}\left(\bE^0_\sigma = \bff^0 , \dots, \bE^{k^\star}_\sigma = \bu^\star \right)\cdot 1\left(\spt(\bff^0, \dots,\bu^\star) = (\calH, \bu^\star)\right)\\
    &= \mathbb{P}\left(\bE^0_\sigma = \be^0 , \dots,\bE^{k^\star - 1}_\sigma = \be^{k^\star - 1},  \bE^{k^\star}_\sigma = \bu^\star \right)= w(\calH)  \prod_{i=1}^{n_{\rm var}} \mathbb{P}_{U_i}(u^\star_{i}),
\end{align*}
where we used Proposition~\ref{prop:bijective} at the next-to-last step, i.e., there exists a unique path $\be^0\to_{\sigma} \be^1$, \dots, $\be^{k^\star - 1} \to _\sigma \bu^\star$ such that  $\spt(\be^0, \dots, \bu^\star) = (\calH, \bu^\star)$.
Now, by marginalizing over the heaps in \eqref{eq:coupling_heap_valid}, 
we find 
\begin{align}
    \sum_{\calH} \mathbb{P}\left(\spt(\bE^0_\sigma, \bE^1_\sigma, \dots, \bE^{N}_\sigma) = (\calH, \bu^\star)\right) = \frac{1}{Z}  \prod_{i=1}^{n_{\rm var}} \mathbb{P}_{U_i}(u^\star_{i}) \triangleq \mu(\bu^\star),\label{eq:marginalization_H}
\end{align}
where
\begin{equation}
    Z^{-1} = \sum\nolimits_{\calH} w(\calH)\label{eq:normalization}
\end{equation}
is the normalization.

Note that the sum over all heap weights, namely the series $\sum_{\calH} w(\calH)$, converges.
This is a consequence of the fact that $\calD$ is a finite set and that that $w(\calH)\in (0,1)$ for all $\calH$; this is proved in the same way as Proposition~\ref{prop:conv_sum_over_heaps}.
The identity \eqref{eq:marginalization_H} shows that: for any ordering of the constraints $\sigma$, running the Markov chain \eqref{eq:step_k_MC} until all the constraints are satisfied indeed samples correctly from the product measure conditioned to $\Phi$ being true.
This shows the correctness of Algorithm~\ref{a:ordered_prs}.

\subsection{Number of resampled bad assignements}
We complete our discussion of Algorithm~\ref{a:ordered_prs} by restating a result of \citep{jerrum2021fundamentals} about the law of the number of resampling of a scope (namely, here a cycle).
A celebrate formula of \citet{viennot2006heaps} allows for rewriting the series \eqref{eq:normalization} as a sum over trivial heaps; see also \citep[Lemma 5]{jerrum2021fundamentals}.
A trivial heap is a heap of pieces for which all elements are maximal, i.e., in our context, it is a set of violated constraints with disjoint scopes.
We have
\begin{equation}
    Z = \sum_{\calT \text{ trivial }} (-1)^{|\calT|} w(\calT),\label{eq:alternating_formula}
\end{equation}
where $|\calT|$ is the number of elements of the (trivial) heap.
As pointed out in \citep[Eq. (4)]{jerrum2021fundamentals}, $Z$ can be interpreted as the probability that $\bU\sim \mathbb{P}_{U_1} \times \dots \times \mathbb{P}_{U_{n_{\rm{var}}}}$ satisfies the constraint $\Phi$, i.e.,
\[
    Z = \mathbb{P}_{U_1 \times \dots \times U_{n_{\rm{var}}}}(\Phi(\bU)= \textsc{true}),
\]
as announced in \eqref{eq:target_distribution}.
The alternating sum \eqref{eq:alternating_formula} is simply proved by applying the inclusion-exclusion principle on the event 
$$
    \{\varphi_{1}(\bU_{\scp(\varphi_1)}) = \textsc{true} \text{ and } \dots \text{ and } \varphi_{p}(\bU_{\scp(\varphi_p)}) = \textsc{true}\}.
$$
%
The moment generating function of the number of resampled scopes can be obtained as a simple consequence of \eqref{eq:coupling_heap_valid}.
\begin{theorem}\label{thm:moment_gen_fct}
    Consider the set of pieces  of the form \eqref{eq:piece_of_heap}.
    Let $\mathbf{t} = (t_1, \dots, t_p)$ be a tuple parameters in $(0,1)^{p}$.
    Denote by $\sharp_j(\calH)$ the number of times the scope of the $j$-th violated constraint appears in the heap $\calH$ for all $j=1,\dots, p$.
    Consider $(\calH,U^\star)$ to be a random pair with law \eqref{eq:coupling_heap_valid}.
    We have
    \begin{equation}
                \mathbb{E}_{(\calH,U^\star)}[t_1^{\sharp_1(\calH)} \dots t_k^{\sharp_k(\calH)}] = \frac{\sum_{\calT \text{ trivial }} (-1)^{|\calT|} w(\calT)}{\sum_{\calT \text{ trivial }} (-1)^{|\calT|}w(\calT)\prod_{i: c_i\in \calT} t_i } .\label{eq:MGF_nbs}
    \end{equation}
\end{theorem}
\begin{proof}
    By using \eqref{eq:coupling_heap_valid}, we have
    \begin{align*}
        \mathbb{E}_{(\calH,U^\star)}[t_1^{\sharp_1(\calH)} \dots t_k^{\sharp_k(\calH)}] &= \sum_{\calH} t_1^{\sharp_1(\calH)} \dots t_k^{\sharp_k(\calH)} w(\calH)\sum_{\bu^\star: \Phi(\bu^\star)}\prod_{i=1}^{n_{\rm{var}}} \mathbb{P}_{U_i}(u^\star_{i})\\ 
        &= \frac{\sum_{\calH} t_1^{\sharp_1(\calH)} \dots t_k^{\sharp_k(\calH)}w(\calH)}{\sum_{\calH}w(\calH)}.        
    \end{align*}
    The result follows by using the formula for the normalization \eqref{eq:normalization}.
\end{proof}
Formally, the MGF of the number of times a cycle is resampled in \eqref{eq:MGF_nbs} resembles the formula \eqref{eq:ratio_det_heap_} for the MGF of the number of \emph{jumps} of the random walker in Algorithm~\ref{a:cyclepopping}. 
We emphasize that \eqref{eq:MGF_nbs} -- which concerns \emph{cycles} -- yields a variant of \eqref{eq:ratio_det_heap_} if $t_i$ is replaced by $t^{|c_i|}$ where $|c_i|$ is the number of edges in the cycle $c_i$.

We conclude with the following result, taken from \citep{jerrum2021fundamentals}, which gives an elegant formula for the expected number of scopes which are resampled in the run of Algorithm~\ref{a:ordered_prs}.
\begin{corollary}
    \label{c:law_of_number_of_resampled_scopes}
    Consider the setting of Theorem~\ref{thm:moment_gen_fct}.
    The expected number of times the oriented cycle $c_1$ is resampled is
    \begin{align*}
            \mathbb{E}_{(\calH,\bU^\star)}[\sharp_1(\calH)] &= \frac{\sum_{\calT \text{ trivial }: c_1 \in \calT} (-1)^{|\calT|-1}w(\calT)}{\sum_{\calT \text{ trivial }} (-1)^{|\calT|} w(\calT)}\\
             &= \frac{\mathbb{P}_{U_1 \times \dots \times U_n}(\bad(\bU)= \{1\})}{\mathbb{P}_{U_1 \times \dots \times U_n}(\Phi(\bU) = \textsc{true})},
    \end{align*}
    where $\bU\sim\mathbb{P}_{ U_1 \times \dots \times U_{n_{\rm{var}}}}$ and $\{\bad(\bU)= \{1\}\} = \{\text{only }\varphi_1(\bU) = \textsc{false}\}$.
    Furthermore, the expected number of resampled scopes is 
    $$
    \mathbb{E}_{(\calH,\bU^\star)}\left[\sum_{\ell=1}^p \sharp_\ell(\calH)\right]  = \frac{\mathbb{P}_{U_1 \times \dots \times U_{n_{\rm{var}}}}(|\bad(\bU)| = 1)}{\mathbb{P}_{U_1 \times \dots \times U_{n_{\rm{var}}}}(\Phi(\bU) = \textsc{true})},
    $$
    where $\{|\bad(\bU)| = 1\} = \{\exists\text{ a unique $\ell$ s.t. $\varphi_\ell(\bU) = \textsc{false}$}\}$.
\end{corollary}
\begin{proof}
    First, we obtain the first part of the result by differentiating \eqref{eq:MGF_nbs} w.r.t. $t_1$.
    The second part of the result is obtained by using the inclusion-exclusion formula.
\end{proof}
From a computational point of view, Corollary~\ref{c:law_of_number_of_resampled_scopes} yields  another measure of the complexity of running Wilson's algorithm.
Note that while the (PRS) variant of Wilson's algorithm studied in this section requires a deterministic order $\sigma$ on cycles, the law in Corollary~\ref{c:law_of_number_of_resampled_scopes} does not depend on that order, so that the result applies to the original, random-order Wilson's algorithm described as Algorithm~\ref{a:cyclepopping}.

\section{Numerical simulations\label{sec:numerics}}

The mean and variance of the number of steps to complete \cyclepopping{} can simply be estimated in a simple $\Uone$-connection graph.
We define here the following random graph.
The Erd\H{o}s-Rényi unicycle model, denoted by  $\mathrm{ER_u}(n,p,\eta)$, is a $\Uone$-connection graph of $n$ nodes where there is an edge $e=uv$ with probability $p$, independently from other edges, and where only one noisy edge sampled uniformly comes an angle $\vartheta = \eta \pi/2$ whereas all the other edges are endowed with a vanishing angle. 
Hence, for $0 < \eta \leq 1$, all cycles containing this noisy edge will have an holonomy satisfying $\cos \theta(c) \geq 0$, whereas all cycles not containing the noisy edge have a unit holonomy.
Assumption~\ref{ass:non-trivial} and Assumption~\ref{ass:weak} are then guaranteed by construction. 
Also, any CRSF sample from \eqref{eq:proba_CRSF} in this random graph will be connected with probability one.

Thus, we sample a $\mathrm{ER_u}(n,p,\eta)$ graph with $n=100$, $p=0.8$ and for each $\eta \in \{0.5, 0.6, \dots, 1\}$. For each random graph, we sample $1000$ CRSFs thanks to \cyclepopping{}.
\footnote{
    \url{https://github.com/For-a-few-DPPs-more/MagneticLaplacianSparsifier.jl/tree/counting_steps}
}
In Figure~\ref{fig:number_of_steps_CRSF}, we observe that the empirical estimates of $\E[T]$ and $\var[T]$ given in Proposition~\ref{prop:mean_var_using_cumulants}, are good approximations.
Notice that we represent the standard deviation as an error bar in Figure~\ref{fig:number_of_steps_CRSF}.
The mean and variance of $T$ naturally decrease as the noise parameter increases since a large $\eta$ intuitively promotes a large value of the least eigenvalue of $\mathsf{\Delta}$.
\begin{figure}
    \centering
    \begin{subfigure}[b]{0.3\linewidth}
        \centering
        \includegraphics[scale=0.25]{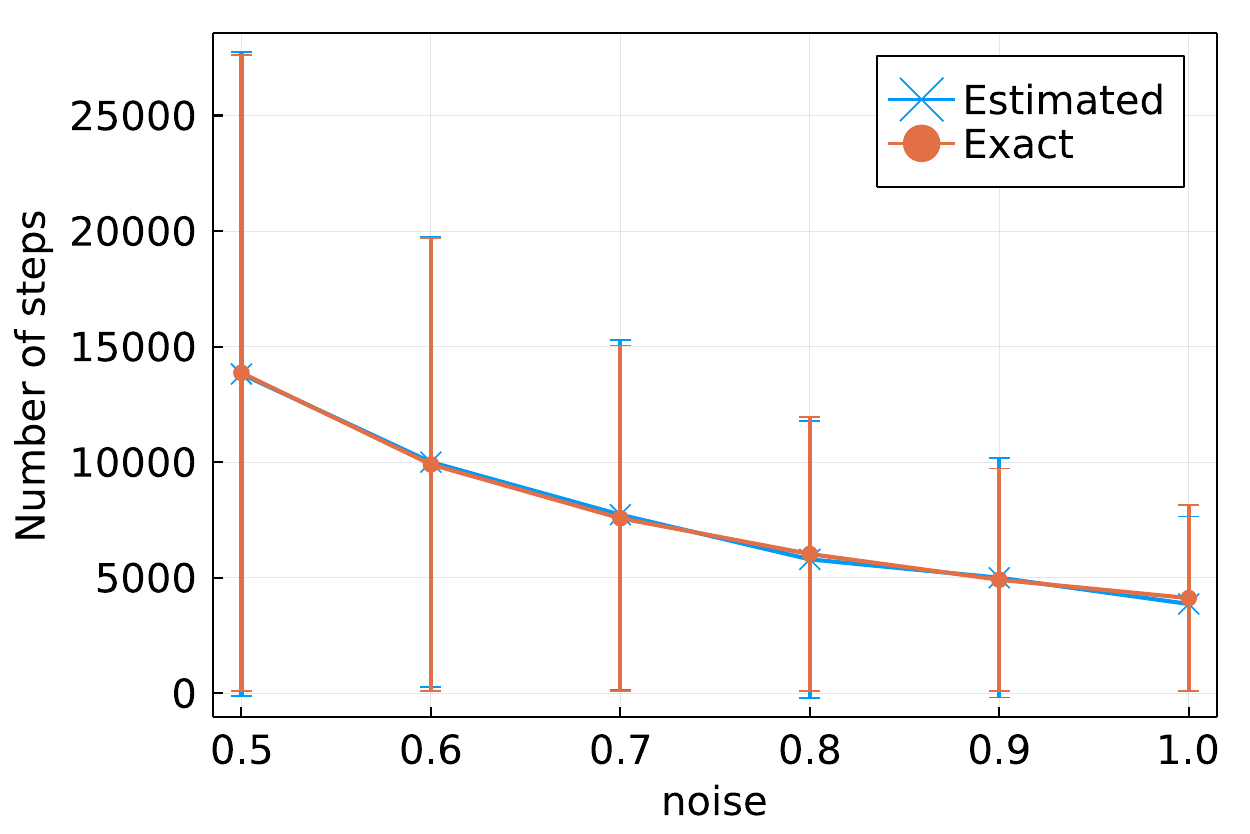}
        \caption{CRSF sampling.
        \label{fig:number_of_steps_CRSF}}
    \end{subfigure}
    \hfill
    \begin{subfigure}[b]{0.3\linewidth}
        \centering
        \includegraphics[scale=0.25] {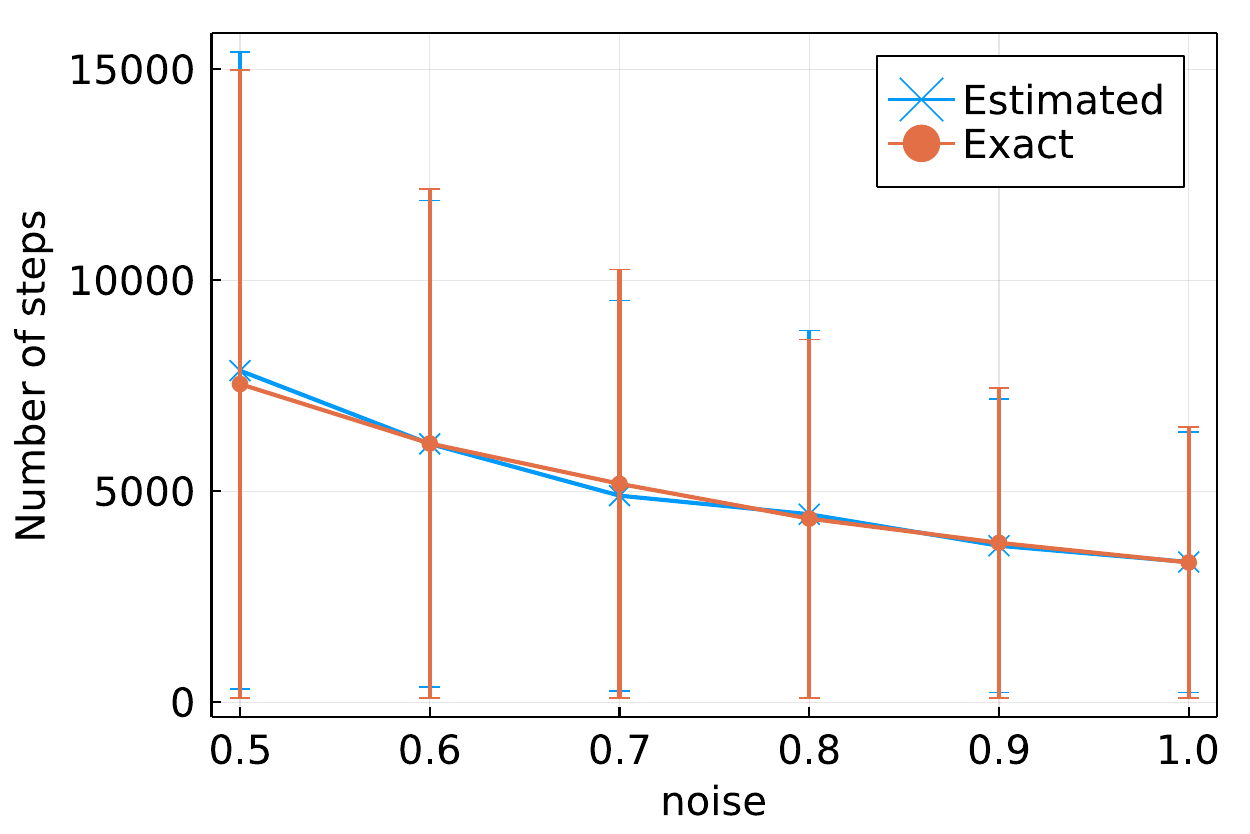}
        \caption{MTSF sampling.\label{fig:number_of_steps_MTSF}}
    \end{subfigure}
    \hspace{2cm}
    \caption{
        Number of steps ($T$) to complete \cyclepopping{}$_q$ for sampling CRSFs (left, $q =0$) and MTSFs (right, $q =5\cdot 10^{-3}$) w.r.t.\ \eqref{eq:MTSF_measure}   in an $\mathrm{ER_u}(n,p,\eta)$ random $\Uone$-connection graph with $n=100$, $p=0.8$ as a function of the noise parameter $\eta$.
        We report $\E[T]$ and $\sqrt{\var[T]}$ (error bar) in red whereas their estimators over $1000$ Monte Carlo runs are displayed in blue.
        The red and blue curves almost overlap.
    }
\end{figure}
The same simulation is repeated by sampling $1000$ MTSFs with \cyclepopping{}$_q$ for $q = 5 \cdot 10^{-3}$ and the comparison of $\E[T]$ and $\var[T]$ with their empirical estimates is given in Figure~\ref{fig:number_of_steps_MTSF}.
\section{Conclusion}
In this paper, we gave a pedestrian proof of the correctness of \cyclepopping{} for sampling measures of the form \eqref{eq:proba_CRSF_non_det}, using only elementary random walk arguments, as \cite{Marchal99} did for Wilson's original algorithm. 
From a computer scientist's point of view, on top of being easy to follow and adapt to more sophisticated variants of the algorithm, the proof yields the distribution of the running time of the algorithm. 
From a more probabilistic point of view, the construction sheds light on other point processes built while running the algorithm, and in particular an intriguing coupling already discussed by \cite{LeJan11} between cycle-rooted spanning forests and a Poisson point process of popped loops, and a related Poisson point process on Viennot pyramids.

\section*{Acknowledgments}
  We are grateful to Simon Barthelmé, Guillaume Gautier and David Dereudre for instructive discussions about Partial Rejection Sampling, and to Martin Rouault for commenting on an early version of the manuscript. 
  Special thanks are given to Adrien Kassel for his comments on a preliminary version of the paper, as well as drawing our attention to the references \citep{PiTa18,GJ2021}.

  We acknowledge support from ERC grant BLACKJACK ERC-2019-STG-851866 and ANR AI chair
  BACCARAT ANR-20-CHIA-0002. 

\appendix
  \section*{Proof of Proposition~\ref{prop_incidence_prob_cycles}  \label{a:proofs}}
  Recall that we want to prove the expression of the incidence probability of cycles in a random CRSF $\calT$ distributed according to the determinantal measure \eqref{eq:proba_CRSF}.
  We begin with the following well-known variant of the matrix tree theorem: let $\calA$ be a subset of nodes ($\calA \subset \calV$), we have
  \begin{equation}
      \det (\Delta_{\overline{\calA}}) = \sum_{\calU\in \calS(\overline{\calA})} \prod_{e\in \calU} w_e \prod_{c\in \calU} (2-2 \cos \theta(c))\label{general_matrix_tree_thm}
  \end{equation}
  where $\calS(\overline{\calA})$ is the set of forests spanning $\overline{\calA}$ whose connected components are either cycle-rooted trees in $\overline{\calA}$ or trees in $\overline{\calA}$ with exactly one root node in $\calA$.    
  This result can be proved by decomposing the left-hand side with the help of the magnetic incidence matrix (see \eqref{eq:mag_incidence}) and Cauchy-Binet identity, as in \citep{kenyon2011}.
  Now, let $\calC$ be a set of non-concurrent cycles.
  By using the expression \eqref{eq:proba_CRSF}, a simple decomposition gives
  \begin{align*}
      \mathbb{P}(\calC \subseteq \cycles(\calU)) &= \frac{1}{\det (\Delta)} \sum_{\calU : \calC \subseteq \calU}\prod_{e\in \calU} w_e \prod_{c\in \calU} (2-2 \cos \theta(c))\\
      & = \frac{\nu(\calC)}{\det (\Delta)} \sum_{{\calU^\prime} \in \calS(\overline{\nodes(\calC)})}\prod_{e\in \calU^\prime} w_e \prod_{c\in \calU^\prime} (2-2 \cos \theta(c)),
  \end{align*}
  where $ \nu(\calC) = \prod_{e\in \calC} w_e \prod_{c \in \calC}(2-2 \cos \theta(c))$.
  At this point, we leverage \eqref{general_matrix_tree_thm} to yield the following compact expression
  \begin{align*}
      \mathbb{P}(\calC \subseteq \cycles(\calF)) = \nu(\calC)\frac{\det (\Delta_{\overline{\nodes(\calC)}})}{\det (\Delta)} = \nu(\calC) \det (\Delta^{-1})_{\nodes(\calC)}.
  \end{align*}
  The last equality is obtained by using the customary formula for the inverse of a matrix with four blocks. 
  This ends the proofs.

\bibliography{References,stats}       

\end{document}